\documentclass[11pt]{article}
\usepackage{}
\usepackage[round]{natbib}
\usepackage{mathrsfs}
\usepackage{caption}
\usepackage{amsfonts}
\usepackage{tabularray}
\usepackage{amsmath}
\usepackage{amssymb}
\usepackage{dsfont,footmisc}
\usepackage[all]{xy}
\usepackage{graphicx}
\usepackage{CJK}
\usepackage{color}
\usepackage{subcaption}
\usepackage[singlespacing]{setspace}
\usepackage{multirow}
\usepackage{rotating}
\usepackage{comment}
\usepackage{tabularx}
\usepackage[font=small]{caption}%
\usepackage[colorlinks,linkcolor=red,anchorcolor=blue,citecolor=blue]%
{hyperref}
\usepackage{booktabs}    

\usepackage{threeparttable} 
\usepackage[T1]{fontenc}

\usepackage{verbatim}
\def\d{\mathrm{d}}

\allowdisplaybreaks

\newcommand{\VaR}{\mathrm{VaR}}

\newcommand{\TVaR}{\mathrm{TVaR}}
\newcommand{\RV}{\mathcal{RV}}

\newcommand{\SES}{\mathrm{SES}}

\newcommand{\ES}{\mathrm{ES}}
\newcommand{\CTE}{\mathrm{CTE}}
\newcommand{\e}{\mathrm{e}}
\newcommand{\MES}{\mathrm{MES}}
\newcommand{\ICE}{\mathrm{ICE}}
\newcommand{\CE}{\mathrm{CE}}

\newcommand{\SICE}{\mathrm{SICE}}
\newcommand{\E}{\mathbb{E}}
\newcommand{\R}{\mathbb{R}}

\newcommand{\N}{\mathbb{N}}
\newcommand{\p}{\mathbb{P}}

\newcommand{\id}{\mathds{1}}

\renewcommand{\(}{\left(}
\renewcommand{\)}{\right)}
\renewcommand{\[}{\left[}
\renewcommand{\]}{\right]}

\newtheorem{theorem}{Theorem}[section]

\newtheorem{corollary}{Corollary}[section]

\newtheorem{definition}{Definition}[section]
\newtheorem{example}{Example}[section]

\newtheorem{lemma}{Lemma}[section]

\newtheorem{remark}{Remark}[section]

\newenvironment{proof}[1][Proof]{\noindent \textbf{#1.} }{\  \rule{0.5em}{0.5em}}
\newtheorem{proposition}{Proposition}[section]

\def\nn{\nonumber}

\begin{document}
\begin{CJK}{GBK}{kai}
\baselineskip=16pt

\title{
Value-at-Risk- and Expectile-based Systemic Risk Measures and Second-order Asymptotics: With Applications to Diversification
}

\author{Bingzhen Geng\thanks{\scriptsize School of Big Data and Statistics, Anhui University, 
China. Email: \texttt{gengbz@ahu.edu.cn}}\quad\quad\quad\quad Yang Liu\thanks{\scriptsize Corresponding Author. School of Science and Engineering, The Chinese University of Hong Kong, Shenzhen, 
China. Email: {\texttt{yangliu16@cuhk.edu.cn}}}
 \quad\quad\quad\quad Yimiao Zhao\thanks{\scriptsize Department of Statistics and Actuarial Science, University of Waterloo, 
Canada. Email: \texttt{l249zhao@uwaterloo.ca}}\\[10pt]
}
\date{}
\maketitle
\begin{abstract}

Systemic risk measures play a crucial role in analyzing individual losses conditional on extreme system-wide disasters. In this paper, we provide a unified asymptotic treatment for systemic risk measures. First, we classify them into two families of Value-at-Risk- (VaR-) and expectile-based systemic risk measures. While VaR-based risk measures have been extensively studied, in the latter family, we propose two new systemic risk measures named the Individual Conditional Expectile (ICE) and the Systemic Individual Conditional Expectile (SICE), as alternatives to Marginal Expected Shortfall (MES) and Systemic Expected Shortfall (SES). Second, to characterize general mutually dependent and heavy-tailed risks, we consider a multivariate loss system following a multivariate Sarmanov distribution with common marginal distributions exhibiting second-order regular variation. Third, within this setting, we provide second-order asymptotic results for both families of systemic risk measures. These results extend standard first-order asymptotics and allow for more accurate tail approximations. Through numerical and analytical examples, we demonstrate the superiority of second-order asymptotics in accurately assessing systemic risk. We further conduct a comprehensive comparison between expectile-based and VaR-based systemic risk measures. The results indicate that expectile-based measures often 
yield higher asymptotic accuracy than VaR-based ones, emphasizing the former's potential advantages in reporting extreme events and tail risk. As a financial application, we use the asymptotic treatment to discuss the diversification benefits associated with various risk measures. Finally, we extend and obtain the second-order asymptotic formulas for generalized-quantile-based systemic risk measures with power functions.

	~
	
\end{abstract}

\setcounter{equation}{0}
\section{Introduction}



Financial risks are potential events that can negatively affect the stability of an individual financial institution, a specific financial market, or even the global economy. {Analyzing extreme events that align with these adverse scenarios constitutes a crucial part of risk management.} To study these extreme events, the extreme value theory (EVT) is a useful framework. EVT offers a contemporary collection of statistical tools and techniques that can be used to address various questions related to risk assessment and management in finance.  Financial risks are typically divided into different categories depending on their characteristics. Market risk, credit risk, and operational risk are the primary risk groups that have been extensively studied using quantitative assessment methods and regulated by authorities. After the global financial crisis in 2008-2009, the concept of systemic risk, which refers to the risk of multiple financial institutions failing together and causing widespread impact, has gained significant attention from regulators and researchers in the field. Various measures related to the systemic risk have been proposed in the literature, including CoVaR by \cite{adrian2016covar}, Systemic Expected Shortfall ($\SES$) and Marginal Expected Shortfall ($\MES$) by \cite{acharya2017measuring} and \cite{chen2022asymptotic},  
scenario-based risk measures by \cite{wang2021scenario}, conditional distortion risk measures by \cite{dhaene2022systemic}, generalized risk measures by \cite{fadina2024framework} and others.


Following recent studies of systemic risks in banking, finance and insurance, we quantify $\SES$ and $\MES$ in a general context of quantitative risk management. 
 {{We denote the individual risks by $X_1, \ldots, X_n$, and the aggregate risk by $S_n=\sum_{i=1}^n X_i$.}} For the sum $S_n$, for $p \in (0,1)$, according to the Euler principle (see \cite{dhaene2012optimal} or \cite{acharya2017measuring}), the
risk allocated to line $m\in \{1,\ldots,n\}$ is defined by
\begin{align*}
    \MES_{p,m}(S_n):=\E\left[X_m|S_n>\VaR_p(S_n)\right],
\end{align*}
or
\begin{align*}
    \SES_{p,m}(S_n):=\E\left[\(X_m-\VaR_p(X_m)\)_+|S_n>\VaR_p(S_n)\right],
\end{align*}
 {where Value-at-Risk ($\VaR$) is defined as the loss quantile of $X$ (the distribution is denoted by $F_X$):} 
\begin{align*}
    \VaR_{p}(X) := F_X^{\leftarrow}(p)=\inf\{t\in\R:F_X(t)\geq p\},
\end{align*}
and $F_X^{\leftarrow}$ is the generalized inverse of the distribution function.  Another popular risk measure is Expected Shortfall ($\ES$):
\begin{align*}
    \ES_p(X) := \frac{1}{1-p}\int_p^1\VaR_t(X)\d t,
\end{align*}
which is the average value on the tail beyond $\VaR_{p}$. The $\ES$ (sometimes called Tail-Value-at-Risk ($\TVaR$)) is a coherent risk measure in the sense of \cite{artzner1999coherent}. If $F_X$ is continuous, $\ES$ coincides with the Conditional Tail Expectation ($\text{CTE}$), which represents the conditional expected loss given that the loss exceeds its $\VaR$:
\begin{align*}
    \ES_{p}(X) = \text{CTE}_{p}(X) := \E\left[X|X>\VaR_{p}(X)\right].
\end{align*}

Though enjoying several merits, $\VaR$ and $\ES$ have some drawbacks. Specifically, $\VaR$ does not possess subadditivity, which excludes $\VaR$ from the good class of coherent risk measures. $\ES$ does not satisfy the {elicitability}, which is a property recently arousing interest in the field of risk management. Here, 
a risk measure is said to be elicitable if it can be defined as the minimizer of a suitable expected loss function. 
Elicitability is important in backtesting a risk measure as it provides a natural methodology for this purpose. Meaningful point forecasts and forecast performance comparisons then become possible for elicitable risk measures; see \cite{ziegel2016coherence} and \cite{HLY2025}. 

Following \cite{newey1987asymmetric}, the expectile $\e_p(X)$ with a level $p\in(0, 1)$ of the variable $X$ can be defined as the
minimizer of an expected piecewise quadratic loss function:
\begin{align*}
\e_p(X) =\arg\min_{\theta\in\mathbb{R}}\left\{p \E\left[\((X-\theta)_+\)^2\right]+(1-p)\E\left[\((X-\theta)_-\)^2\right]\right\},
\end{align*}
where $x_+:=\max(x, 0)$ and $x_-:=\min(x, 0)$. The minimization problem is well defined if $X\in\mathrm{L}^2$ (i.e., $\E|X|^2<\infty$). The related first-order necessary condition of optimality can be written in several ways, one of them being
\begin{align}
\e_p(X)-\E[X]=\frac{2p-1}{1-p}\E\left[(X-\e_p(X))_+\right].
\label{eq:1.1}
\end{align}
This equation has a unique solution for all $X\in\mathrm{L}^1$. Therefore, expectiles of a distribution function $F$ with a finite first-order moment are well defined, and we assume that $\E|X|<\infty$ throughout. Expectiles summarize the distribution function in much the same way as the quantiles; see \cite{gneiting2011making}. The expectile and $\VaR$ are elicitable while $\ES$ is not. Actually, the expectile $\e_p$ with $p \geq \frac{1}{2}$ is the only risk measure which is elicitable, law-invariant and coherent; see \cite{bellini2015elicitable}.  As a result, the expectile is suggested (see \cite{emmer2015best}) as a potential alternative to both $\VaR$ and $\ES$. The study on expectiles also becomes increasingly popular in the econometrics literature; see, for example, \cite{de2009quantiles} and \cite{kuan2009assessing}. 

In this paper, we provide a unified asymptotic treatment for the systemic risk measures. The treatment has three steps. The first step is that the systemic risk measures are classified into two representative families: VaR- and expectile-based measures. From the definitions, $\VaR_p$ and $\e_p$ exhibit distinct mathematical properties. As we can see below, both of them lead to many systemic risk measures; e.g., $\CTE_p$ (or $\ES_p$), $\MES_{p,m}$ and $\SES_{p,m}$ are categorized within the family of VaR-based risk measures. As a result, $\VaR_p$ and $\e_p$ can serve as building blocks in assessing systemic risk, and form essential foundations for further study. In particular,
\cite{taylor2008estimating} introduced an expectile-based alternative of ES, known as the Conditional Expectile ($\CE_p$):
\begin{align*}
    \CE_p(X):=\E\left[X | X > \e_p(X)\right],
\end{align*}
 {where $\CE_p$ represents the conditional expected loss given that the loss exceeds the expectile $\e_p$ and $\CE_p$ satisfies positive homogeneity and translation invariance (see \cite{daouia2020tail}). } 

To evaluate the allocation of each individual agent to the systemic risk, we further propose two new expectile-based systemic risk measures on the sum variable. We call them the \textit{Individual Conditional Expectile} ($\ICE$) and the \textit{Systemic Individual Conditional Expectile} ($\SICE$):
\begin{align*}
    \ICE_{p,m}(S_n):=\E\left[X_m|S_n>\e_p(S_n)\right],
\end{align*}
and
\begin{align*}
    \SICE_{p,m}(S_n):=\E\left[\(X_m-\e_p(X_m)\)_+|S_n>\e_p(S_n)\right].
\end{align*}
 {$\ICE$ and $\SICE$ adopt an expectile-based viewpoint to capture an individual agent's risk profile in the event of a system-wide catastrophe.} For comparison, the $\VaR$ estimation knows only whether an observation is below or above the predictor. It would be inaccurate to measure an extreme risk based on only the frequency of tail losses but not on their values. The expectile makes more efficient use of the available data since it optimizes the discrepancy between the observations and the predictor. Particularly, $\ICE$ represents the potential losses an individual would suffer conditional on the tail of the system's loss distribution. $\SICE$ is an improved version of $\ICE$ and reveals the individual's excess loss to his/her expectile $\e_p(X_m)$ conditional on the systemic catastrophe.




In the second step, to characterize a general dependence and heavy-tailed risks, we assume that the risks $X_1,\ldots,X_n$ are dependent on each other through a multivariate Sarmanov distribution.  {The Sarmanov family is a very broad and flexible dependence structure: by selecting different kernel functions, it is capable of generating a wide range of joint models, including many commonly used asymmetric and heterogeneous dependence patterns; e.g., the commonly used Farlie-Gumbel-Morgenstern (FGM) copula. Meanwhile, the FGM copula has some drawbacks in terms of correlation coefficients; see \cite{yang2013extremes} and Section \ref{sec:sarmanov}. In particular, the Sarmanov distribution is flexible in combining different types of marginals, making it suitable for modeling various risks. The advantage of using the Sarmanov distribution lies in its ability to capture dependencies and  {to facilitate} the evaluation of joint probabilities. More importantly, the analytical method developed in the paper is not specific to the Sarmanov family. The same second-order asymptotic approximation technique can be adapted to other dependence structures; for instance, by following our derivations, one can obtain second-order asymptotics of the Clayton copula. }


\begin{table}[htbp]%
  \centering
  \setlength{\tabcolsep}{2.0mm}{
  \begin{tabular}{ |c|p{14em}|p{16em}| }%
  \hline
  Systemic risk measure &   First-order asymptotic &Second-order asymptotic   \\
\hline                     
   $\VaR_p(S_n)$& \cite{bingham1989regular,barbe2006tail,embrechts2009additivity} and so on. & \cite{degen2010risk,mao2015risk}; Theorem \ref{the:3.1} of our paper\\
\hline    
 $\CTE_p(S_n)$& \cite{alink2005analysis,chen2012extreme,kley2020modelling} 
 and so on.&\cite{mao2013second,lv2013asymptotics}; Theorem \ref{the:3.1} of our paper\\
 \hline 
$\MES_{p,m}(S_n)$  &\cite{asimit2011asymptotics,joe2011tail,jaune2022asymptotic,geng2025asymptotics} and so on.&\cite{hua2011second}; Theorem \ref{the:3.2} of our paper\\
 \hline 
  $\SES_{p,m}(S_n)$&\cite{chen2022asymptotic,geng2025asymptotics} and so on.  &Theorem \ref{the:3.2} of our paper\\
   \hline                     
   $\e_{p}(S_n)$& \cite{bellini2014generalized,bellini2017risk}& \cite{mao2015asymptotic,mao2015risk}; Theorem \ref{the:4.1} of our paper\\
  \hline 
 $\CE_{p}(S_n)$ &\cite{dhaene2022systemic} & Theorem \ref{the:4.1} of our paper\\
 \hline 
 $\ICE_{p,m}(S_n)$& \cite{emmer2015best,tadese2020relative}&Theorem \ref{the:4.2} of our paper\\
 \hline 
$\SICE_{p,m}(S_n)$ & Theorem \ref{the:4.2} of our paper & Theorem \ref{the:4.2} of our paper\\
 \hline 
Generalized-quantile-based & Theorem \ref{the:7.1} of our paper & Theorem \ref{the:7.1} of our paper\\
  \hline 
  \end{tabular}}
   \caption{Contribution of our paper compared to the literature. Here VaR-based systemic risk measures include $\VaR$, $\CTE$, $\MES$ and $\SES$, while expectile-based systemic risk measures include $\e$, $\CE$, $\ICE$ and $\SICE$.}
  \label{tab:1}
\end{table}

In the third step, we obtain the second-order asymptotics of the two families of systemic risk measures, as summarized in Table \ref{tab:1}. Here, we make our most theoretical contributions on asymptotic approximations. First, we investigate the second-order asymptotic formulas of the tail probability of $S_n$ under a multivariate Sarmanov distribution (Proposition \ref{pro:sum} below), which 
generalizes the results of \cite{mao2013second} and Theorem 4.4 of \cite{mao2015risk}.  Second, we study second-order asymptotic formulas of $\VaR_p(S_n)$ and $\CTE_p(S_n)$ (Theorem \ref{the:3.1} below) as well as those of $\MES_{p,m}(S_n)$ and $\SES_{p,m}(S_n)$ (Theorem \ref{the:3.2} below). 
Third, we use a different method to obtain the second-order asymptotic formula of the expectile, which extends Theorem 3.1 of \cite{mao2015risk} (Proposition \ref{pro:e} below).  
Fourth, we consider the second-order asymptotic formulas of $\e_{p}(S_n)$ and $\CE_{p}(S_n)$ (Theorem \ref{the:4.1} below) as well as $\ICE_{p,m}(S_n)$ and $\SICE_{p,m}(S_n)$ (Theorem \ref{the:4.2} below). Lastly, we perform two examples to demonstrate different risk measures based on $\VaR$ and expectile, where we use the Monte Carlo method to conduct the numerical simulation. Numerical and analytical examples illustrate that our second-order asymptotics provide an accurate estimate and behave much better than the first-order asymptotics. Further, we conduct a comprehensive comparison between these two families of systemic risk measures. We find that expectile-based systemic risk measures often yield higher asymptotic accuracy than VaR-based systemic risk measures, and the former has a potential advantage in reporting extreme events and revealing the tail risk. In addition, we analytically compare the asymptotic behavior of these two families in a Pareto-like setting with different tail indices (Proposition \ref{pro:5.1} below). 


 {As a financial application, we use our asymptotic treatment to discuss the diversification effect.} The idea of portfolio diversification dates back to the celebrated Markowitz mean-variance model, revealing the importance of mitigating risks in the investment. After that, diversification becomes a crucial topic in banking and insurance for risk management, integral to regulatory frameworks like Basel II and Solvency II. 
A larger body of literature propose quantitative ways for diversification. For instance, \cite{chen2022ordering} delved into comparing diversification advantages under the worst-case VaR and ES in the context of dependence uncertainty; see \cite{cui2021diversification} for more results. 

Among them, the diversification benefit denotes the capital preserved by considering all risks in a portfolio collectively, rather than addressing each risk in isolation; see \cite{mcneil2015quantitative} and Section \ref{sec:app} for a detailed definition. Based on our asymptotic treatment, we obtain the second-order asymptotics for the diversification benefit based on different risk measures, including VaR, CTE, $\e$ and $\CE$. Furthermore, we extend our results of expectiles to generalized quantiles (\cite{bellini2014generalized}) with power functions and obtain the corresponding second-order asymptotic formulas for generalized-quantile-based systemic risk measures in Section \ref{sec:7} (Theorem \ref{the:7.1} below). 




Finally, we discuss our results compared with the literature. In the field of risk management and capital allocation, it is necessary to determine how to allocate the acquired economic capital among different risks. In this case, the solvency capital has already been calculated using risk aggregation techniques; 
see \cite{blanchet2020convolution} for a recent treatment of risk aggregation. It is of increasing interest to study the selection of appropriate models for multivariate risk factors, such as the choice of the dependence structure model and the marginal distributions. Some recent contributions include the use of the FGM distribution (\cite{yang2013extremes}, \cite{chen2014ruin}, etc.), the Sarmanov distribution (\cite{qu2013approximations}, \cite{yang2013tail}, etc.), and multivariate regular variation (MRV) (\cite{embrechts2009multivariate}, \cite{asimit2011asymptotics}, etc.). They usually aim to obtain the first-order asymptotics of some risk measures. On the contrary, our results provide a series of second-order asymptotics, which are much more accurate than the first-order asymptotics; this will be shown in the tables and figures later. This fact is also studied in  \cite{degen2010risk}, \cite{mao2012second} and \cite{mao2013second}, but they obtained the second-order asymptotics for independent and identically distributed (iid) random variables (rvs) with second-order regularly varying tails. This independence assumption is too restrictive for practical problems. \cite{mao2015risk} studied the case that $X_1,\ldots,X_n$ are dependent on each other through a multivariate FGM distribution, and derived second-order approximations of the risk concentrations of Value-at-Risk and expectile. 
However, our study contributes to the advancement of systemic risk measurement by introducing novel measures, developing a modeling framework and providing enhanced asymptotic tools for risk assessment. Technically, the proposed two wide families of systemic risk measures can include the risk measures in the literature of \cite{mao2015risk} and others. In particular, we provide rigorous and necessary lemmas for asymptotic treatment (Proposition \ref{pro:sum} and Lemmas \ref{lem:2rv}-\ref{lem:beta}) and offer a simpler proof for the key results. 



The rest of the paper is organized as follows. In Section \ref{pre}, we introduce the definitions of $\RV$ and $2\RV$ and discuss the $n$-dimensional Sarmanov distribution.  In Sections \ref{VaR} and \ref{expectile}, we obtain several second-order asymptotics of $\VaR$- and expectile-based systemic risk measures and present examples to explain the main results. In Section \ref{sim}, we give concrete examples to numerically illustrate these risk measures. Further, 
we apply the above asymptotic treatment to discuss financial diversification in Section \ref{sec:app}. 
In Section \ref{sec:7}, we obtain the second-order asymptotics of generalized-quantile-based systemic risk measures.
Section \ref{sec:conc} concludes the paper. The Appendix 
provides details for the proofs.

\setcounter{equation}{0}
\section{Preliminaries}\label{pre}

In this section, we first review the definitions and basic properties of regular variation ($\RV$) and the second-order regular variation (2$\RV$). $2\RV$ is a concept that is a generalization of regular variation ($\RV$), which has various applications in areas such as applied probability, statistics, risk management, telecommunication networks and so on. The idea of $2\RV$ was initially proposed to investigate the rate of convergence of the extreme order statistics in EVT; see \cite{de1996generalized} and \cite{de2006extreme}.  Next, we introduce the $n$-dimensional Sarmanov distribution. Its applications in many insurance contexts show its flexible structure when modeling the dependence
between multivariate risks given the marginal distributions; see \cite{qu2013approximations}, \cite{abdallah2016sarmanov}, \cite{ratovomirija2016mixed} and so on.

\subsection{Regular variation}

\begin{definition}(Regular variation)\\
A measurable function $f:\mathbb{R}\rightarrow\mathbb{R}_+$ is said to be regularly varying at $t_0\in[-\infty,\infty]$ with an index $\alpha\in\mathbb{R}$ and denoted by $f\in \RV_{\alpha}^{t_0}$, if
\begin{align*}
    \lim_{t\rightarrow t_0}\frac{f(tx)}{f(t)}=x^{\alpha}, ~~\mbox{for all} ~x>0.
\end{align*}
If $t_0=\infty$, we write $\RV_{\alpha}^{t_0}=\RV_{\alpha}$. In addition, if $\alpha=0$, then $f$ is said to be slowly varying at infinity.
\end{definition}

\begin{definition}(Second-order regular variation)\\
A measurable function $f:\mathbb{R}\rightarrow\mathbb{R}_+$ is said to be second-order regularly varying with a first-order index $\alpha\in\mathbb{R}$ and a second-order index $\beta\leq 0$ and denoted by $f\in 2\mathcal{RV}_{\alpha,\beta}^{t_0}$, if there exists some eventually positive or negative measurable function $A(\cdot)$ with $A(t)\rightarrow 0$ as $t\rightarrow t_0$ such that
\begin{align*}
    \lim_{t\rightarrow t_0}\frac{\frac{f(tx)}{f(t)}-x^{\alpha}}{A(t)}=x^{\alpha}\frac{x^{\beta}-1}{\beta}:=H_{\alpha,\beta}(x),~~\mbox{for all} ~x>0.
\end{align*}
Here $H_{\alpha,\beta}(x)$ is  $x^{\alpha}\log x$ if $\beta=0$, and $A(\cdot)$  is called an auxiliary function of $f$. It is worth noting that the auxiliary function $A(\cdot)$ is of $\RV_{\beta}$; see Theorem 2.3.3 of
\cite{de2006extreme}. If $t_0=\infty$, we write $2\mathcal{RV}_{\alpha,\beta}^{t_0}=2\mathcal{RV}_{\alpha,\beta}$.
\end{definition}

Let $X_1, X_2, ..., X_n$ denote the financial losses, which are identically distributed random variables with a distribution function $F$. Here $\overline{F}(\cdot)$ means the survival function $\overline{F}(x) = 1 - F(x)$ and $F^{\leftarrow}(\cdot)$ means the generalized inverse function $F^{\leftarrow}(y) =\inf\{x\in \mathbb{R}: F(x) \geq y\}$. 
The tail quantile function associated with the distribution function $F$ is denoted by $U_F(\cdot)=\(1/\overline{F}\)^{\leftarrow}(\cdot)=F^{\leftarrow}\(1-1/\cdot\)$. Note that $\overline{F}(\cdot)\in \RV_{-\alpha}$ for all $\alpha\in \R$ is equivalent to $U_F(\cdot)\in \RV_{1/\alpha}$  for all $\alpha\in \R$ (see Corollary 1.2.10 of \cite{de2006extreme}). Furthermore, if $\overline{F}(\cdot)\in 2\RV_{-\alpha,\beta}$ with $\alpha>0, \beta\leq 0$ and an auxiliary function $A(\cdot)$, by Theorem 2.3.9 of \cite{de2006extreme}, one can easily check that $U_F(\cdot)\in 2\RV_{1/\alpha,\beta/\alpha}$ with an auxiliary function $\alpha^{-2}A\circ U_F(\cdot)$. Generally, the equality $F\(F^{\leftarrow} (p)\) =p$ does not hold true. It can be shown that if $\overline{F}(\cdot)\in \RV_{-\alpha}$
 with $\alpha>0$, then $F\(F^{\leftarrow}(p)\)\sim p$ as $p\uparrow1$. If we further assume that $\overline{F}(\cdot)\in 2\RV_{-\alpha,\beta}$
with an auxiliary function $A(\cdot)$, then
\begin{align}\label{eq:2.1}
U_F(1/\overline{F}(t))=t\(1+o(A(t))\) ~\mbox{as}~ t\rightarrow \infty~ \text{ and }~F\(F^{\leftarrow}(p)\)= p\(1+o\(A\(F^{\leftarrow}(p)\)\)\)~ \mbox{as} ~p\uparrow 1;
\end{align}
see \cite{mao2015risk} and Exercise 2.11 of \cite{de2006extreme}.


\subsection{Multivariate Sarmanov distribution}\label{sec:sarmanov}


The Sarmanov distribution is widely adopted in different fields. It was originally introduced by \cite{sarmanov1966generalized} in the bivariate case and then extended by \cite{lee1996properties} and \cite{kotz2004continuous} in the multivariate case:
\begin{align}\label{eq:sarmanov}
\p\(X_1\in \d x_1,\ldots,X_n\in \d x_n\)=\(1+\sum_{1\leq i<j\leq n}a_{ij}\phi_{i}(x_i)\phi_j(x_j)\)\prod_{k=1}^{n}\d F(x_k),
\end{align}
 where $F$ is the corresponding marginal distribution of $X_i$, $i = 1, \dots, n$. Particularly, the parameters $a_{ij}$ are real numbers and the kernels $\phi_i$ are functions satisfying
 $$\E[\phi_i(X_i)]=0, ~i=1,\ldots,n, $$
 and
$$1+\sum_{1\leq i< j\leq n}a_{ij}\phi_{i}(x_i)\phi_j(x_j) \geq 0 ~\mbox{ for all}~x_i\in D_{X_i} \text{ and } x_j\in D_{X_j},$$
 where $D_{X_i}=\{x\in \R: \p\(X_i\in (x-\delta,x+\delta)\)>0~\mbox{for all}~\delta>0\}, ~i=1,\ldots, n.$
 
 Similar to those pointed out in  \cite{yang2013tail}, three common choices for the kernels $\phi_i,~i=1,\ldots,n$ are listed below:\\
 (i) $\phi_i(x)=1-2F(x)$ for all $x\in D_{X_i}$, leading to the well-known standard FGM distribution;\\
 (ii) $\phi_i(x)=x^r-\E[X_i^r]$ for all $x\in D_{X_i}$ and there exists  $r\in \R$ such that $\E[X^r_i] < \infty$;\\
(iii) $\phi_i(x)=e^{-x}-\E[e^{-X_i}]$ for all $x\in D_{X_i}$.

 We further discuss the dependence structure of two rvs $(X_1, X_2)$ following a Sarmanov distribution with different kernel functions. 
To model the dependence between the two rvs $X_1$ and $X_2$, we shall use the Pearson correlation coefficient, which is defined as
\begin{align*}
   \rho_{12}=\frac{\E[X_1X_2]-\E[X_1]\E[X_2]}{\sqrt{Var(X_1)}\sqrt{Var(X_2)}}.
\end{align*}
In the case of the Sarmanov distribution, $\rho_{12}$ can be rewritten as
\begin{align} \label{eq:rho}  
\rho_{12}=\frac{a_{12}\E[X_1\phi_1(X_1)]\E[X_2\phi_2(X_2)]}{\sqrt{Var(X_1)}\sqrt{Var(X_2)}}.
\end{align}
Based on \eqref{eq:rho}, we hereafter present Pearson correlation coefficients for different kernel functions along with their maximal and minimal values.


\noindent \textbf{Case 1:} Set the kernel function $\phi_i(x)=1-2F(x)$, which corresponds to the FGM distribution. It is well known that a Pearson correlation coefficient $\rho_{12}$ of the FGM distribution lies between $-\frac{1}{3}$ and $\frac{1}{3}$ (see \cite{schucany1978correlation}), which is an important drawback of the FGM distribution. \cite{huang1984correlation} showed that, by considering the iterated generalization of the FGM distribution proposed by \cite{johnson1977some}, the range of correlation coefficients can be enlarged.

\noindent \textbf{Case 2:} Set the kernel function  $\phi_i(x)=x^r-\E[X_i^r]$. A usual choice is $r = 1$, which leads to a Pearson correlation coefficient $\rho_{12} = a_{12}\sigma_1\sigma_2$.  {In this situation, according to \cite{lee1996properties}, we assume that the correlation coefficient of $X_1$ and $X_2$ exists. Let $T_i~(i = 1,2)$ denote the upper truncation points of the marginals, and consider the marginal densities defined only for non-negative values. Then, $a_{12}$ satisfies the condition that}
\begin{align*}
\max \left\{\frac{-1}{\mu_1\mu_2},\frac{-1}{(T_1-\mu_1)(T_2-\mu_2)}\right\} \leq a_{12} \leq \min \left\{\frac{1}{\mu_1(T_2 - \mu_2)},\frac{1}{\mu_2(T_1 -\mu_1)}\right\},
\end{align*}
where $\mu_i = \E[X_i]$. Then, the maximal and minimal values of the correlation coefficients are, respectively, given by
\begin{align*}
    \rho^{\max}_{12} = \frac{\sigma_1\sigma_2}{\max \{\mu_1(T_2  - \mu_2),\mu_2(T_1  - \mu_1)\}}, \, \rho^{\min}_{12} = \frac{-\sigma_1\sigma_2}{\max \{\mu_1\mu_2, (T_1 -\mu_1)(T_2-\mu_2)\}}.
\end{align*}

\noindent \textbf{Case 3:} Set the kernel function $\phi_i(x)=e^{-x}-\E[e^{-X_i}]$ for all $x\in D_{X_i}$. In this case, if the correlation coefficient of $X_1$ and $X_2$ exists, 
the range of $a_{12}$ is (see \cite{lee1996properties})
\begin{align*}
    \frac{-1}{\max \{\mathcal{L}_1(1)\mathcal{L}_2(1), (1-\mathcal{L}_1(1))(1-\mathcal{L}_2(1))\}} \leq a_{12} \leq  \frac{1}{\max \{\mathcal{L}_1(1)(1-\mathcal{L}_2(1)), \mathcal{L}_2(1)(1-\mathcal{L}_1(1))\}},
\end{align*}
where $\mathcal{L}_i(t) = \int^\infty_0 \exp(-tx_i)\d F_i(x_i)$. Then, according to (\ref{eq:rho}), the maximal value of the  Pearson correlation coefficient $\rho_{12}$ can be written as 
\begin{align*}
    \rho^{\max}_{12} = \frac{[-\mathcal{L}^{'}_1(1) - \mathcal{L}_1(1)\mu_1][-\mathcal{L}^{'}_2(1) - \mathcal{L}_2(1)\mu_2]}{\max \{\mathcal{L}_1(1)(1-\mathcal{L}_2(1)), \mathcal{L}_2(1)(1-\mathcal{L}_1(1))\}\sigma_1 \sigma_2},
\end{align*}
and the minimal value can be expressed as
\begin{align*}
    \rho^{\min}_{12} = -\frac{[-\mathcal{L}^{'}_1(1) - \mathcal{L}_1(1)\mu_1][-\mathcal{L}^{'}_2(1) - \mathcal{L}_2(1)\mu_2]}{\max \{\mathcal{L}_1(1)\mathcal{L}_2(1), (1-\mathcal{L}_1(1))(1-\mathcal{L}_2(1))\}\sigma_1 \sigma_2}.
\end{align*}

\setcounter{equation}{0}
\section{Second-order asymptotics of VaR-based systemic risk measures}\label{VaR}

In this section, we study the second-order asymptotics of VaR-based systemic risk measures with multivariate Sarmanov distributions.  Denote the distribution function of the aggregate risk $S_n=\sum_{i=1}^nX_i$ by $G(t):=\p(S_n\leq t)$. Before stating some results, we use the following notation:\\
(i) $\eta_{\alpha}:=\alpha\int_0^{1/2}\((1-x)^{-\alpha}-1\)x^{-\alpha-1}\d x+2^{2\alpha-1}-2^\alpha;$
~~~~(ii) $\mu:=\E[X_1] = \dots = \E[X_n];$\\
(iii)  $\mu_i(t):=\int_0^t x\phi_i(x)\d F(x), ~i=1,\dots, n;$ 
~~~~~~~~~~~~~~~~~~~~(iv) $\mu(t):=\int_0^t x \d F(x)$.

First, we establish the second-order asymptotics of the random sum under multivariate Sarmanov distributions.  

\begin{proposition}\label{pro:sum}
 Let $X_1,\ldots,X_n$ be nonnegative random variables with a common marginal distribution $F$ satisfying $\overline{F}(\cdot)\in 2\RV_{-\alpha,\beta}$ with $\alpha>0$, $\beta\leq0$ and an auxiliary function $A(\cdot)$. Suppose that $(X_1,\ldots,X_n)$ follows an $n$-dimensional Sarmanov distribution given by \eqref{eq:sarmanov} and $\lim\limits_{t\rightarrow\infty}\phi_i(t)=d_i\in \R, \phi_i(\cdot)-d_i\in \RV_{\rho_i}$ with $\rho_i\leq 0$  for each $i=1,\ldots,n$.  Then as $t\rightarrow\infty$, we get that
\begin{align*}
\frac{\overline{G}(t)}{\overline{F}(t)}=n\(1+\widetilde{A_n}(t)\(1+o(1)\)\),
\end{align*}
where
\begin{align}\label{eq:th3.1}
\widetilde{A_n}(t)=\left\{
\begin{array}
[c]{lll}%
\alpha t^{-1}\mu_n^*(t)+o\(|A(t)|+\sum\limits_{i=1}^n|\phi_i(t)-d_i|\), &  &\alpha\geq 1,  \\\\
\eta_{\alpha}\kappa_n\overline{F}(t)+o\(|A(t)|+\sum\limits_{i=1}^n|\phi_i(t)-d_i|\), &  & 0<\alpha<1,  \\
\end{array}
\right.
\end{align}
and
$$\mu_n^*(t):=(n-1)\mu(t)+\sum\limits_{1\leq i< j\leq n}\frac{a_{ij}\(d_i\mu_j(t)+d_j\mu_i(t)\)}{n},$$
$$\kappa_n:=n-1+\sum\limits_{1\leq i< j\leq n}\frac{2a_{ij}d_id_j}{n}.$$
\end{proposition}
\begin{proof} We require Lemma \ref{lem:2rv} in the Appendix for this proof. 
 For $t>0$, denote the region $\Omega_t=\{(x_1,\ldots,x_n)\in \R_+^n:\sum_{i=1}^nx_i> t\}$. In addition, let $X_1^*,\ldots,X_n^*$ be iid with a distribution function $F$.   According to Proposition 1.1 of \cite{yang2013tail}, there exist $n$ constants $c_i>1, ~i=1,\ldots,n,$ such that $|\phi_i(x_i)|\leq c_i-1$ for all $x_i\in D_{X_i}$. Let $\widetilde{X_1^*},\ldots,\widetilde{X_n^*}$ be mutually independent rvs with marginal distributions $\widetilde{F_1},\ldots,\widetilde{F_n}$, which are also independent of $X_1^*,\ldots,X_n^*$. Particularly, $\widetilde{F_1},\ldots,\widetilde{F_n}$ are defined by
 \begin{align*}
\d\widetilde{F_i}(x_i):=\(1-\frac{\phi_{i}(x_i)}{c_i}\)\d F(x_i), ~~~~i=1,\ldots, n.
 \end{align*}
 By Lemma \ref{lem:2rv} in the Appendix, we have   $\overline{\widetilde{F_i}}(\cdot)\in 2\RV_{-\alpha,\gamma_i}$ with $\gamma_i=\max\{\beta,\rho_i\}$ and an auxiliary function $\widetilde{A^i}(\cdot)=A(\cdot)-\frac{\rho_i\alpha}{(c_i-d_i)(\alpha-\rho_i)}\(\phi_i(\cdot)-d_i\)$. In addition, as $t\rightarrow\infty$, we obtain
 $$\frac{\overline{\widetilde{F_i}}(t)}{\overline{F}(t)}=\(1-\frac{d_i}{c_i}\)\(1-\frac{\alpha}{(c_i-d_i)(\alpha-\rho_i)}(\phi_i(t)-d_i)(1+o(1))\),$$ for all $i=1,\dots, n.$
  Recall $\mu_i(t):=\int_0^t x\phi_i(x)\d F(x), ~i=1,\dots, n$. We have that
$$\int_0^t x\d\widetilde{F_i}( x)=\int_0^t x\(1-\frac{\phi_{i}(x)}{c_i}\)\d F(x)=\mu(t)-\frac{\mu_i(t)}{c_i}.$$
Next, we can split $\overline{G}(t)$ as
 \begin{align}\label{eq:3.4}
 \overline{G}(t)&=\int_{\Omega_t}\(1+\sum_{1\leq i< j\leq n}a_{ij}\phi_{i}(x_i)\phi_j(x_j)\)\prod_{k=1}^{n}\d F(x_k)\nonumber\\
 &=\int_{\Omega_t}\Bigg(1+\sum_{1\leq i<j\leq n}a_{ij}c_ic_j\(1-\(1-\frac{\phi_{i}(x_i)}{c_i}\)-\(1-\frac{\phi_{j}(x_j)}{c_j}\)\right.\nonumber\\
 &\left.\quad+\(1-\frac{\phi_{i}(x_i)}{c_i}\)\(1-\frac{\phi_{j}(x_j)}{c_j}\)\)\Bigg)\prod_{k=1}^{n}\d F(x_k)\nonumber\\
 &=\int_{\Omega_t}\prod_{k=1}^{n}\d F(x_k)+\sum_{1\leq i< j\leq n}a_{ij}c_i c_j\(\int_{\Omega_t}\prod_{k=1}^{n}\d F(x_k)-\int_{\Omega_t}\prod_{k=1,k\neq i}^{n}\d F(x_k)\d\widetilde{F_i}(x_i)\right.\nonumber\\
&\quad\left.-\int_{\Omega_t}\prod_{k=1,k\neq j}^{n}\d F(x_k)\d\widetilde{F_j}(x_j)+\int_{\Omega_t}\prod_{k=1,k\neq i,j}^{n}\d F(x_k)\d\widetilde{F_i}(x_i)\d\widetilde{F_j}(x_j)\)\nonumber\\
&:=I(t)+\sum_{1\leq i< j\leq n}a_{ij}c_i c_j\(I(t)-I_{i}(t)-I_{j}(t)+I_{i,j}(t)\).
 \end{align}
 To deal with $I(t)$, according to Propositions 3.6-3.7 and Remark 3.1 of \cite{mao2015second} with a common distribution $\overline{F}(\cdot)\in 2\RV_{-\alpha,\beta}$, it follows that
 \begin{align}
\frac{I(t)}{\overline{F}(t)}=\left\{
\begin{array}
[c]{lll}%
n\(1+(n-1)\alpha t^{-1}\mu(t)\(1+o(1)\)\)+o\(|A(t)|\), &  &\alpha\geq 1,\\ \\
n\(1+(n-1)\eta_{\alpha}\overline{F}(t)\(1+o(1)\)\)+o\(|A(t)|\), &  & 0<\alpha<1.\nonumber
\end{array}
\right.
\end{align}
By a similar analysis, $\overline{\widetilde{F_i}}(\cdot)\in 2\RV_{-\alpha,\gamma_i}$. According to \eqref{eq:8.1} in the Appendix, we have
 \begin{align}
\frac{I_i(t)}{\overline{F}(t)}=\left\{
\begin{array}
[c]{lll}%
\(n-\frac{d_i}{c_i}\)\(1+\alpha (n-1)t^{-1}\mu(t)\(1+o(1)\)\)-\(\frac{\alpha(n-1)\mu_i(t)}{c_it}+\frac{\alpha(\phi_i(t)-d_i)}{c_i(\alpha-\rho_i)}\)(1+o(1))\\
\quad+o\(|A(t)|\), &  &\alpha\geq 1,\\ \\
\(n-\frac{d_i}{c_i}\)+\eta_{\alpha}\(n-1\)\(n-\frac{2d_i}{c_i}\)\overline{F}(t)(1+o(1))-\frac{\alpha(\phi_i(t)-d_i)}{c_i(\alpha-\rho_i)}(1+o(1))+o\(|A(t)|\), &  & 0<\alpha<1,\nonumber
\end{array}
\right.
\end{align}
and
\begin{align}
\frac{I_{i,j}(t)}{\overline{F}(t)}=\left\{
\begin{array}
[c]{lll}%
\(n-\frac{d_i}{c_i}-\frac{d_j}{c_j}\)\(1+\alpha t^{-1}(n-1)\mu(t)\(1+o(1)\)\)-\alpha\(\(n-1-\frac{d_j}{c_j}\)\frac{\mu_i(t)}{c_i t}\right.\\
\left.\quad+\(n-1-\frac{d_i}{c_i}\)\frac{\mu_j(t)}{c_j t}+\frac{(\phi_i(t)-d_i)}{c_i(\alpha-\rho_i)}+\frac{(\phi_j(t)-d_j)}{c_j(\alpha-\rho_j)}\)\(1+o(1)\)+o\(|A(t)|\), &  &\alpha\geq 1, \\ \\
\(n-\frac{d_i}{c_i}-\frac{d_j}{c_j}\)+\eta_{\alpha}\((n-1)\(n-2\(\frac{d_i}{c_i}+\frac{d_j}{c_j}\)\)+\frac{2d_id_j}{c_i c_j}\)\overline{F}(t)\(1+o(1)\)\\
\quad-\(\frac{\alpha(\phi_i(t)-d_i)}{c_i(\alpha-\rho_i)}+\frac{\alpha(\phi_j(t)-d_j)}{c_j(\alpha-\rho_j)}\)\(1+o(1)\)+o\(|A(t)|\), &  & 0<\alpha<1.\nonumber
\end{array}
\right.
\end{align}
Plugging all the asymptotics for $I(t)$, $I_{i}(t)$, and $I_{i,j}(t)$ into \eqref{eq:3.4} yields that
\begin{align}
\frac{\overline{G}(t)}{n\overline{F}(t)}=\left\{
\begin{array}
[c]{lll}%
1+\alpha t^{-1}\((n-1)\mu(t)+\sum\limits_{1\leq i< j\leq n}\frac{a_{ij}\(d_i\mu_j(t)+d_j\mu_i(t)\)}{n}\)(1+o(1))\\
+o\(|A(t)|+\sum\limits_{i=1}^n|\phi_i(t)-d_i|\), &  &\alpha\geq 1,\\ \\
1+\eta_{\alpha}\(n-1+\sum\limits_{1\leq i< j\leq n}\frac{2a_{ij}d_id_j}{n}\)\overline{F}(t)\(1+o(1)\)\\
+o\(|A(t)|+\sum\limits_{i=1}^n|\phi_i(t)-d_i|\), &  & 0<\alpha<1. \nonumber
\end{array}
\right.
\end{align}
 Thus, this completes the proof.
\end{proof}
\begin{remark}
Recall the notation in Section \ref{pre}. If $\phi_i(x)=1-2F(x)$ for all $x\in D_{X_i}, ~i=1,\ldots,n$, then $d_i=-1$ and $\rho_i=-\alpha$.  If $\phi_i(x)=x^r-\E[X_i^r]$, then $d_i=-\E[X_i^r]$ and $\rho_i=r$ for all $r\leq 0$. If $\phi_i(x)=e^{-x}-\E[e^{-X_i}]$, then $d_i=-\E[e^{-X_i}]$ and $\rho_i=-\infty$. 
\end{remark}

In view of Proposition \ref{pro:sum}, we can easily obtain the $2\RV$ property of $\overline{G}(\cdot)$.
\begin{corollary}\label{cor:3.1}
Under the conditions of Proposition \ref{pro:sum}, we have $\overline{G}(\cdot)\in 2\RV_{-\alpha,\lambda}$ with $\lambda=\max\{-1,-\alpha,\beta\}$ and an auxiliary function $A_G^{(n)}(\cdot)$ given by
\begin{align}\label{eq:cor3.1}
A_G^{(n)}(t)=\left\{
\begin{array}
[c]{lll}%
A(t)-\alpha t^{-1}\mu_n^*(t), & & \alpha\geq1,\\ \\
A(t)-\alpha\eta_{\alpha}\kappa_n\overline{F}(t), &  & 0<\alpha<1.
\end{array}
\right.
\end{align}
\end{corollary}

\begin{proof}
According to the definition of $\widetilde{A_n}(\cdot)$, it is easy to check that $\widetilde{A_n}(\cdot)\in \RV_{\widetilde{\lambda}}$, where $\widetilde{\lambda}=\max\{-1,-\alpha\}$. Note that  $\overline{F}(\cdot)\in 2\RV_{-\alpha,\beta}$ with an auxiliary function $A(\cdot)$. For any $x>0$, as $t\rightarrow\infty$, we have that
\begin{align*}
\frac{\overline{G}(tx)}{\overline{G}(t)}&=\frac{\overline{G}(tx)}{\overline{F}(tx)}\frac{\overline{F}(tx)}{\overline{F}(t)}\frac{\overline{F}(t)}{\overline{G}(t)}\\
&=\frac{n\(1+\widetilde{A_n}(tx)(1+o(1))\)}{n\(1+\widetilde{A_n}(t)(1+o(1))\)}\(x^{-\alpha}+H_{-\alpha,\beta}(x)A(t)\(1+o(1)\)\)\\
&=x^{-\alpha}+H_{-\alpha,\beta}(x)A(t)\(1+o(1)\)+x^{-\alpha}(x^{\widetilde{\lambda}}-1)\widetilde{A_n}(t)\(1+o(1)\)\\
&=\left\{
\begin{array}
[c]{lll}%
x^{-\alpha}+\(H_{-\alpha,\beta}(x)A(t)-H_{-\alpha,-1}(x)\alpha t^{-1}\mu_n^*(t)\)\(1+o(1)\)+o\(\sum\limits_{i=1}^n|\phi_i(t)-d_i|\), &  &\alpha\geq1, \\ 
x^{-\alpha}+\(H_{-\alpha,\beta}(x)A(t)-H_{-\alpha,-\alpha}(x)\alpha\eta_{\alpha}\kappa_n\overline{F}(t)\)\(1+o(1)\)+o\(\sum\limits_{i=1}^n|\phi_i(t)-d_i|\), &  & 0<\alpha<1.
\end{array}
\right.
\end{align*}
Thus, we complete this proof.
\end{proof}

\begin{remark}
(1) If $\phi_i(x)=1-2F(x)$ for all $x\in D_{X_i},~ i=1,\dots,n$, Proposition \ref{pro:sum} and Corollary \ref{cor:3.1} reduce to Theorem 4.4 and Corollary 4.5 of \cite{mao2015risk}.

(2) If $\alpha\geq 1$ and $\beta<-1$, then $A(t)=o(\mu_n^*(t))$. If $\alpha<1$ and $\beta<-\alpha$, then $A(t)=o(\overline{F}(t))$. If $\beta>-\(1\wedge\alpha\)$, then $\mu_n^*(t)=o(A(t))$ and $\overline{F}(t)=o(A(t))$.
\end{remark}

Second, we are ready to show the second-order asymptotics of $\VaR_p(S_n)$ and $\CTE_p(S_n)$.

\begin{theorem}\label{the:3.1}
Under the conditions of Proposition \ref{pro:sum}, we have, as $p\uparrow1$,
\begin{align*}
    \frac{\VaR_p(S_n)}{F^{\leftarrow}(p)}=\left\{
\begin{array}
[c]{lll}%
n^{1/\alpha}\(1+\(\frac{\mu_n^*\(F^{\leftarrow}(p)\)}{n^{1/\alpha}F^{\leftarrow}(p)}+\frac{n^{\beta/\alpha}-1}{\alpha\beta}A\(F^{\leftarrow}(p)\)\)\(1+o(1)\)\)\\
\quad +o\(\sum\limits_{i=1}^n|\phi_i\(F^{\leftarrow}(p)\)-d_i|\), &  &\alpha\geq 1,\\ \\
n^{1/\alpha}\(1+\(\frac{\eta_{\alpha}\kappa_n}{\alpha n}\overline{F}\(F^{\leftarrow}(p)\)+\frac{n^{\beta/\alpha}-1}{\alpha\beta}A\(F^{\leftarrow}(p)\)\)\(1+o(1)\)\)\\
\quad+o\(\sum\limits_{i=1}^n|\phi_i\(F^{\leftarrow}(p)\)-d_i|\), &  & 0<\alpha<1.
\end{array}
\right.
\end{align*}
For $\alpha>1$, as $p\uparrow1$,
\begin{align*}
    \frac{\CTE_p(S_n)}{F^{\leftarrow}(p)}&=
\frac{\alpha n^{1/\alpha}}{\alpha-1}\(1+\zeta_{\alpha,\beta}^nA\(F^{\leftarrow}(p)\)\(1+o(1)\)\)+\frac{\mu_n^*\(F^{\leftarrow}(p)\)}{F^{\leftarrow}(p)}\(1+o(1)\)\nonumber\\
&\quad+o\(\sum\limits_{i=1}^n|\phi_i\(F^{\leftarrow}(p)\)-d_i|\),
\end{align*}
and
\begin{align}\label{zate}
\zeta_{\alpha,\beta}^n=\left\{
\begin{array}
[c]{lll}%
\frac{1}{\alpha\beta}\(\frac{n^{\beta/\alpha}(\alpha-1)}{\alpha-\beta-1}-1\),& & \beta<0,\\ \\
\alpha^{-2}\log n+\frac{1}{\alpha(\alpha-1)}, & & \beta=0.
\end{array}
\right.
\end{align}
 Clearly, the first-order asymptotics of $\VaR_p(S_n)$ and $\CTE_p(S_n)$ are $n^{1/\alpha}F^{\leftarrow}(p)$ and $\frac{\alpha n^{1/\alpha}}{\alpha-1}F^{\leftarrow}(p)$.
\end{theorem}
\begin{proof}
    Since $\overline{F}(\cdot)\in 2\RV_{-\alpha,\beta}$ with an auxiliary function $A(\cdot)$, by Theorem 2.3.9 of \cite{de2006extreme}, one can easily check that $U_F(\cdot)\in 2\RV_{1/\alpha,\beta/\alpha}$ with an auxiliary function $\alpha^{-2}A\circ U_F(\cdot)$. Let $t=G^{\leftarrow}(p)$. If $p\uparrow 1$, then $t\rightarrow\infty$. By the relation \eqref{eq:2.1}, Lemma 2.1 of \cite{mao2015risk} and Proposition \ref{pro:sum}, we get that
\begin{align*}
\frac{\VaR_p(S_n)}{F^{\leftarrow}(p)}&=\frac{G^{\leftarrow}(p)}{F^{\leftarrow}(p)}=\frac{U_F\(1/\overline{F}(t)\)}{U_F\(1/\overline{G}(t)\)}+o\(A(t)\)\\
&=\(\frac{\overline{G}(t)}{\overline{F}(t)}\)^{1/\alpha}\(1+\frac{\(\frac{\overline{G}(t)}{\overline{F}(t)}\)^{\beta/\alpha}-1}{\beta/\alpha}\alpha^{-2}A\circ U_F\(\frac{1}{\overline{G}(t)}\)\)\\
&=\(n\(1+\widetilde{A_n}(t)(1+o(1))\)\)^{1/\alpha}\(1+\frac{n^{\beta/\alpha}-1}{\alpha\beta}A\circ U_F\(\frac{1}{\overline{G}(t)}\)\(1+o(1)\)\),
\end{align*}
 where we use the transformation $t\mapsto G^{\leftarrow}(p)$ and  {Taylor's expansion}. Note that $\widetilde{A_n}\(G^{\leftarrow}(p)\)\sim\widetilde{A_n}\(n^{1/\alpha}F^{\leftarrow}(p)\)\sim n^{\widetilde{\lambda}/\alpha}\widetilde{A_n}\(F^{\leftarrow}(p)\)$ with $\widetilde{\lambda}=\max\{-1,-\alpha\}$. As $p\uparrow 1$, we have that
\begin{align*}
\frac{\VaR_p(S_n)}{F^{\leftarrow}(p)}&=n^{1/\alpha}\(1+\frac{1}{\alpha  }\widetilde{A_n}\(G^{\leftarrow}(p)\)\(1+o(1)\)+\frac{n^{\beta/\alpha}-1}{\alpha\beta}A\(F^{\leftarrow}(p)\)\(1+o(1)\)\)\\
&=n^{1/\alpha}\(1+\frac{1}{\alpha }n^{\widetilde{\lambda}/\alpha}\widetilde{A_n}\(F^{\leftarrow}(p)\)\(1+o(1)\)+\frac{n^{\beta/\alpha}-1}{\alpha\beta}A\(F^{\leftarrow}(p)\)\(1+o(1)\)\)\\
&=\left\{
\begin{array}
[c]{lll}%
n^{1/\alpha}\(1+\(\frac{\mu_n^*\(F^{\leftarrow}(p)\)}{n^{1/\alpha}F^{\leftarrow}(p)}+\frac{n^{\beta/\alpha}-1}{\alpha\beta}A\(F^{\leftarrow}(p)\)\)\(1+o(1)\)\)\\
\quad+o\(\sum\limits_{i=1}^n|\phi_i\(F^{\leftarrow}(p)\)-d_i|\), &  &\alpha\geq 1,\\
n^{1/\alpha}\(1+\(\frac{\eta_{\alpha}\kappa_n}{\alpha n}\overline{F}\(F^{\leftarrow}(p)\)+\frac{n^{\beta/\alpha}-1}{\alpha\beta}A\(F^{\leftarrow}(p)\)\)\(1+o(1)\)\)\\
\quad+o\(\sum\limits_{i=1}^n|\phi_i\(F^{\leftarrow}(p)\)-d_i|\),&  & 0<\alpha<1.
\end{array}
\right.
\end{align*}
Furthermore, by the definition of $\CTE_p(S_n)$ and Lemma \ref{lem:cr} in the Appendix, for $\alpha>1$, we have
\begin{align*}
\frac{\CTE_p(S_n)}{F^{\leftarrow}(p)}&=\frac{\E\left[S_n\big|S_n>\VaR_p(S_n)\right]}{F^{\leftarrow}(p)}\\
&=\frac{\alpha}{\alpha-1}\(1+\frac{1}{\alpha(\alpha-\lambda-1)}A_G^{(n)}\(G^{\leftarrow}(p)\)(1+o(1))\)\frac{\VaR_p(S_n)}{F^{\leftarrow}(p)}\\
&=\frac{\alpha n^{1/\alpha}}{\alpha-1}\(1+\frac{n^{\lambda/\alpha}}{\alpha(\alpha-\lambda-1)}A_G^{(n)}\(F^{\leftarrow}(p)\)\(1+o(1)\)\)\\
&\quad \cdot \(1+\(\frac{\mu_n^*\(F^{\leftarrow}(p)\)}{n^{1/\alpha}F^{\leftarrow}(p)}+\frac{n^{\beta/\alpha}-1}{\alpha\beta}A\(F^{\leftarrow}(p)\)\)\(1+o(1)\)\)+o\(\sum\limits_{i=1}^n|\phi_i\(F^{\leftarrow}(p)\)-d_i|\)\\
&=\frac{\alpha n^{1/\alpha}}{\alpha-1}\(1+\frac{1}{\alpha\beta}\(\frac{n^{\beta/\alpha}(\alpha-1)}{\alpha-\beta-1}-1\)A\(F^{\leftarrow}(p)\)\(1+o(1)\)\)+\frac{\mu_n^*\(F^{\leftarrow}(p)\)}{F^{\leftarrow}(p)}\(1+o(1)\)\\
&\quad+o\(\sum\limits_{i=1}^n|\phi_i\(F^{\leftarrow}(p)\)-d_i|\).
\end{align*}
This completes the proof.
\end{proof}

\begin{remark}
    For $\VaR_p(S_n)$ in Theorem \ref{the:3.1}, our results cover Theorem 4.6 of \cite{mao2015risk} and provide a simpler proof. If $a_{ij}=0$ for all $1\leq i\neq j\leq n$, the $n$-dimensional Sarmanov distribution reduces to the independence structure. In this case, $\CTE_p(S_n)$ of Theorem \ref{the:3.1} is consistent with Theorem 3.1 of \cite{mao2012second}. 
\end{remark}

The following example is used to illustrate Theorem \ref{the:3.1} under the Pareto distribution with different values of the tail parameter $\alpha$.

\begin{example}\label{ex:pareto}
(Pareto distribution) A Pareto distribution function $F$ satisfies that
$$F(x)=1-\(\frac{k}{x+k}\)^\alpha,~~~~ x>0,$$
 with  parameters $\alpha,k>0$. It can be described that $\overline{F}(\cdot)\in 2\RV_{-\alpha,-1}$  with an auxiliary function $A(t)=\frac{\alpha k}{t}$ and $F^{\leftarrow}(p)=k\((1-p)^{-1/\alpha}-1\)$. If $\alpha>1$, we have $\mu=\frac{k}{\alpha-1}$ and $\mu(t)=\frac{k}{\alpha-1}-\frac{\alpha t+k}{\alpha-1}\overline{F}(t)$. Let $X_1$ and $X_2$ have an identical Pareto distribution $F$. Suppose that the random vector $(X_1, X_2)$ follows a Sarmanov distribution in \eqref{eq:sarmanov} with $\phi_i(\cdot)=1-2F(\cdot), ~i=1,2$. Clearly, $d_i=-1,~\rho_i=-\alpha$ and $\mu_i(t)=\frac{k}{2\alpha-1}-\frac{2\alpha t+k}{2\alpha-1}\(\overline{F}(t)\)^2-\mu(t)$, $i=1,2$. Then
\begin{align*}
    \VaR_p(S_2)=\left\{
\begin{array}
[c]{lll}%
2^{1/\alpha}F^{\leftarrow}(p)+\mu\(F^{\leftarrow}(p)\)-a_{1 2}\mu_1\(F^{\leftarrow}(p)\)+k\(2^{1/\alpha}-1\),
& &\alpha\geq 1,\\ \\
2^{1/\alpha}F^{\leftarrow}(p)\(1+\frac{\eta_{\alpha}(1+a_{1 2})}{2\alpha}(1-p)\)+k\(2^{1/\alpha}-1\), &  & 0<\alpha<1,
\end{array}
\right.
\end{align*}
and 
\begin{align*}
    \CTE_p(S_2)=
\frac{2^{1/\alpha}\alpha}{\alpha-1}\(F^{\leftarrow}(p)+k\)-k+\mu\(F^{\leftarrow}(p)\)-a_{12}\mu_{1}\(F^{\leftarrow}(p)\), \quad \alpha>1.
\end{align*}
\end{example}

\begin{table}[htbp]%
  \fontsize{6}{2}\centering
  \setlength{\tabcolsep}{0.1mm}{
  \begin{tabular}{cccc ccc cccc}%
  \hline\hline\noalign{\smallskip}
  \textbf{$\alpha$} & $\VaR_{p}(S_2)_{MC}$&$\VaR_{p}(S_2)_{1st}$& $\VaR_{p}(S_2)_{2nd}$ &$\CTE_{p}(S_2)_{MC}$&$\CTE_{p}(S_2)_{1st}$&$\CTE_{p}(S_2)_{2nd}$&$\frac{\VaR_{p}(S_2)_{1st}}{\VaR_{p}(S_2)_{MC}}$ &$\frac{\VaR_{p}(S_2)_{2nd}}{\VaR_{p}(S_2)_{MC}}$&$\frac{\CTE_{p}(S_2)_{1st}}{\CTE_{p}(S_2)_{MC}}$&$\frac{\CTE_{p}(S_2)_{2nd}}{\CTE_{p}(S_2)_{MC}}$ \\
\noalign{\smallskip}\hline\noalign{\smallskip}     
  1.1& 126.8065& 121.5690& 126.1943&1561.6015& 1337.2588 &1360.6628& 0.9587& 0.9952 &0.8563& 0.8713\\
1.5&  35.3132&  32.6000&  34.9843&  102.9805& 97.8000 & 103.3591& 0.9232& 0.9907 &0.9497& 1.0037\\
2.0&  14.4215&  12.7262&  14.1893& 28.5141&  25.4524&   28.3298& 0.8824& 0.9839& 0.8926& 0.9935\\
2.5&   8.2435&   7.0060&   8.0581 &13.8310&  11.6767&   13.6085& 0.8499& 0.9775& 0.8442& 0.9839\\
3.0&   5.5535 &  4.5847 &  5.4053& 8.5145&   6.8771&    8.3277& 0.8255& 0.9733& 0.8077& 0.9781\\
4.0&  3.2479&   2.5708&   3.1404&4.5244&    3.4277&    4.3938& 0.7915& 0.9669& 0.7576& 0.9711\\
5.0&  2.2591 &  1.7378&   2.1739&2.9994&    2.1722&    2.8956& 0.7692 &0.9623& 0.7242& 0.9654 \\
 \noalign{\smallskip}\hline 
  \end{tabular}}
\caption{Simulation values (MC) versus first-order (1st) and second-order (2nd) asymptotic values of $\VaR_p(S_2)$ and $\CTE_p(S_2)$ with various values of $\alpha$. We use the Pareto distribution with $k=1$ and $p=0.99$, $a_{12}=0.5$.}
\label{tab:2}
\end{table}

 In Table \ref{tab:2}, we find that the second-order asymptotics of $\VaR$ and $\CTE$ are closer to the simulation values than the first-order asymptotics.
 Specifically, the asymptotic values of $\CTE$ are comparatively less accurate when $\alpha$ approaches 1, as $\CTE$ may not exist under extremely heavy tails.
Moreover, the accuracy of both asymptotic methods increases as the tail becomes heavier (i.e., smaller $\alpha$). 
 {For a Pareto distribution, the auxiliary function is $A(t) = \frac{\alpha k}{t}$, which is proportional to $\alpha$ for a given threshold $t$. Hence, a larger $\alpha$ amplifies the deviation from the first-order limit for a given threshold $t$, slowing convergence to the extreme value limit.}  
Overall, the results confirm the efficiency of our asymptotics in capturing tail behaviors under extreme scenarios.


The following theorem obtains second-order asymptotics of $\MES$ and $\SES$ under an $n$-dimensional Sarmanov distribution. These results are important in a wide range of systemic risk.

\begin{theorem}\label{the:3.2}
Let $X_1,\ldots,X_n$ be nonnegative random variables with a common marginal distribution $F$ satisfying $\overline{F}(\cdot)\in 2\RV_{-\alpha,\beta}$ with $\alpha>1$, $\beta\leq0$ and an auxiliary function $A(\cdot)$. Suppose that $(X_1,\ldots,X_n)$ follows an $n$-dimensional Sarmanov distribution given by \eqref{eq:sarmanov} and $\lim\limits_{x_i\rightarrow\infty}\phi_i(x_i)=d_i\in \R, \phi_i(\cdot)-d_i\in \RV_{\rho_i}$ with $\rho_i\leq 0$  for each $i=1,\ldots,n$.  Then as $p\uparrow1$, we get that
\begin{align*}
    \frac{\MES_{p,m}(S_n)}{F^{\leftarrow}(p)}=\frac{\alpha n^{1/\alpha} }{(\alpha-1)n}\(1+\zeta_{\alpha,\beta}^nA\(F^\leftarrow(p)\)(1+o(1))\)+o\(\frac{1}{F^{\leftarrow}(p)}+\sum\limits_{i=1}^n\left|\phi_i\(F^{\leftarrow}(p)\)-d_i\right|\),
\end{align*}
where $\zeta_{\alpha,\beta}^n$ is defined in \eqref{zate}
and
\begin{align*}
    \frac{\SES_{p,m}(S_n)}{F^\leftarrow(p)}=\frac{\MES_{p,m}(S_n)}{F^\leftarrow(p)}-\frac{1}{n}.
\end{align*}
Clearly, the first-order asymptotics of $\MES_p(S_n)$ and $\SES_p(S_n)$ are $\frac{\alpha n^{1/\alpha}F^\leftarrow(p)}{(\alpha-1)n}$ and $\frac{\alpha n^{1/\alpha}F^\leftarrow(p)}{(\alpha-1)n}-\frac{F^\leftarrow(p)}{n}$.
\end{theorem}

\begin{proof}
  Here we require Lemma \ref{lem:csum} in the Appendix and Theorem \ref{the:3.1}. Define $\widetilde{B}(t)=\frac{1}{\alpha(\alpha-\beta-1)}A(t)-\frac{\mu_n^*(t)}{t}$. We have that
   $\widetilde{B}(\cdot)\in \RV_\kappa$ with $\kappa=\max\{-1,\beta\}$. It follows that, as $p\uparrow 1$,
\begin{align*}
\frac{\MES_{p,m}(S_n)}{F^{\leftarrow}(p)}&=\frac{\E\left[X_m\big|S_n>\VaR_p(S_n)\right]}{F^{\leftarrow}(p)}\\
&=\frac{\alpha}{(\alpha-1)n}\(1+\widetilde{B}\(\VaR_p(S_n)\)(1+o(1))\)\frac{\VaR_p(S_n)}{F^{\leftarrow}(p)}+o\(\sum\limits_{i=1}^n|\phi_i\(F^{\leftarrow}(p)\)-d_i|\)\nonumber\\
&=\frac{\alpha n^{1/\alpha}}{(\alpha-1)n}\(1+n^{\kappa/\alpha}\widetilde{B}\(F^\leftarrow(p)\)(1+o(1)) \)\nonumber\\
&\quad \cdot \(1+\(\frac{\mu_n^*\(F^{\leftarrow}(p)\)}{n^{1/\alpha}F^{\leftarrow}(p)}+\frac{n^{\beta/\alpha}-1}{\alpha\beta}A\(F^{\leftarrow}(p)\)\)\(1+o(1)\)\)+o\(\sum\limits_{i=1}^n|\phi_i\(F^{\leftarrow}(p)\)-d_i|\)\\
&=\frac{\alpha n^{1/\alpha}}{(\alpha-1)n}\(1+\frac{1}{\alpha\beta}\(\frac{n^{\beta/\alpha}(\alpha-1)}{\alpha-\beta-1}-1\)A\(F^\leftarrow(p)\)\)(1+o(1))\nonumber\\
&\quad+o\(\frac{1}{F^{\leftarrow}(p)}+\sum\limits_{i=1}^n|\phi_i\(F^{\leftarrow}(p)\)-d_i|\).
\end{align*}
Using Lemma \ref{lem:5.2} in the Appendix and Theorem \ref{the:3.1}, we conclude that
\begin{align*}
\frac{\SES_{p,m}(S_n)}{F^\leftarrow(p)}=\frac{\E\left[\(X_m-\VaR_p(X_m)\)_+\big|S_n>\VaR_p(S_n)\right]}{F^\leftarrow(p)}=\frac{\MES_{p,m}(S_n)}{F^\leftarrow(p)}-\frac{1}{n}.
\end{align*}
This completes the proof for $\MES_{p,m}(S_n)$ and $\SES_{p,m}(S_n)$.
\end{proof}

Lastly,  Example \ref{ex:burr} is used to explain Theorem \ref{the:3.2} with a Burr distribution. 

\begin{example}\label{ex:burr}
(Burr distribution) A Burr distribution function $F$ satisfies that
$$F(x)=1-\(1+x^{-\beta}\)^{\alpha/\beta},~~~~ x>1,$$
 with  parameters $\alpha>1$ and $\beta<0$. It is easy to check that $\overline{F}(\cdot)\in 2\RV_{-\alpha,\beta}$  with an auxiliary function $A(t)=\alpha t^{\beta}$, $\mu=-\frac{\alpha}{\beta} B\(\frac{1-\alpha}{\beta},\frac{\beta-1}{\beta}\)$ and $F^{\leftarrow}(p)=\((1-p)^{\beta/\alpha}-1\)^{-1/\beta}$, where $B\(\cdot,\cdot\)$ represents the Beta function. Let $X_1$ and $X_2$ have an identical Burr distribution $F$. Suppose that the random vector $(X_1,X_2)$ follows a Sarmanov distribution in \eqref{eq:sarmanov} with $\phi_i(\cdot)=1-2F(\cdot)$. Clearly, $d_i=-1, ~i=1,2$. Then,  
\begin{align*}
    \MES_{p,1}(S_2)=
\frac{2^{1/\alpha}\alpha F^\leftarrow(p)}{2(\alpha-1)}\(1+\frac{1}{\beta}\(\frac{2^{\beta/\alpha}(\alpha-1)}{\alpha-\beta-1}-1\)\(F^\leftarrow(p)\)^{\beta}\),
\end{align*}
 and 
\begin{align*}
    \SES_{p,1}(S_2)=
\frac{2^{1/\alpha}\alpha F^\leftarrow(p)}{2(\alpha-1)}\(1+\frac{1}{\beta}\(\frac{2^{\beta/\alpha}(\alpha-1)}{\alpha-\beta-1}-1\)\(F^\leftarrow(p)\)^{\beta}\)-\frac{F^\leftarrow(p)}{2}.
\end{align*}
It is easy to see that the first-order asymptotics of $\MES_{p,1}(S_2)$ and $\SES_{p,1}(S_2)$ are $\frac{2^{1/\alpha}\alpha F^\leftarrow(p)}{2(\alpha-1)}$ and $\frac{2^{1/\alpha}\alpha F^\leftarrow(p) }{2(\alpha-1)}-\frac{F^\leftarrow(p)}{2}$. By Figure \ref{f:1}, the second-order asymptotics of $\MES_p(S_2)$ and $\SES_p(S_2)$ are much closer to the simulation values than the first-order asymptotics for $p\in[0.95,1)$.
\begin{figure}[htbp]%
  \begin{minipage}[t]{0.47\textwidth}
    \centering
    \includegraphics[width=\textwidth]{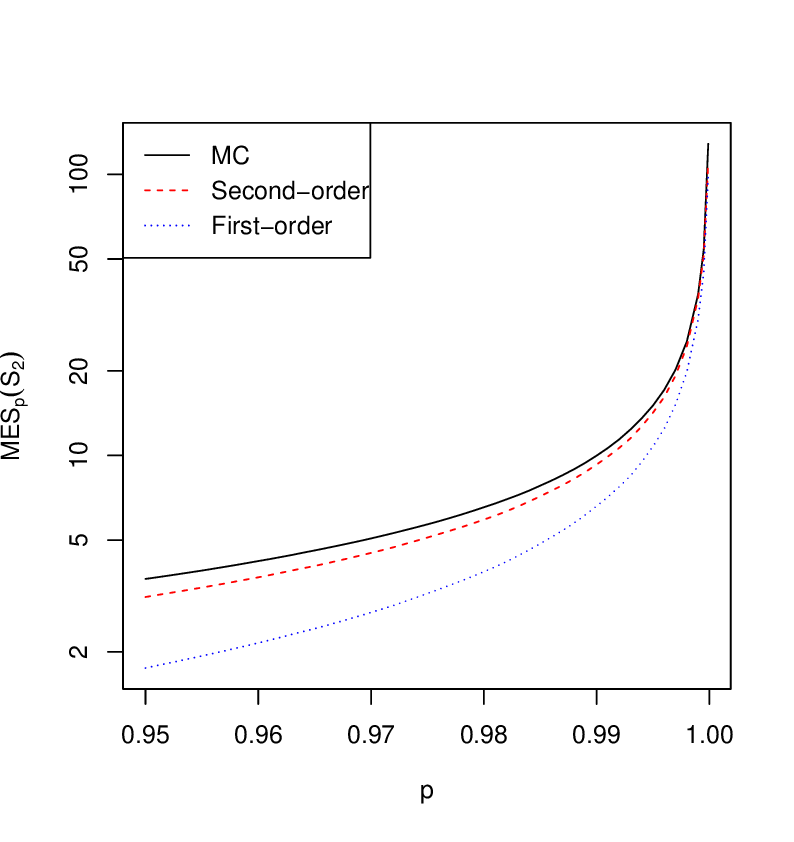}
  \end{minipage}%
  \begin{minipage}[t]{0.47\textwidth}
    \centering
    \includegraphics[width=\textwidth]{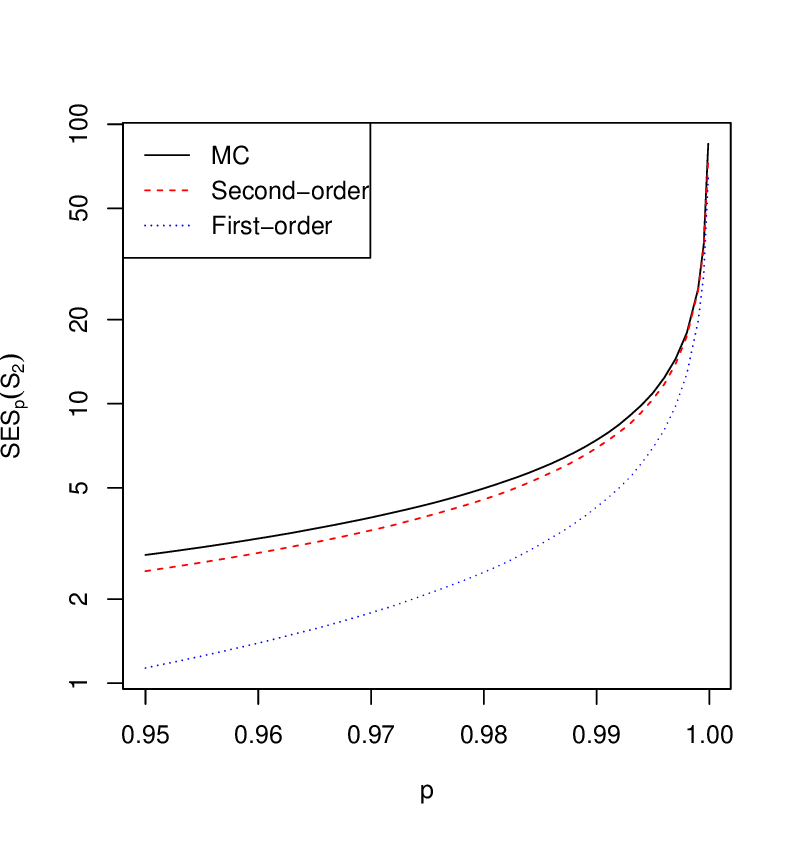}
  \end{minipage}%
 \caption{Simulation values (MC) versus the first-order and second-order asymptotic values of $\MES_{p,1}(S_2)$ for the left panel and $\SES_{p,1}(S_2)$ for the right panel. We use the Burr distribution with $\alpha=2$, $\beta=-0.5$ and $a_{12}=0.5$. To save space in the figures, we use the notation $\MES_{p}(S_2)$ and $\SES_{p}(S_2)$; a same convention is used for the figures below.}\label{f:1}
\end{figure}
\end{example}

\setcounter{equation}{0}
\section{Second-order asymptotics of expectile-based systemic risk measures}\label{expectile}

In this section, we consider the second-order asymptotics of expectile-based systemic risk measures under multivariate Sarmanov distributions. Firstly, we establish the second-order asymptotics of the expectile using a method different from those in the literature.

\begin{proposition}\label{pro:e}
Let $\overline{F}(\cdot)\in 2\RV_{-\alpha,\beta}$ with $\alpha>1,~\beta\leq 0$ and an auxiliary function $A(\cdot)$. Then we have that as $p\uparrow1$,
\begin{align*}
\frac{\e_p(X)}{F^{\leftarrow}(p)}=\(\alpha-1\)^{-1/\alpha}\(1+\xi_{\alpha,\beta}A(F^{\leftarrow}(p))\(1+o(1)\)\)+\frac{\mu}{\alpha F^{\leftarrow}(p)}\(1+o(1)\),
\end{align*}
where
\begin{align}\label{eq:3.1}
\xi_{\alpha,\beta}=\left\{
\begin{array}
[c]{lll}%
\frac{1}{\alpha\beta}\(\frac{(\alpha-1)^{1-\beta/\alpha}}{\alpha-\beta-1}-1\), & & \beta<0,\\ \\
-\alpha^{-2}\log (\alpha-1)+\frac{1}{\alpha(\alpha-1)}, & & \beta=0.
\end{array}
\right.
\end{align}
\end{proposition}
\begin{proof}
See the Appendix.
\end{proof}

 Proposition \ref{pro:e} is derived by a different method with Corollary 1 of \cite{daouia2018estimation}.  In addition, due to the fact that $1-p\sim\overline{F}(F^{\leftarrow}(p))$ as $p\uparrow 1$ and $\alpha>1$, we derive that $1-p=o\(1/F^{\leftarrow}(p)\)$. Thus, Proposition \ref{pro:e} is consistent with Proposition 3.1 of \cite{mao2015asymptotic} and Theorem 3.1 of \cite{mao2015risk}.

Secondly, we obtain the second-order asymptotics of $\e_p(S_n)$ and $\CE_p(S_n)$ with an $n$-dimensional Sarmanov distribution.

\begin{theorem}\label{the:4.1}
Let $X_1,\ldots,X_n$ be nonnegative random variables with a common marginal distribution $F$ satisfying $\overline{F}(\cdot)\in 2\RV_{-\alpha,\beta}$ with $\alpha>1$, $\beta\leq0$ and an auxiliary function $A(\cdot)$. Suppose that ($X_1,\ldots,X_n$) follows an $n$-dimensional Sarmanov distribution given by \eqref{eq:sarmanov} and $\lim\limits_{x_i\rightarrow\infty}\phi_i(x_i)=d_i\in \R, \phi_i(\cdot)-d_i\in \RV_{\rho_i}$ with $\rho_i\leq 0$  for each $i=1,\ldots,n$. Then as $p\uparrow 1$, we get that
 \begin{align*}
\frac{\e_p(S_n)}{F^{\leftarrow}(p)}&=\(\frac{n}{\alpha-1}\)^{1/\alpha}\(1+\frac{1}{\alpha\beta}\(\frac{n^{\beta/\alpha}(\alpha-1)^{1-\beta/\alpha}}{\alpha-\beta-1}-1\)A\(F^{\leftarrow}(p)\)\(1+o(1)\)\)\\
&\quad+\frac{(\alpha-1)\mu_n^*\(F^{\leftarrow}(p)\)+n\mu}{\alpha F^{\leftarrow}(p) }\(1+o(1)\)+o\(\sum\limits_{i=1}^n|\phi_i\(F^{\leftarrow}(p)\)-d_i|\),
\end{align*}
and
\begin{align*}
\frac{\CE_p(S_n)}{F^{\leftarrow}(p)}&=
\frac{\alpha n^{1/\alpha}}{(\alpha-1)^{1/\alpha+1}}\(1+\chi_{\alpha,\beta}^nA\(F^{\leftarrow}(p)\)\(1+o(1)\)\)\\
&\quad+\frac{(\alpha-2)\mu_n^*\(F^{\leftarrow}(p)\)+n\mu }{(\alpha-1)F^{\leftarrow}(p)} \(1+o(1)\)+o\(\sum\limits_{i=1}^n|\phi_i\(F^{\leftarrow}(p)\)-d_i|\),
\end{align*}
where
\begin{align}\label{eq:4.2}
\chi_{\alpha,\beta}^n=\left\{
\begin{array}
[c]{lll}%
\frac{1}{\alpha\beta}\(\(\frac{n}{\alpha-1}\)^{\beta/\alpha}\frac{\alpha+\beta-1}{\alpha-\beta-1}-1\), & & \beta<0,\\ \\
\alpha^{-2}(\log n-\log(\alpha-1))+\frac{2}{\alpha(\alpha-1)}, & & \beta=0.
\end{array}
\right.
\end{align}
Clearly, the first-order asymptotics of $\e_p(S_n)$ and  $\CE_p(S_n)$ are $\frac{n^{1/\alpha}F^{\leftarrow}(p)}{(\alpha-1)^{1/\alpha}}$ and $\frac{\alpha n^{1/\alpha}F^{\leftarrow}(p)}{(\alpha-1)^{1/\alpha+1}}$.
\end{theorem}

\begin{proof}
Firstly,  to deal with $\E[S_n]$, owing to the fact that $(X_1,\ldots,X_n)$ follows an $n$-dimensional Sarmanov distribution in \eqref{eq:sarmanov}, we have
\begin{align*}    \E[S_n]&=\int_{[0,\infty]^n}\sum_{i=1}^nx_i\(1+\sum_{j=i+1}^na_{ij}\phi_{i}(x_i)\phi_j(x_j)\)\prod_{k=1}^{n}\d F(x_k)\\
  &=\sum_{i=1}^n\int_{[0,\infty]^n}x_i\(1+\sum_{j=i+1}^na_{ij}\phi_{i}(x_i)\phi_j(x_j)\)\prod_{k=1}^{n}\d F(x_k)\\
    &=n\mu.
\end{align*}
 According to Proposition \ref{pro:e} and Corollary \ref{cor:3.1}, we have
\begin{align*}
\e_p(S_n)&=\(\alpha-1\)^{-1/\alpha}G^{\leftarrow}(p)\(1+\xi_{\alpha,\lambda}A_{G}^{(n)}\(G^{\leftarrow}(p)\)\(1+o(1)\)\)+\frac{\E[S_n]}{\alpha}\(1+o(1)\)\\
&=\(\frac{n}{\alpha-1}\)^{1/\alpha}F^{\leftarrow}(p)\(1+\(\frac{\mu_n^*\(F^{\leftarrow}(p)\)}{n^{1/\alpha}F^{\leftarrow}(p)}+\frac{n^{\beta/\alpha}-1}{\alpha\beta}A\(F^{\leftarrow}(p)\)\)\(1+o(1)\)\)\\
&\quad \cdot\( 1+n^{\lambda/\alpha}\xi_{\alpha,\lambda}A_{G}^{(n)}\(F^{\leftarrow}(p)\)\(1+o(1)\)\)+\frac{n\mu}{\alpha}\(1+o(1)\)\\
&=\(\frac{n}{\alpha-1}\)^{1/\alpha}F^{\leftarrow}(p)\(1+\frac{1}{\alpha\beta}\(\frac{n^{\beta/\alpha}(\alpha-1)^{1-\beta/\alpha}}{\alpha-\beta-1}-1\)A\(F^{\leftarrow}(p)\)\(1+o(1)\)\)\\
&\quad+\alpha^{-1}\((\alpha-1)\mu_n^*\(F^{\leftarrow}(p)\)+n\mu\)\(1+o(1)\)+o\(\sum\limits_{i=1}^n|\phi_i\(F^{\leftarrow}(p)\)-d_i|\).
\end{align*}
Furthermore, by the definition of $\CE_p(S_n)$ and Lemma \ref{lem:cr} in the Appendix, we obtain, as $p\uparrow 1$,
\begin{align*}
\frac{\CE_p(S_n)}{F^{\leftarrow}(p)}&=\frac{\E\left[S_n\big|S_n>\e_p(S_n)\right]}{F^{\leftarrow}(p)}\\
&=\frac{\alpha}{\alpha-1}\(1+\frac{1}{\alpha(\alpha-\lambda-1)}A_G^{(n)}(\e_p(S_n))(1+o(1))\)\frac{\e_p(S_n)}{F^{\leftarrow}(p)}\\
&=
\frac{\alpha n^{1/\alpha}}{(\alpha-1)^{1/\alpha+1}}\(1+\frac{1}{\alpha\beta}\(\(\frac{n}{\alpha-1}\)^{\beta/\alpha}\frac{\alpha+\beta-1}{\alpha-\beta-1}-1\)A\(F^{\leftarrow}(p)\)\(1+o(1)\)\)\\
&\quad+\frac{(\alpha-2)\mu_n^*\(F^{\leftarrow}(p)\)+n\mu}{(\alpha-1)F^{\leftarrow}(p)}\(1+o(1)\)+o\(\sum\limits_{i=1}^n|\phi_i\(F^{\leftarrow}(p)\)-d_i|\).
\end{align*}
Thus, this completes the proof. 
\end{proof}

The following example is applied to interpret Theorem \ref{the:4.1} with the absolute Student $t_\alpha$ distribution.

 \begin{example}\label{ex:t}
    (Absolute Student $t_\alpha$ distribution) A standard Student $t_\alpha$ distribution has a density function 
    $$f(x)=\frac{\Gamma((\alpha+1)/2)}{\sqrt{\alpha\pi}\Gamma(\alpha/2)}\(1+\frac{x^2}{\alpha}\)^{-(\alpha+1)/2}, ~~~x\in \R, $$
  with a parameter $\alpha>1$. Denote by $F$ the distribution function of $|X|$. According to Example 4.2 of \cite{mao2012second} and Example 4.4 of \cite{hua2011second}, we know that $\overline{F}(\cdot)\in2\RV_{-\alpha,-2}$ with an auxiliary function 
 $A(t)=\frac{\alpha^2}{\alpha+2}t^{-2}$, $\mu=\E|X|=\frac{\alpha}{\alpha-1}$ and $F^{\leftarrow}(p)=t_{\alpha}^{\leftarrow}((1+p)/2)$. Let $X_1$ and $X_2$  have an identical Student $t_\alpha$ distribution $F$.  Suppose that the random vector $(X_1,X_2)$ follows a Sarmanov distribution in \eqref{eq:sarmanov} with $\phi_i(\cdot)=1-2F(\cdot), ~i=1,2$. Clearly, $d_i=-1$. Then
 \begin{align*}
\e_p(S_2)&=\frac{2^{1/\alpha}}{(\alpha-1)^{1/\alpha}}F^{\leftarrow}(p)+\frac{\alpha-1}{\alpha }\big[\mu\(F^{\leftarrow}(p)\)-a_{12}\mu_1\(F^{\leftarrow}(p)\)\big]+\frac{2}{\alpha-1},
\end{align*}
and
 \begin{align*}
\CE_p(S_2)&=\frac{2^{1/\alpha}\alpha}{(\alpha-1)^{1/\alpha+1}}F^{\leftarrow}(p)+\frac{\alpha-2}{\alpha-1 }\big[\mu\(F^{\leftarrow}(p)\)-a_{12}\mu_1\(F^{\leftarrow}(p)\)\big]+\frac{2\alpha}{(\alpha-1)^2  }.
\end{align*}
\begin{table}[htbp]%
   \fontsize{8}{2}\centering
  \setlength{\tabcolsep}{0.5mm}{
  \begin{tabular}{cccc ccc cccc}%
  \hline\hline\noalign{\smallskip}
  $p$ & $\e_{p}(S_2)_{MC}$& $\e_{p}(S_2)_{1st}$& $\e_{p}(S_2)_{2nd}$ &$\CE_{p}(S_2)_{MC}$&$\CE_{p}(S_2)_{1st}$&$\CE_{p}(S_2)_{2nd}$  & $\frac{\e_{p}(S_2)_{1st}}{\e_{p}(S_2)_{MC}}$ & $\frac{\e_{p}(S_2)_{2nd}}{\e_{p}(S_2)_{MC}}$& $\frac{\CE_{p}(S_2)_{1st}}{\CE_{p}(S_2)_{MC}}$ & $\frac{\CE_{p}(S_2)_{2nd}}{\CE_{p}(S_2)_{MC}}$    \\
\noalign{\smallskip}\hline\noalign{\smallskip}
0.9500&  5.7991&  4.0097&  5.5194&  9.2238&  6.6828&  9.0031& 0.6914& 0.9518& 0.7245& 0.9761\\
0.9600&  6.2368&  4.4366&  5.9522 & 9.9000&  7.3943 & 9.7178 &0.7114& 0.9544& 0.7469& 0.9816\\
0.9700&  6.8510&  5.0387&  6.5609& 10.8592&  8.3978& 10.7250& 0.7355& 0.9577& 0.7733& 0.9876\\
0.9800&  7.8270&  6.0016&  7.5313& 12.4033& 10.0027& 12.3340& 0.7668& 0.9622& 0.8065& 0.9944\\
0.9900&  9.8652 & 8.0309&  9.5699& 15.6877& 13.3849& 15.7213& 0.8141& 0.9701& 0.8532& 1.0021\\
0.9950& 12.5120& 10.6764& 12.2216& 20.0218& 17.7940& 20.1339& 0.8533& 0.9768& 0.8887& 1.0056\\
0.9990& 22.2665& 20.4956& 22.0482& 36.1840& 34.1594& 36.5034& 0.9205& 0.9902& 0.9440& 1.0088\\
0.9999 &53.1965& 51.3901& 52.9461& 88.2526& 85.6502& 87.9961& 0.9660& 0.9953& 0.9705& 0.9971\\
  \noalign{\smallskip}\hline
  \end{tabular}}
  \caption{Simulation values (MC) versus first-order (1st) and second-order (2nd) asymptotic values of $\e_p(S_2)$ and $\CE_p(S_2)$ with different $p$. We use the Student $t_\alpha$ distribution with $\alpha=2.5$ and $a_{12}=-0.5$.}
  \label{tab:3}
\end{table}
According to Table \ref{tab:3}, the second-order asymptotics of $\e_p(S_2)$ and $\CE_p(S_2)$ are much closer to the simulation values than their first-order asymptotics for $p\in[0.95,1)$. 
\end{example}
Next, we get the second-order asymptotics of $\ICE_{p,m}(S_n)$ and $\SICE_{p,m}(S_n)$ with an $n$-dimensional Sarmanov distribution. 
\begin{theorem}\label{the:4.2}
Under the conditions of Theorem \ref{the:4.1}, as $p\uparrow1$, we get that
\begin{align*}
\frac{\ICE_{p,m}(S_n)}{F^{\leftarrow}(p)}&=
\frac{\alpha}{(\alpha-1)n}\Bigg(\(\frac{n}{\alpha-1}\)^{1/\alpha}\Big(1+\chi_{\alpha,\beta}^nA\(F^\leftarrow(p)\)(1+o(1))\Big)\\
&\quad+\frac{n\mu-\mu_n^*\(F^\leftarrow(p)\)}{\alpha F^{\leftarrow}(p)}(1+o(1))\Bigg)+o\(\sum\limits_{i=1}^n|\phi_i\(F^{\leftarrow}(p)\)-d_i|\),
\end{align*}
with $\chi_{\alpha,\beta}^n$ defined in  \eqref{eq:4.2} and
\begin{align*}
\frac{\SICE_{p,m}(S_n)}{F^\leftarrow(p)}&=\frac{\ICE_{p,m}(S_n)}{F^\leftarrow(p)}-\frac{\e_p(X_m)}{nF^\leftarrow(p)}.
\end{align*}
Clearly, the first-order asymptotics of $\ICE_{p,m}(S_n)$ and $\SICE_{p,m}(S_n)$ are $\frac{n^{1/\alpha}\alpha F^\leftarrow(p) }{n(\alpha-1)^{1+1/\alpha}}$ and $\frac{n^{1/\alpha}\alpha F^\leftarrow(p) }{n(\alpha-1)^{1+1/\alpha}}- \frac{F^\leftarrow(p)}{n(\alpha-1)^{1/\alpha}}$.
\end{theorem}

\begin{proof}We require Lemmas \ref{lem:csum}-\ref{lem:5.2} in the Appendix for this proof. 
 By the definition of $\ICE_{p,m}(S_n)$, Theorem \ref{the:4.1} and Lemma \ref{lem:csum}, as $p\uparrow1$, we have that
\begin{align*}
\frac{\ICE_{p,m}(S_n)}{F^\leftarrow(p)}&=\frac{\E\left[X_m\big|S_n>\e_p(S_n)\right]}{F^{\leftarrow}(p)}\\
&=\frac{\alpha}{(\alpha-1)n}\(1+\widetilde{B}\(\e_p(S_n)\)(1+o(1))\)\frac{\e_p(S_n)}{F^{\leftarrow}(p)}\nonumber\\
&=\frac{\alpha}{(\alpha-1)n}\Bigg(\(\frac{n}{\alpha-1}\)^{1/\alpha}\(1+\frac{1}{\alpha\beta}\(\frac{n^{\beta/\alpha}(\alpha-1)^{1-\beta/\alpha}}{\alpha-\beta-1}-1\)A\(F^{\leftarrow}(p)\)\(1+o(1)\)\)\\
&\quad+\frac{(\alpha-1)\mu_n^*\(F^{\leftarrow}(p)\)+n\mu}{\alpha F^{\leftarrow}(p)}\(1+o(1)\)\Bigg)\(1+\(\frac{n}{\alpha-1}\)^{\rho/\alpha}\widetilde{B}\(F^\leftarrow(p)\)\(1+o(1)\)\)\nonumber\\
&=
\frac{\alpha }{(\alpha-1)n}\(\(\frac{n}{\alpha-1}\)^{1/\alpha}\(1+\chi_{\alpha,\beta}^nA\(F^\leftarrow(p)\)(1+o(1))\)\right.\\
&\quad\left.+\frac{n\mu-\mu_n^*\(F^{\leftarrow}(p)\)}{\alpha F^{\leftarrow}(p)}(1+o(1))\)+o\(\sum\limits_{i=1}^n|\phi_i\(F^{\leftarrow}(p)\)-d_i|\).
\end{align*}
Using Theorem \ref{the:4.1} and Lemma \ref{lem:5.2}, we conclude that
\begin{align*}
\frac{\SICE_{p,m}(S_n)}{F^\leftarrow(p)}=\frac{\E\left[\(X_m-\e_p(X_m)\)_+|S_n>\e_p(S_n)\right]}{F^\leftarrow(p)}=\frac{\ICE_{p,m}(S_n)}{F^\leftarrow(p)}-\frac{\e_p(X_m)}{nF^\leftarrow(p)}.
\end{align*}
This completes the proof.
\end{proof}


Lastly, the following example illustrates the result of Theorem \ref{the:4.2} under the Fr\'{e}chet distribution.
\begin{example}\label{ex:Frechet}
    (Fr\'{e}chet distribution) A Fr\'{e}chet  distribution function $F$ satisfies that
\begin{align*}
    F(x)=1-\exp(-x^{-\alpha}),~~~~ x>0,
\end{align*}
 with a parameter $\alpha>1$. Clearly, $\overline{F}(\cdot)\in 2\RV_{-\alpha,-\alpha}$ with the auxiliary function $A(t)=\frac{\alpha}{2} t^{-\alpha}$.  Let $X_1$ and $X_2$ have an identical Fr\'{e}chet distribution $F$. Suppose that the random vector $(X_1,X_2)$ follows a Sarmanov distribution in \eqref{eq:sarmanov} with $\phi_i(X_i)=X_i^{-1}-\E[X_i^{-1}]$. Then, $d_i=-\E[X_i^{-1}]$ and $\rho_i=-1,~i=1,2$. Thus,
\begin{align*}
\ICE_{p,1}(S_2)&=
\frac{2^{1/\alpha-1}\alpha F^{\leftarrow}(p)}{(\alpha-1)^{1/\alpha+1}}+\frac{2\mu-\mu\(F^{\leftarrow}(p)\)-a_{12}d_1\mu_1\(F^{\leftarrow}(p)\)}{2(\alpha-1)},
\end{align*}
and
\begin{align*}
\SICE_{p,1}(S_2)=\ICE_{p,1}(S_2)-\(\frac{F^{\leftarrow}(p)}{2(\alpha-1)^{1/\alpha}}+\frac{\mu}{2\alpha }\).
\end{align*}

\begin{table}[t]%
  \centering
  \setlength{\tabcolsep}{2.0mm}{
  \begin{tabular}{cccc cc}%
  \hline\hline\noalign{\smallskip}
  $p$ & $\ICE_{p,1}(S_2)_{MC}$& $\ICE_{p,1}(S_2)_{1st}$& $\ICE_{p,1}(S_2)_{2nd}$  & $\frac{\ICE_{p,1}(S_2)_{1st}}{\ICE_{p,1}(S_2)_{MC}}$ &$\frac{\ICE_{p,1}(S_2)_{2nd}}{\ICE_{p,1}(S_2)_{MC}}$  \\
\noalign{\smallskip}\hline\noalign{\smallskip}
 0.9500 & 8.0518 &  6.2481 &  7.2566 &0.7760& 0.9012\\
 0.9600 &   8.8178 &  7.0051&   7.9763& 0.7944& 0.9046\\
0.9700&   9.9471&   8.1096 &  9.0379& 0.8153& 0.9086\\
 0.9800&    11.8197&   9.9749&10.8508& 0.8439 &0.9180\\
0.9900&  16.0617 & 14.1623&  14.9689&0.8817& 0.9320\\
 0.9990 & 47.6943 & 45.0164 & 45.7053& 0.9439& 0.9583\\
0.9999&146.0483& 144.7467& 145.3974& 0.9911& 0.9955\\
  \noalign{\smallskip}\hline
  \end{tabular}}
  \caption{Simulation values $\ICE_{p,1}(S_2)_{MC}$ versus the first-order asymptotic values $\ICE_{p,1}(S_2)_{1st}$  and the second-order asymptotic values $\ICE_{p,1}(S_2)_{2nd}$ with $\alpha=2,a_{12}=0.5$.}
  \label{tab:4}
\end{table}
\begin{table}[t]%
  \centering
  \setlength{\tabcolsep}{2.0mm}{
  \begin{tabular}{cccc cc}%
  \hline\hline\noalign{\smallskip}
  $p$ & $\SICE_{p,1}(S_2)_{MC}$&$\SICE_{p,1}(S_2)_{1st}$  &$\SICE_{p,1}(S_2)_{2nd}$  & $\frac{\SICE_{p,1}(S_2)_{1st}}{\SICE_{p,1}(S_2)_{MC}}$& $\frac{\SICE_{p,1}(S_2)_{2nd}}{\SICE_{p,1}(S_2)_{MC}}$    \\
\noalign{\smallskip}\hline\noalign{\smallskip}
 0.9500 & 4.4381 &  4.0390 &  4.6044&  0.9101& 1.0375\\
 0.9600 &  4.9085&   4.5284&   5.0565& 0.9226& 1.0301\\
0.9700&  5.6135 &  5.2424 &  5.7276 &0.9339& 1.0203\\
 0.9800&    6.7875&   6.4482&   6.8810&0.9500& 1.0138\\
0.9900&  9.4881  & 9.1552  & 9.5187&0.9649& 1.0032\\
 0.9990 & 29.4512&  29.1007 & 29.4465&0.9881& 0.9998\\
0.9999&93.8148 & 93.5710 & 93.7786& 0.9974&0.9996\\
  \noalign{\smallskip}\hline
  \end{tabular}}
  \caption{Simulation values $\SICE_{p,1}(S_2)_{MC}$ versus the first-order asymptotic values $\SICE_{p,1}(S_2)_{1st}$  and the second-order asymptotic values $\SICE_{p,1}(S_2)_{2nd}$ with $\alpha=2,a_{12}=0.5$.}
  \label{tab:5}
\end{table}
Again, Tables \ref{tab:4}-\ref{tab:5} reveal that the second-order asymptotics of  $\ICE_p(S_2)$ and $\CE_p(S_2)$ are closer to simulation values for $p\in[0.95,1)$ and provide better estimates than the first-order asymptotics.
\end{example}

\setcounter{equation}{0}
\section{Numerical illustration}\label{sim}

In this section, we numerically illustrate our asymptotic results with a comprehensive comparison among different families and types of systemic risk measures.
\begin{figure}[htbp]
\vspace{-0.5cm}
  \begin{minipage}[t]{0.33\textwidth}
    \centering
\includegraphics[width=\textwidth]{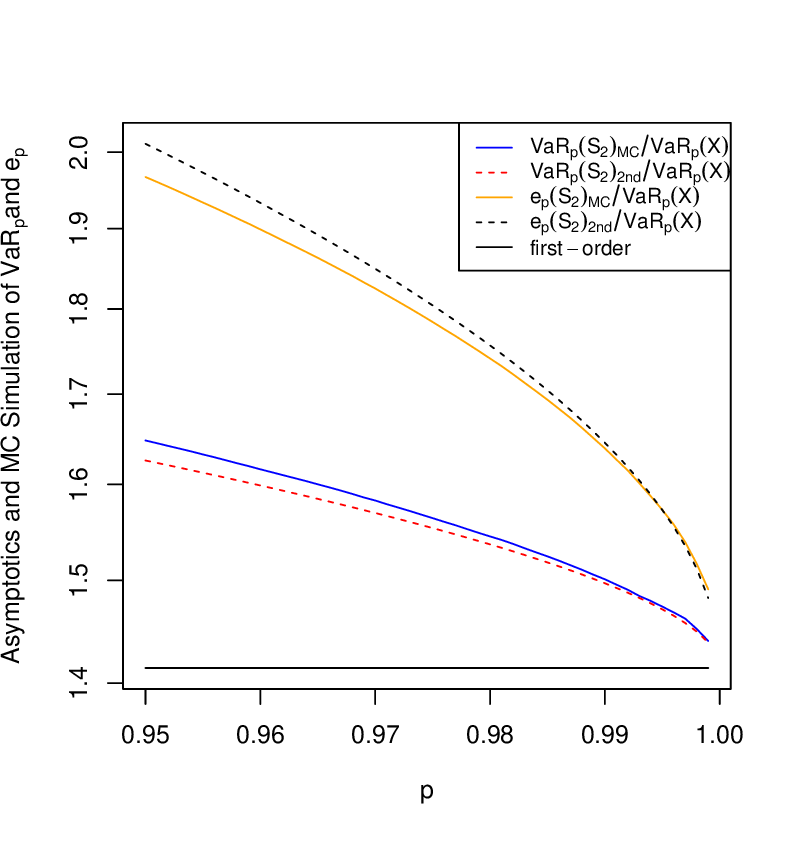}
  \end{minipage}%
  \begin{minipage}[t]{0.33\textwidth}
    \centering \includegraphics[width=\textwidth]{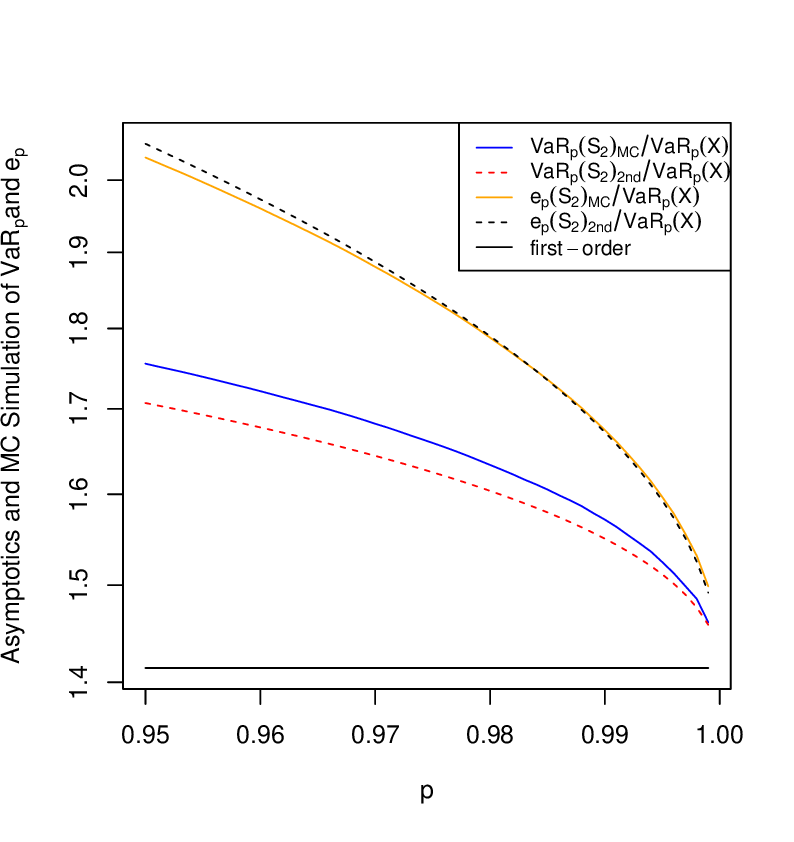}
  \end{minipage}%
\begin{minipage}[t]{0.33\textwidth}
    \centering \includegraphics[width=\textwidth]{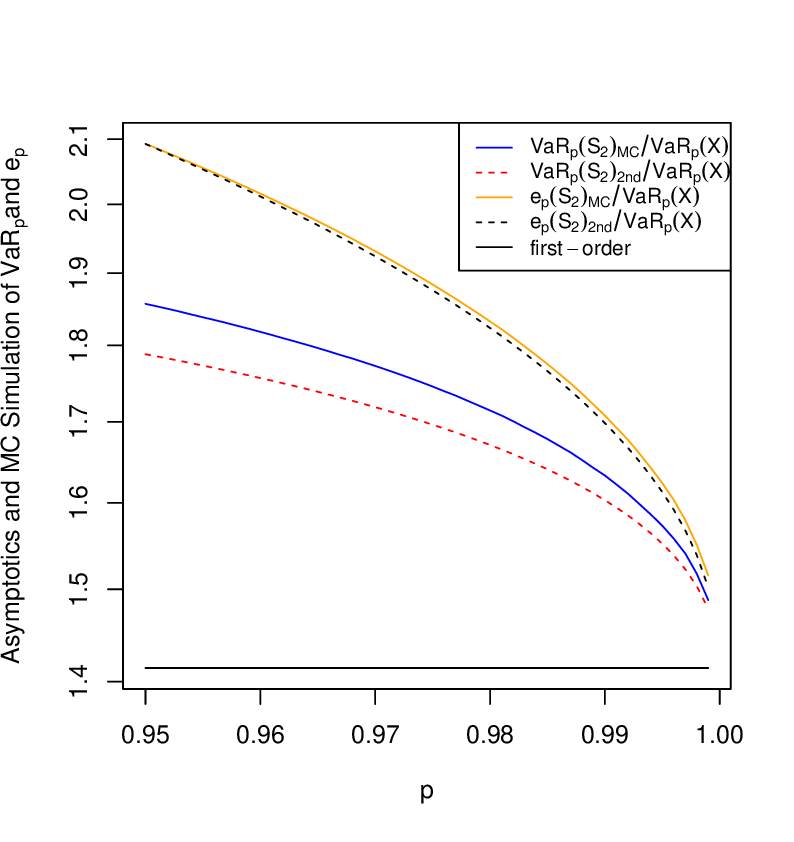}
  \end{minipage}%
\vspace{-0.5cm}  
  \begin{minipage}[t]{0.33\textwidth}
    \centering
\includegraphics[width=\textwidth]{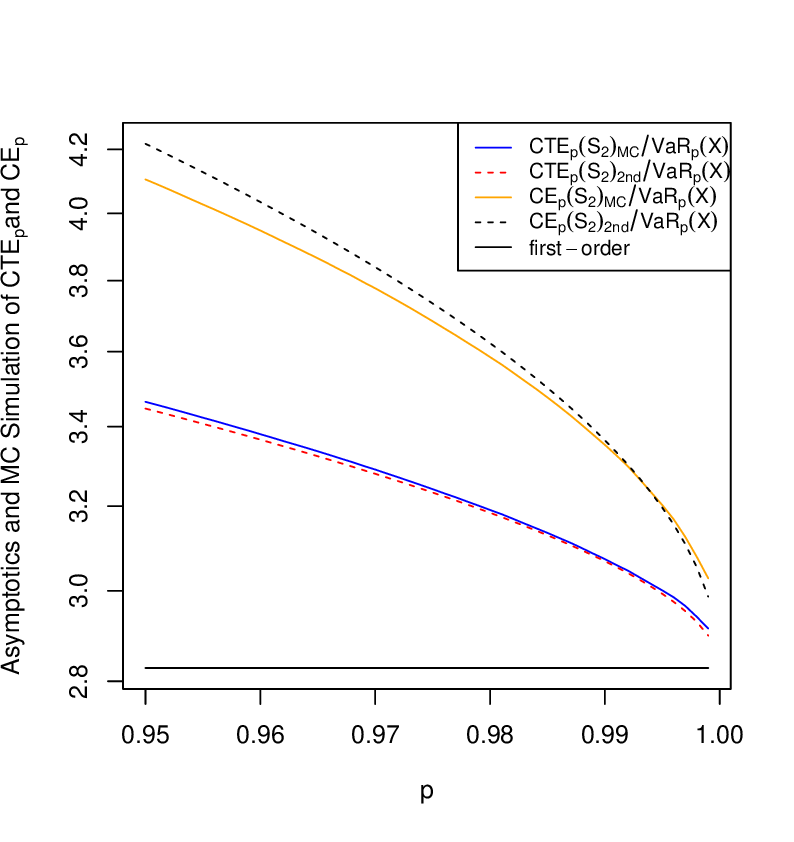}
  \end{minipage}%
  \begin{minipage}[t]{0.33\textwidth}
    \centering
\includegraphics[width=\textwidth]{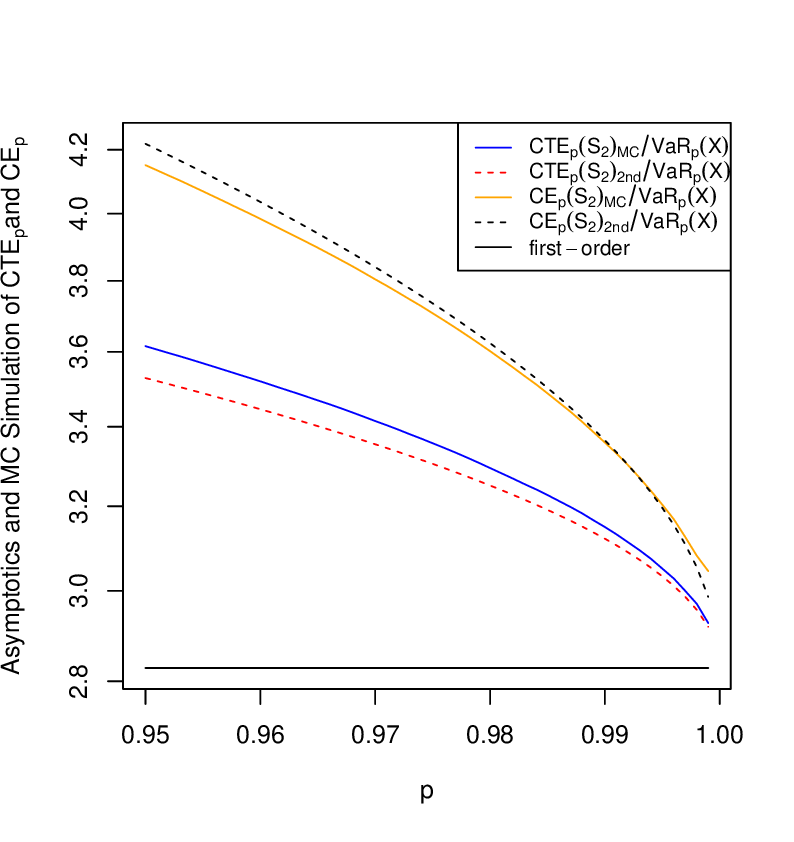}
  \end{minipage}
\begin{minipage}[t]{0.33\textwidth}
    \centering  \includegraphics[width=\textwidth]{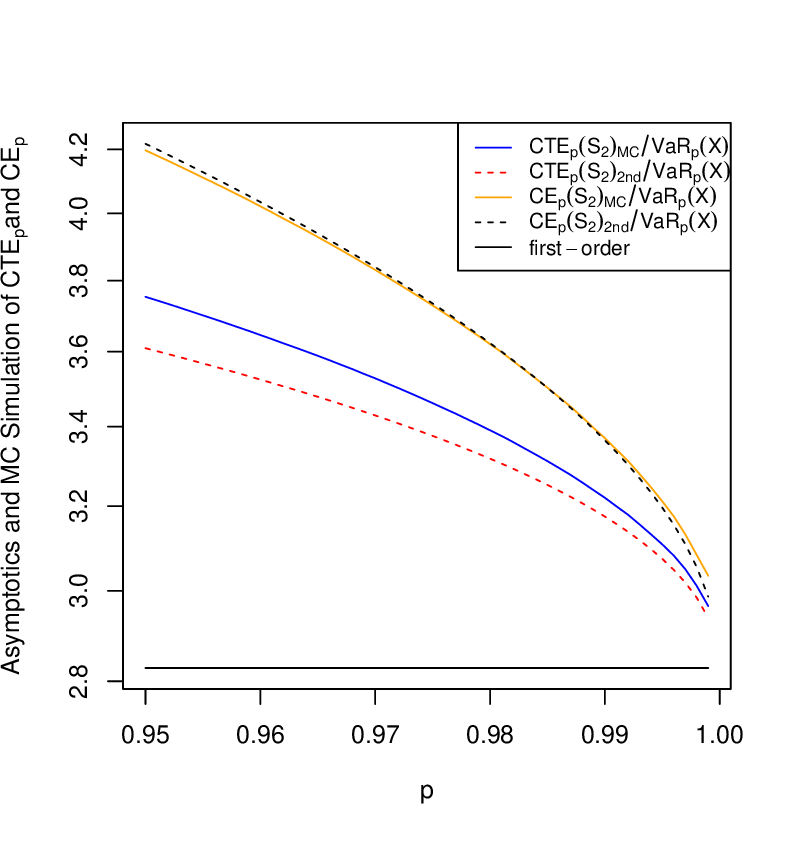} 
    \end{minipage}%
    \vspace{-0.5cm}
  \begin{minipage}[t]{0.33\textwidth}
    \centering
    \includegraphics[width=\textwidth]{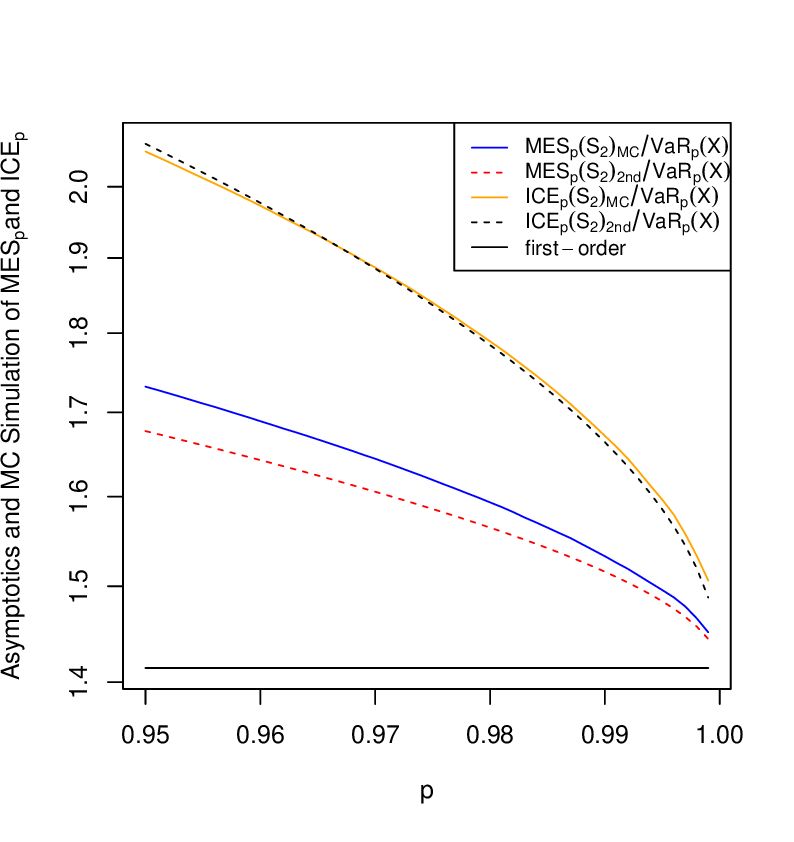}
  \end{minipage}%
  \begin{minipage}[t]{0.33\textwidth}
    \centering
    \includegraphics[width=\textwidth]{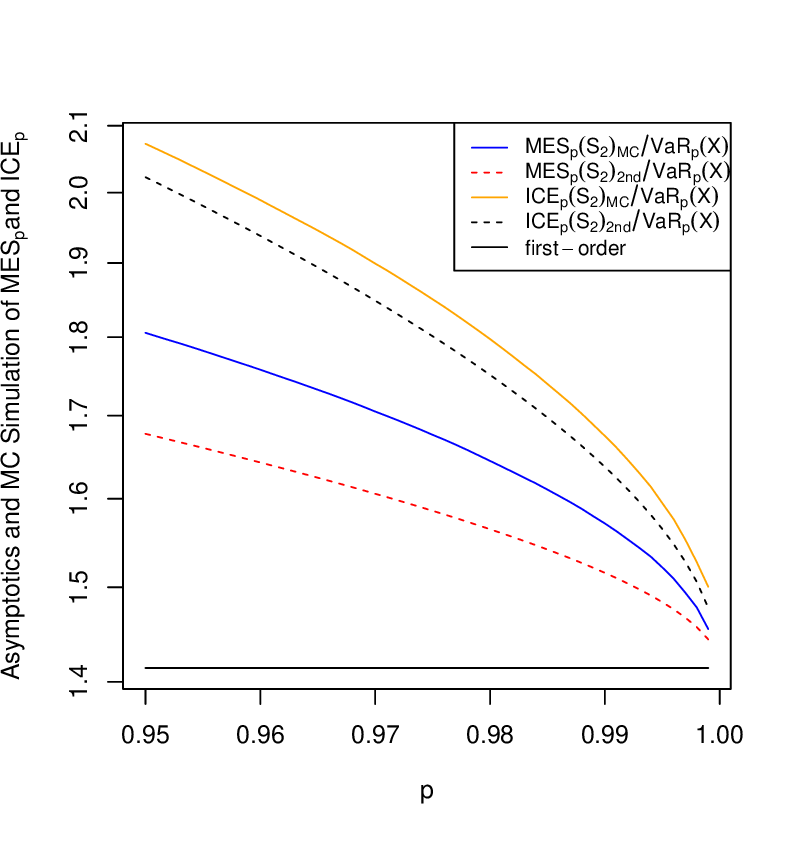}
  \end{minipage}
\begin{minipage}[t]{0.33\textwidth}
    \centering
    \includegraphics[width=\textwidth]{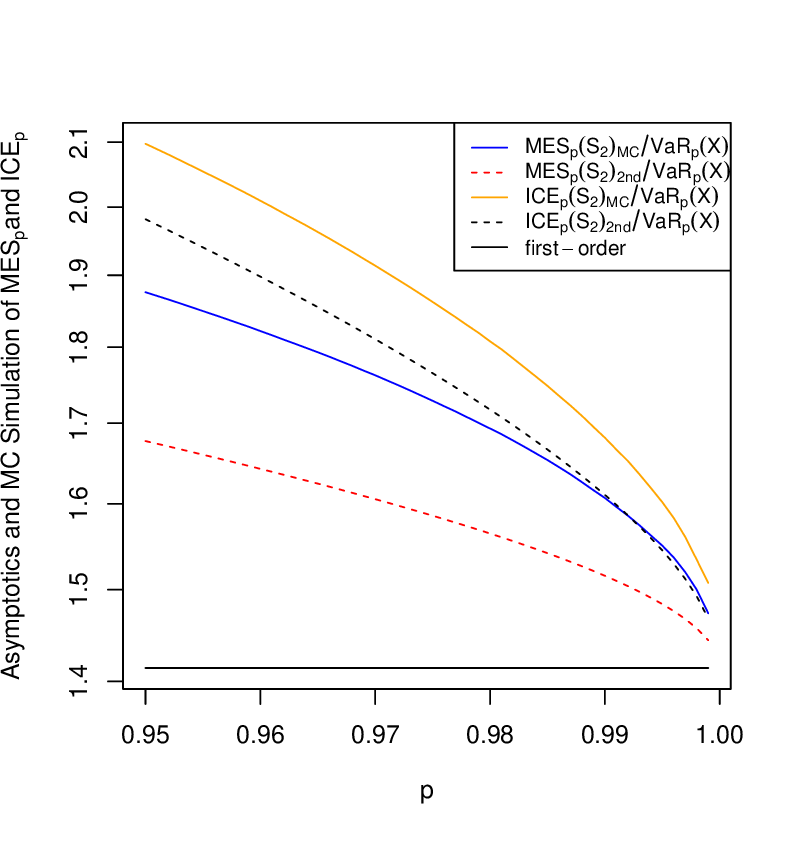}
  \end{minipage}%
  \vspace{-0.5cm}
  \begin{minipage}[t]{0.33\textwidth}
    \centering
    \includegraphics[width=\textwidth]{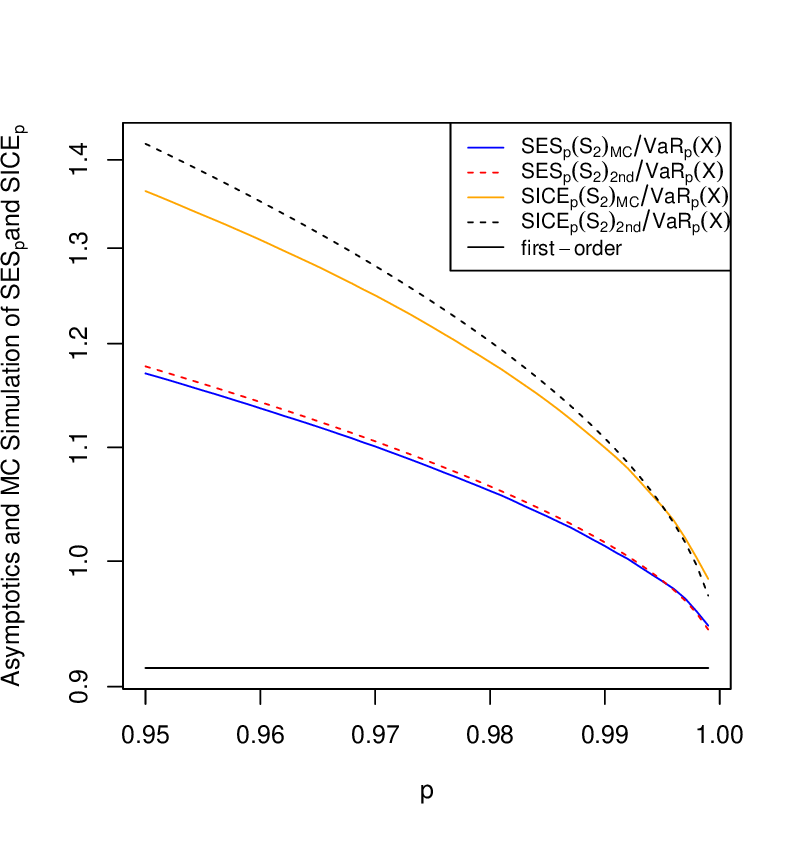}
  \end{minipage}%
  \begin{minipage}[t]{0.33\textwidth}
    \centering
    \includegraphics[width=\textwidth]{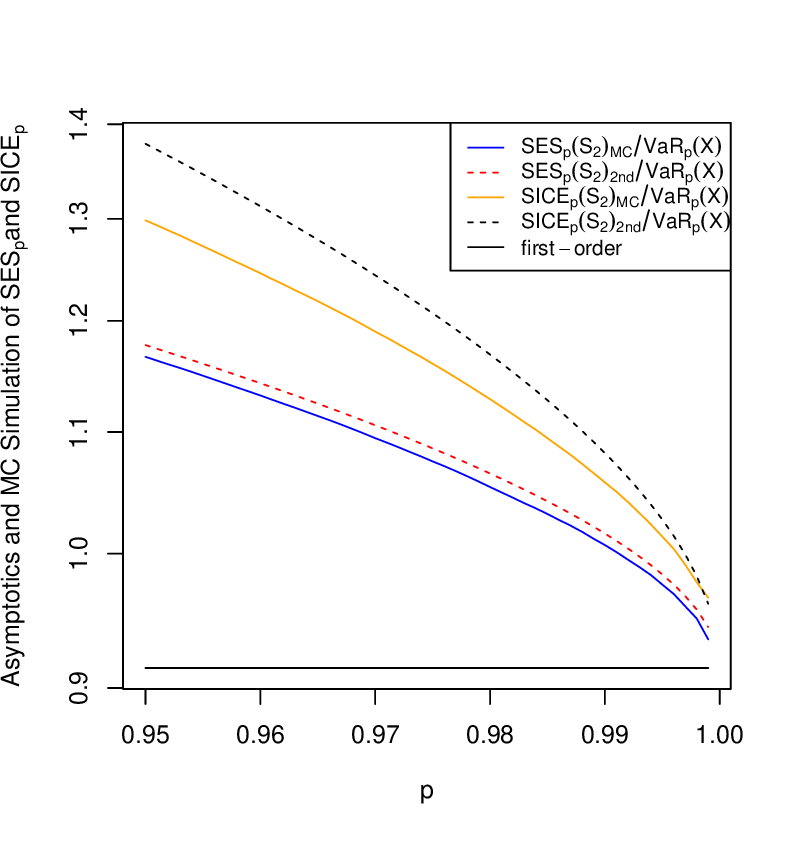}
  \end{minipage}
\begin{minipage}[t]{0.33\textwidth}
    \centering
    \includegraphics[width=\textwidth]{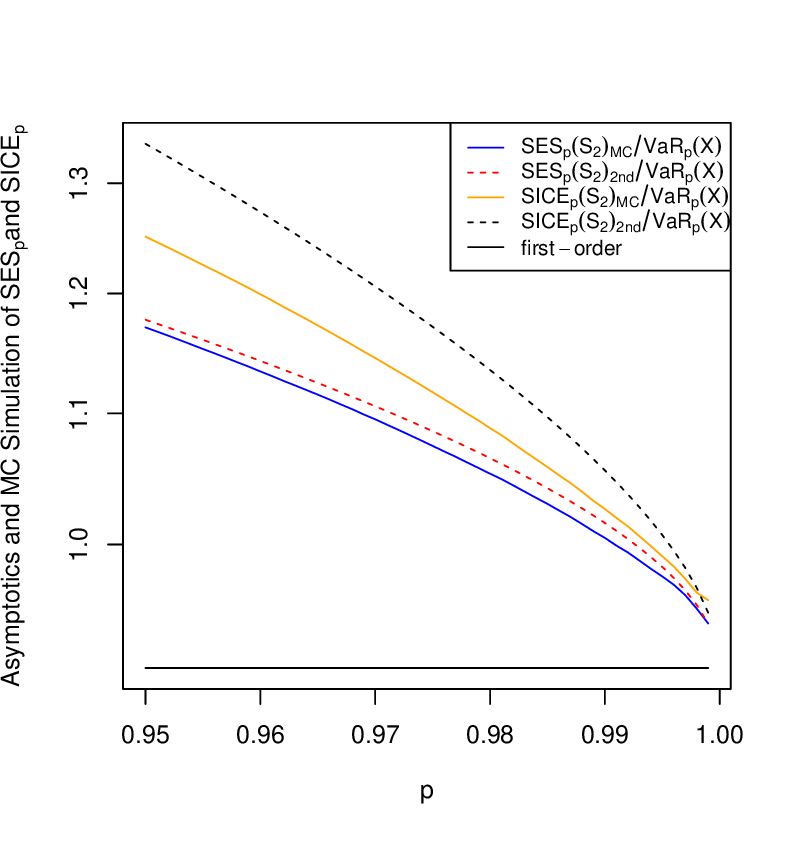}
  \end{minipage}%
  \vspace{-0.3cm}
\caption{\scriptsize Comparison of the simulation values (MC), first-order and second-order asymptotic values of  $\VaR_p(S_2)/\VaR_p(X)$ and $\e_p(S_2)/\VaR_p(X)$, $\CTE_p(S_2)/\VaR_p(X)$ and $\CE_p(S_2)/\VaR_p(X)$, $\MES_p(S_2)/\VaR_p(X)$ and $\ICE_p(S_2)/\VaR_p(X)$, as well as $\SES_p(S_2)/\VaR_p(X)$ and $\SICE_p(S_2)/\VaR_p(X)$. We use the Pareto distribution and the Sarmanov distribution with $a_{12}=-1$ for the left panel, $a_{12}=0$ for the middle panel and $a_{12}=1$ for the right panel. Since the first-order asymptotic values of all risk measures coincide within each plot, they are displayed as a single solid black flat line.}\label{f:2}
\end{figure}

\begin{example}
Under the setting of Example \ref{ex:pareto}  with $\alpha=2$, $k=1$ and $a_{12}=-1,0,1$,  by Theorems \ref{the:3.1}, \ref{the:3.2},  \ref{the:4.1} and \ref{the:4.2}, we have that the first-order asymptotics (as $p\uparrow1$) of $\VaR$, $\e$, $\MES$ and $\ICE$ are equivalent (the ratio to $F^{\leftarrow}(p)$ is $2^{1/2}$). Those of $\CTE$ and $\CE$ are equivalent (the ratio to $F^{\leftarrow}(p)$ is $2^{3/2}$). Those of $\SES$ and $\SICE$ are equivalent (the ratio to $F^{\leftarrow}(p)$ is $2^{1/2}-1/2$). In Figure \ref{f:2}, we have the following observations: 
\begin{itemize}
\item  The second-order asymptotics of systemic risk measures are closer to simulation values than the first-order asymptotics as  $p\uparrow 1$. Except for the comparison between SES and SICE, the expectile-based systemic risk measures usually achieve a higher asymptotic accuracy than their $\VaR$-based alternatives. Expectile-based risk measures are typically calibrated at higher confidence levels than VaR-based risk measures (e.g., $p = 0.99855$, which is comparable to $\VaR_{0.99}$; see \cite{bellini2017risk}). Our results further demonstrate that the second-order asymptotics of the expectile-based systemic risk measures maintain a satisfactory accuracy even at moderately lower levels and improve progressively with increasing levels, thereby supporting their potential practical applicability for regulatory use.

\item {The second-order asymptotics of $\VaR$, $\CTE$ and $\MES$ are closer to the simulation values as $a_{12}$ (i.e., the dependence coefficient) decreases. The second-order asymptotics of expectile and $\CE$ are closer to the simulation values as $a_{12}$ increases.} The numerical evidence does not clearly indicate that subadditivity leads to this asymptotic performance, as expectile ($p \uparrow 1$) and CTE are subadditive risk measures and VaR and CE are not (see \cite{daouia2020tail}).

\item  {The second-order asymptotics of $\ICE$ are closer to the simulation values as $a_{12}$ decreases and those of $\ICE$ are closer to the simulation values than those of $\MES$.}

\item  {The second-order asymptotics of $\SICE$ are  closer to the simulation values as $a_{12}$ decreases and those of $\SES$ are closer to the simulation values than those of $\SICE$.}
\end{itemize}

In addition, we check the asymptotic performance under different dependence structures ranging from independence to comonotonicity (positive dependence) in Table \ref{tab:9}. While $a_{12} = 0$ represents independence, the comonotonicity cannot be presented by a specific value of $a_{12}$. Table \ref{tab:9} shows that the second-order asymptotic results perform well for both independence and comonotonicity, but it is not conclusively established whether a positive dependence would affect asymptotic accuracy. The asymptotic approximation formula under countermonotonicity is currently unavailable and would be an interesting problem left for future research.

\begin{table}[t]%
  \centering
  \setlength{\tabcolsep}{1.0mm}{
  \begin{tabular}{ccccccc }%
  \hline\hline\noalign{\smallskip}
\multicolumn{2}{c}{dependence level} &$\MES_{p,1}(S_2)$ & $\ICE_{p,1}(S_2)$&$\SES_{p,1}(S_2)$ &$\SICE_{p,1}(S_2)$\\
 \noalign{\smallskip}\hline
\multirow{3}{*}{comonotonicity (i.e., $S_2 = X+X$)} 
&$\text{Asy}(X+X)$  & 18.9990 &20.4990&9.9995&10.4995\\
 &$\text{MC}(X+X)$  & 19.0204 & 20.9517&10.0209&10.9904 \\
&$\frac{\text{Asy}(X+X)}{\text{MC}(X+X)}$&  0.9989 &0.9784&0.9979 & 0.9553\\
 \noalign{\smallskip}\hline
 \multirow{3}{*}{independence (i.e., $a_{12}=0$)} 
 &$\text{Asy}(S_2)$& 13.6414& 14.7364&9.1417 &9.7367\\
 &$\text{MC}(S_2)$ &   14.1833& 15.1080&9.1003&9.5411\\
&$\frac{\text{Asy}(S_2)}{\text{MC}(S_2)}$& 0.9618 & 0.9754&1.0045& 1.0205\\
  \noalign{\smallskip}\hline
  \end{tabular}}
  \caption{ Comparison of second-order asymptotic approximations (Asy) and Monte Carlo simulations (MC) for $\MES_{p,1}(S_2)$, $\ICE_{p,1}(S_2)$, $\SES_{p,1}(S_2)$ and $\SICE_{p,1}(S_2)$ at $p=0.99$ under comonotonicity and independence. We use the same setting of Figure \ref{f:2}.}
  \label{tab:9}
\end{table}

\end{example}

Based on the above example, inspired by \cite{koenker1993when} and \cite{bellini2014generalized}, we obtain the following asymptotic results of the Pareto-like distribution. 

\begin{proposition}\label{pro:5.1}
Let $X_1,\dots,X_n$ be nonnegative random variables with a common marginal distribution $F$ satisfying $\overline{F}(\cdot) \in \RV_{-\alpha}$ with $\alpha>1$ (e.g., a Pareto distribution). Suppose that $(X_1,\ldots,X_n)$ follows an $n$-dimensional Sarmanov distribution given by \eqref{eq:sarmanov} and $\lim\limits_{t\rightarrow\infty}\phi_i(t)=d_i\in \R, ~\phi_i(\cdot)-d_i\in \RV_{\rho_i}$ with $\rho_i\leq 0$ for all $i=1,\ldots,n$. Then, as $p\uparrow1$, we have the following asymptotic relationship between VaR-based and expectile-based risk measures: 
\begin{align}\label{eq:5.1}
   \lim_{p\uparrow1}\frac{\e_p(S_n)}{\VaR_p(S_n)}=\lim_{p\uparrow1}\frac{\CE_p(S_n)}{\CTE_p(S_n)}=\lim_{p\uparrow1}\frac{\ICE_p(S_n)}{\MES_p(S_n)} = \lim_{p\uparrow1}\frac{\SICE_p(S_n)}{\SES_p(S_n)}=(\alpha-1)^{-1/\alpha}.
\end{align}
Hence, we have: 
\begin{enumerate}
    \item If $\alpha=2$, then VaR-based risk measures are asymptotically equal to expectile-based risk measures as $p \uparrow 1$ (i.e., $\VaR_p(S_n)\sim\e_p(S_n),~ \CTE_p(S_n)\sim\CE_p(S_n),~ \MES_p(S_n) \sim \ICE_p(S_n),~ \SES_p(S_n)\sim \SICE_p(S_n)$ as $p\uparrow1$).
    \item  If $\alpha<2$, then 
    VaR-based risk measures are asymptotically less than expectile-based risk measures as $p \uparrow 1$ (i.e., $\VaR_p(S_n)<\e_p(S_n),~ \CTE_p(S_n) < \CE_p(S_n),~ \MES_p(S_n) < \ICE_p(S_n),~ \SES_p(S_n) < \SICE_p(S_n)$ as $p \uparrow 1$).
    \item If $\alpha>2$,  then 
    VaR-based risk measures are asymptotically greater than expectile-based risk measures as $p \uparrow 1$ (i.e., $\VaR_p(S_n)>\e_p(S_n),~ \CTE_p(S_n) > \CE_p(S_n),~ \MES_p(S_n) >\ICE_p(S_n),~ \SES_p(S_n) >\SICE_p(S_n)$ as $p\uparrow1$). 
\end{enumerate}
  
\end{proposition}
\begin{proof}
According to Theorems \ref{the:3.1}-\ref{the:3.2}, as $p\uparrow 1$, it follows that the first-order asymptotics of VaR-based risk measures are 
$$
\begin{aligned}
&\VaR_p(S_n)\sim n^{1/\alpha}F^{\leftarrow}(p),
~~~~~~~~~~~&&\CTE_p(S_n)\sim\frac{\alpha n^{1/\alpha}}{\alpha-1}F^{\leftarrow}(p),\\&\MES_p(S_n)\sim\frac{\alpha n^{1/\alpha}F^\leftarrow(p)}{n(\alpha-1)},
~~~~&&\SES_p(S_n)\sim\frac{\alpha n^{1/\alpha}F^\leftarrow(p)}{n(\alpha-1)}-\frac{F^\leftarrow(p)}{n}.
\end{aligned}
$$
By Theorems \ref{the:4.1}-\ref{the:4.2}, as $p\uparrow 1$, it follows that the first-order asymptotics of expectile-based risk measures are
$$
\begin{aligned}
&\e_p(S_n)\sim\frac{n^{1/\alpha}F^{\leftarrow}(p)}{(\alpha-1)^{1/\alpha}}, ~~~~~~~~~~~~&&\CE_p(S_n)\sim\frac{\alpha n^{1/\alpha}F^{\leftarrow}(p)}{(\alpha-1)^{1/\alpha+1}},\\
&\ICE_{p,m}(S_n)\sim\frac{\alpha n^{1/\alpha} F^\leftarrow(p) }{n(\alpha-1)^{1/\alpha+1}},~~~~ &&\SICE_{p,m}(S_n)\sim\frac{\alpha n^{1/\alpha} F^\leftarrow(p) }{n(\alpha-1)^{1/\alpha+1}}- \frac{F^\leftarrow(p)}{n(\alpha-1)^{1/\alpha}}.
\end{aligned}
$$
Thus, \eqref{eq:5.1} holds. It is easy to see that $(\alpha-1)^{-1/\alpha} <  1$ if $\alpha>2$ and $(\alpha-1)^{-1/\alpha} > 1$ if $\alpha<2$. This completes the proof.
\end{proof}

\begin{example} 
Under the conditions of Example \ref{ex:burr} with $\alpha=1.5$, $\beta=-0.5$ or $-2$ under 
a Sarmanov distribution in \eqref{eq:sarmanov} with $\phi_i(\cdot)=1-2F(\cdot), ~i=1,2$ and $a_{12}=-0.5$, by Theorems \ref{the:3.1}, \ref{the:3.2},  \ref{the:4.1} and \ref{the:4.2}, we have the following facts.
\begin{itemize}
    \item Figures \ref{f:3}-\ref{f:6} show the second-order asymptotics and simulation values of Theorems \ref{the:3.1}, \ref{the:3.2},  \ref{the:4.1} and \ref{the:4.2}. Again, the second-order asymptotics can approximate the simulation values as $p\uparrow 1$ better than the first-order asymptotics. 
    
    \item  {Figures \ref{f:3}-\ref{f:6} reveal that almost all the second-order asymptotics of $\VaR$, $\e$, $\CTE$, $\CE$, $\MES$, $\ICE$, $\SES$ and $\SICE$ with $\beta=-0.5$ are more accurate than those with $\beta=-2$.} 

\end{itemize}

\end{example}

\begin{figure}[htpb]
  \begin{minipage}[t]{0.5\textwidth}
    \centering
    \includegraphics[width=\textwidth]{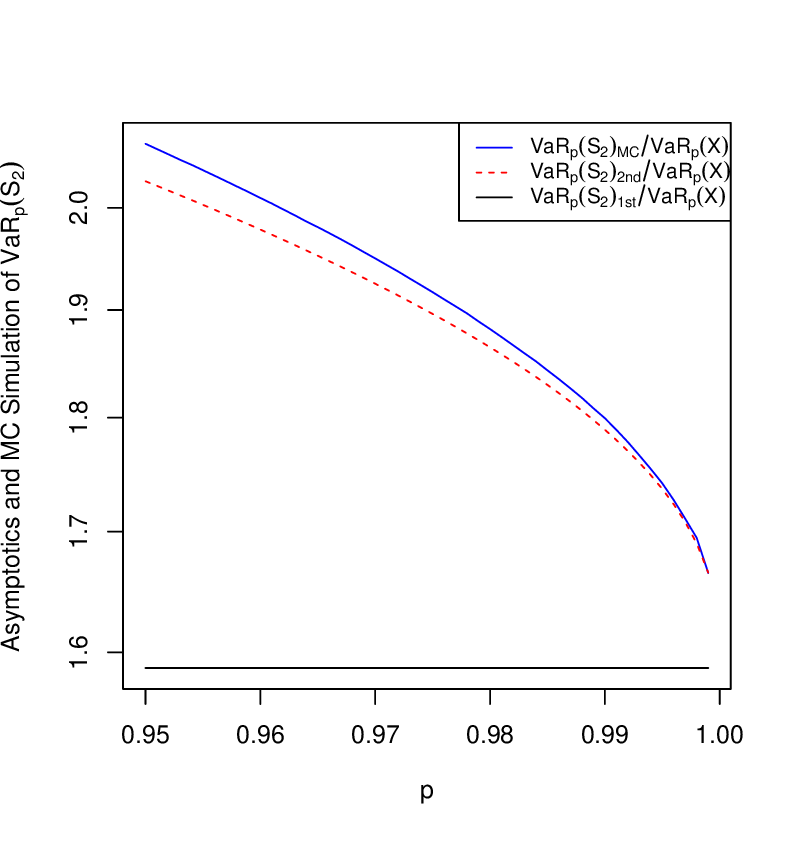}
  \end{minipage}%
  \begin{minipage}[t]{0.5\textwidth}
    \centering
    \includegraphics[width=\textwidth]{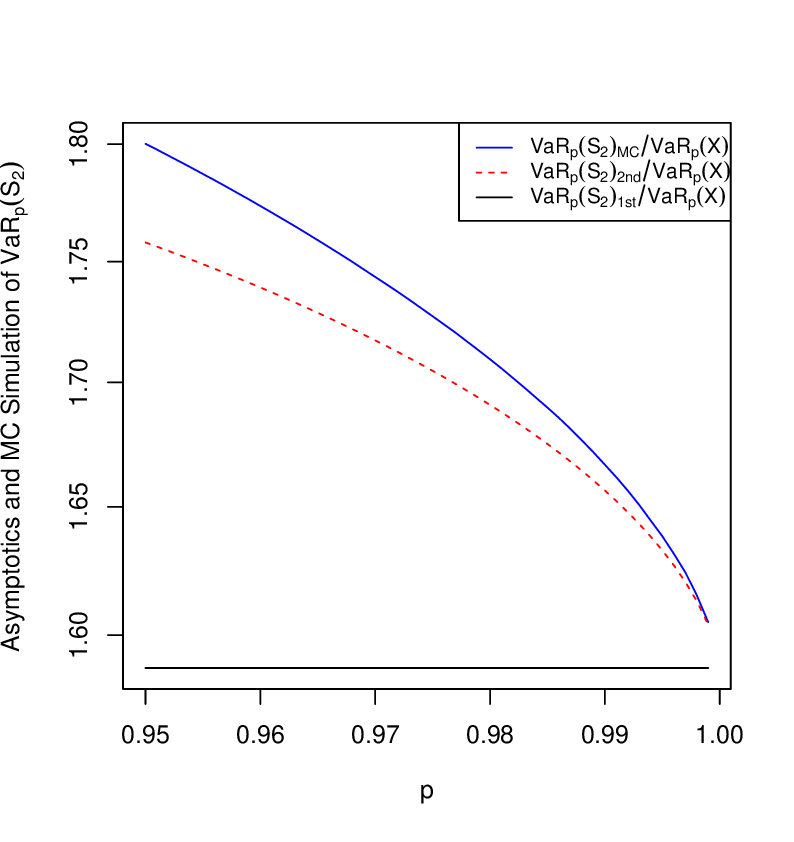}
  \end{minipage}
  \begin{minipage}[t]{0.5\textwidth}
    \centering
    \includegraphics[width=\textwidth]{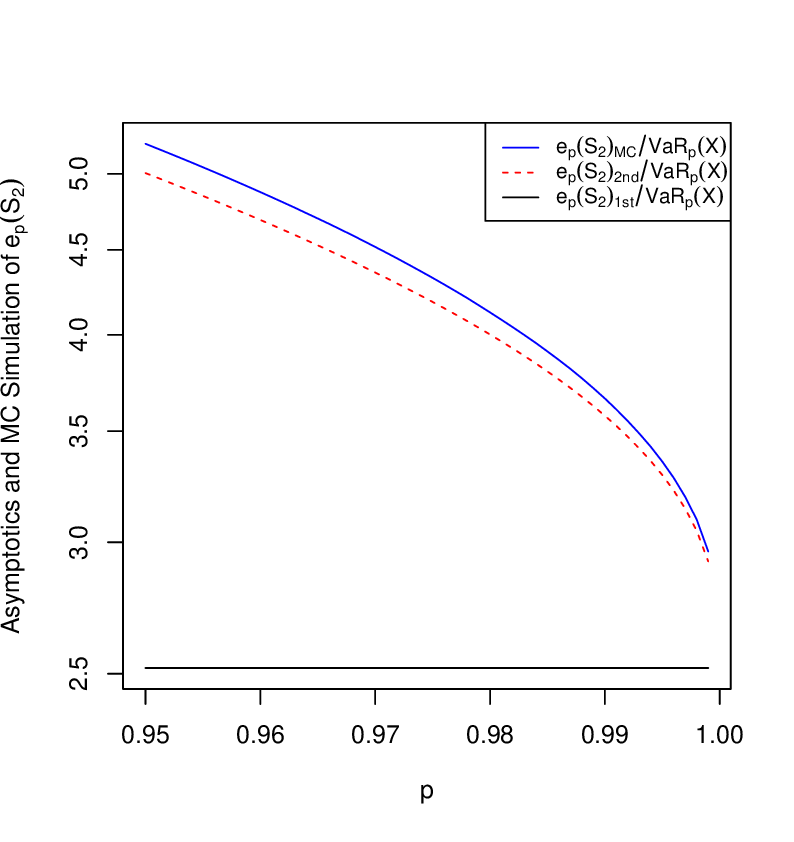}
  \end{minipage}%
  \begin{minipage}[t]{0.5\textwidth}
    \centering
    \includegraphics[width=\textwidth]{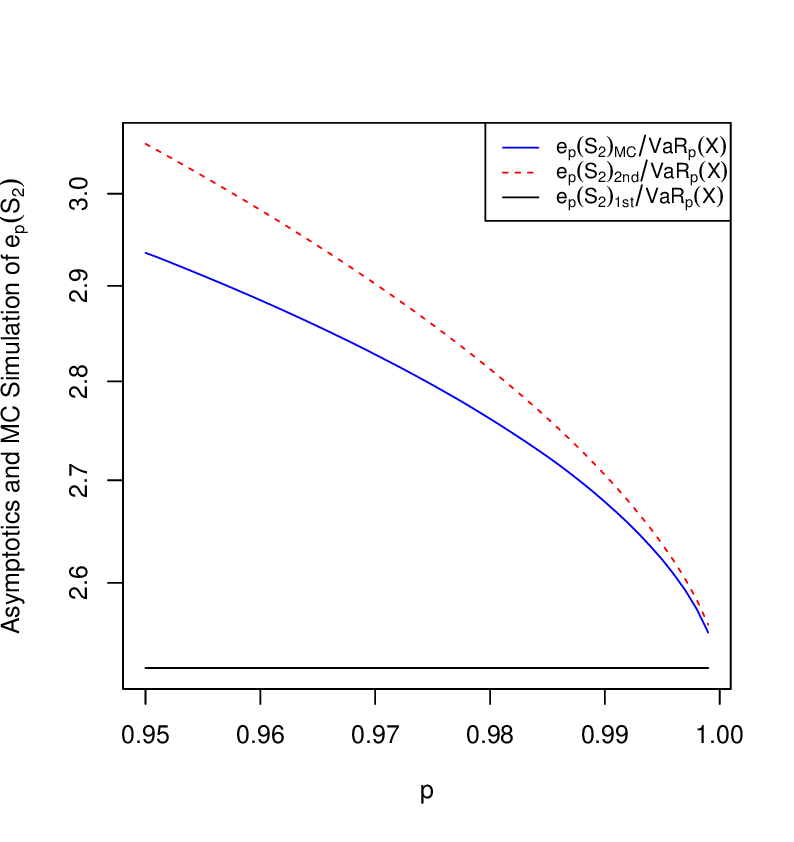}
  \end{minipage}
   \caption{Comparison of the simulation values (MC), first-order and second-order asymptotic values of $\VaR_p(S_2)/\VaR_p(X)$ and $\e_p(S_2)/\VaR_p(X)$. We use the Burr distribution with $\beta= - 0.5$ for the left panel and $\beta= -2$ for the right panel.}\label{f:3}
\end{figure}

\begin{figure}[thpb]
  \begin{minipage}[t]{0.5\textwidth}
    \centering
    \includegraphics[width=\textwidth]{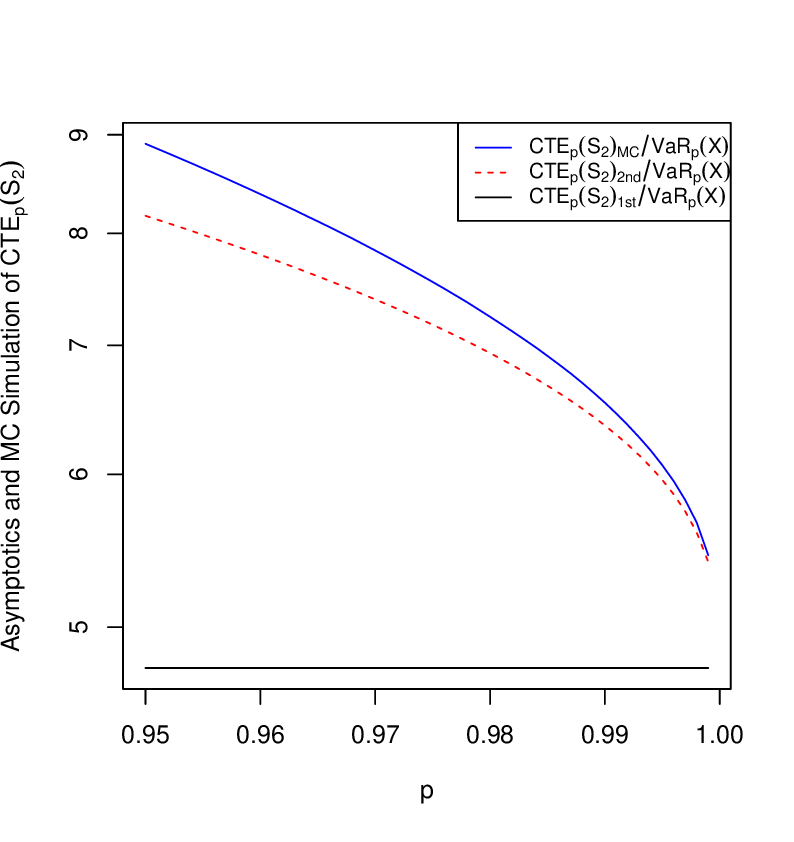}
  \end{minipage}%
  \begin{minipage}[t]{0.5\textwidth}
    \centering
    \includegraphics[width=\textwidth]{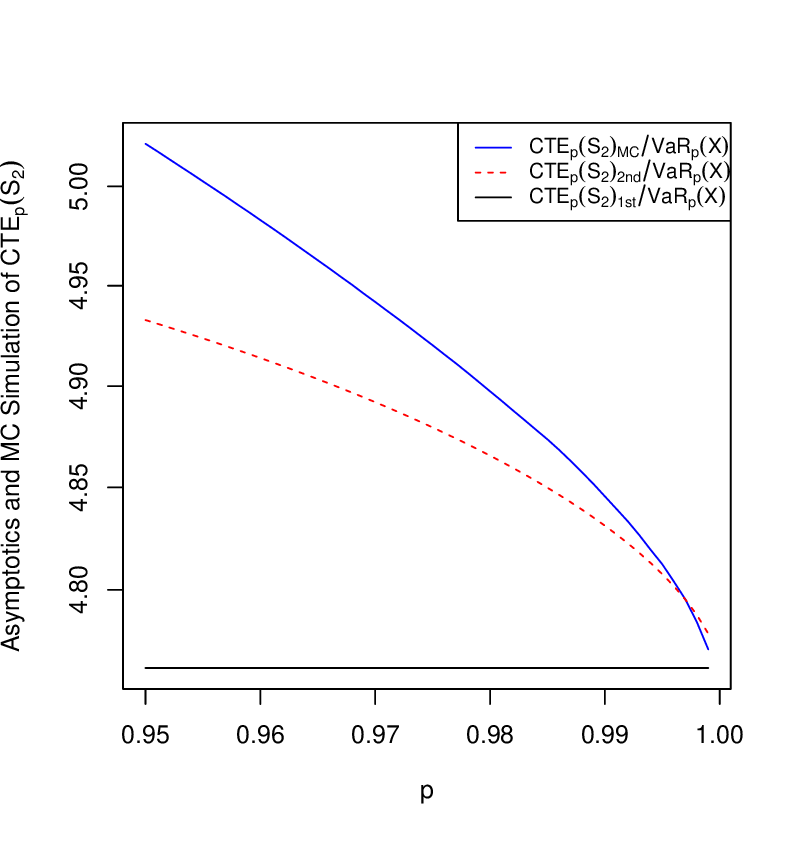}
  \end{minipage}
  \begin{minipage}[t]{0.5\textwidth}
    \centering
    \includegraphics[width=\textwidth]{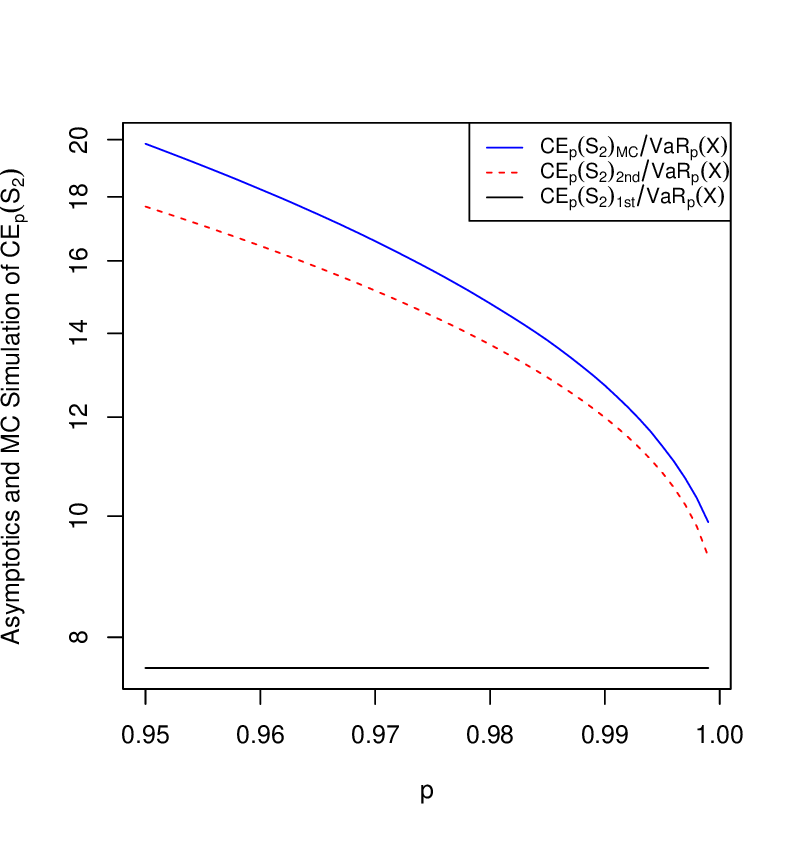}
  \end{minipage}%
  \begin{minipage}[t]{0.5\textwidth}
    \centering
    \includegraphics[width=\textwidth]{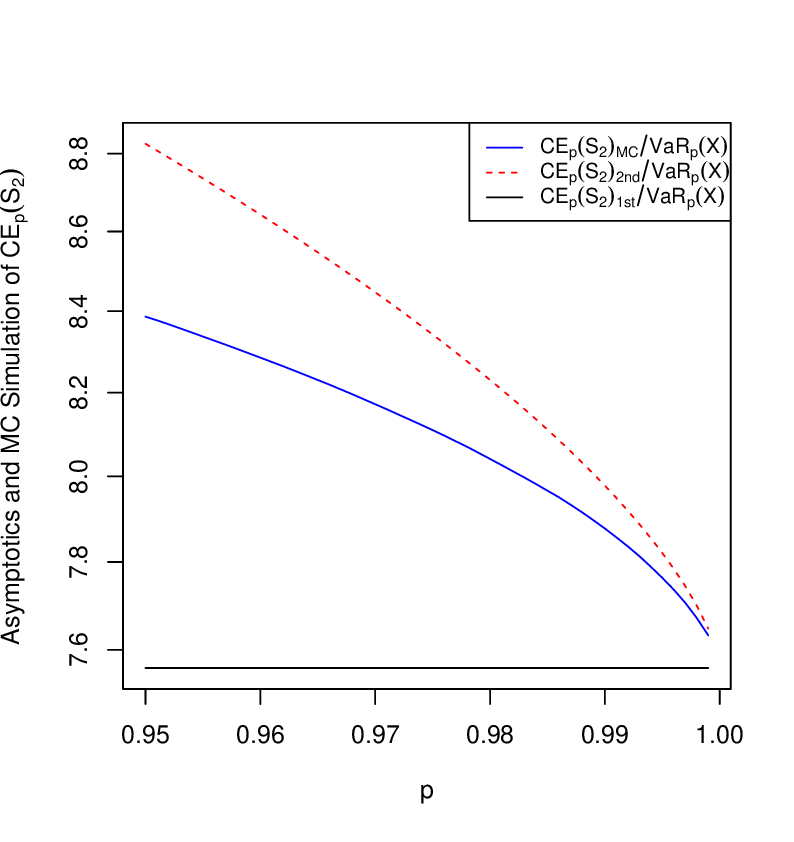}
  \end{minipage}
   \caption{Comparison of the simulation values (MC), first-order and second-order asymptotic values of  $\CTE_p(S_2)/\VaR_p(X)$ and $\CE_p(S_2)/\VaR_p(X)$. We use the Burr distribution with $\beta= - 0.5$ for the left panel and  $\beta= - 2$ for the right panel.}\label{f:4}
\end{figure}

\begin{figure}[thpb]
  \begin{minipage}[t]{0.5\textwidth}
    \centering
    \includegraphics[width=\textwidth]{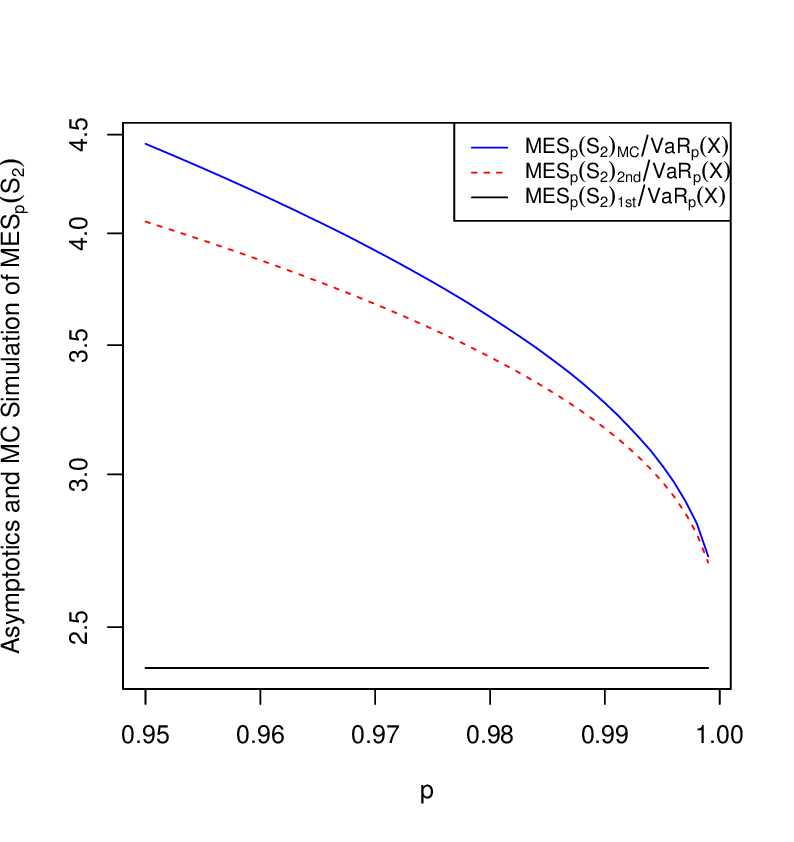}
  \end{minipage}%
  \begin{minipage}[t]{0.5\textwidth}
    \centering
    \includegraphics[width=\textwidth]{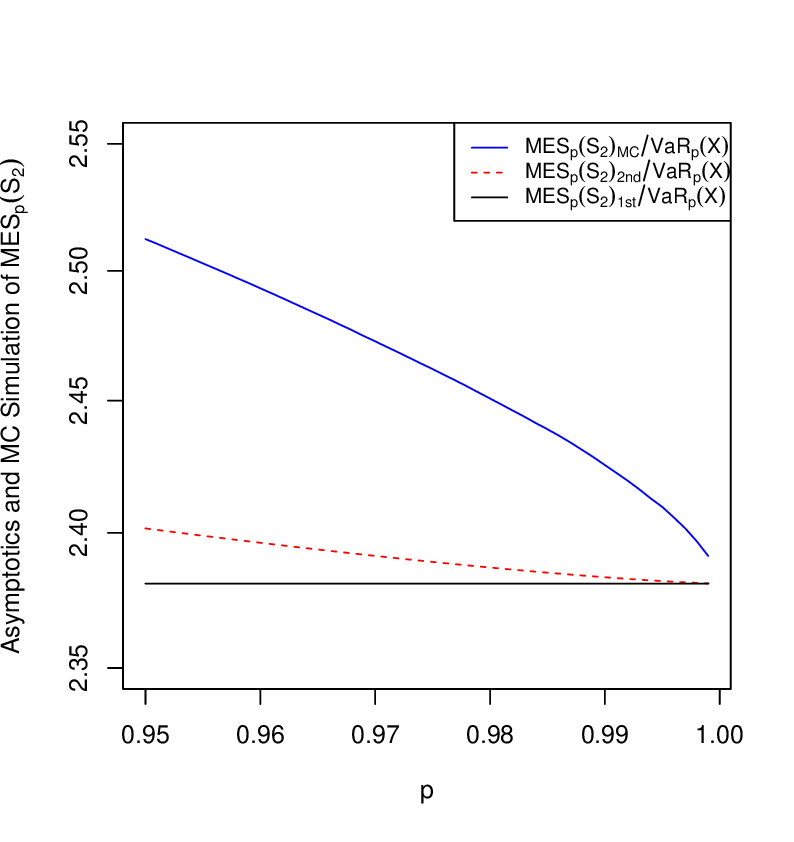}
  \end{minipage}
  \begin{minipage}[t]{0.5\textwidth}
    \centering
    \includegraphics[width=\textwidth]{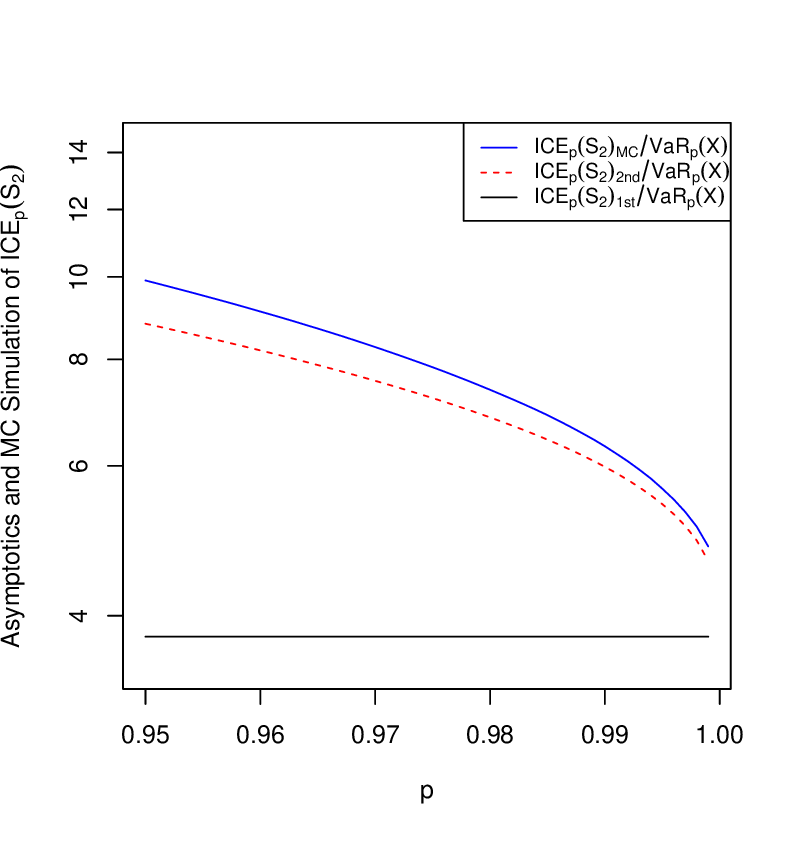}
  \end{minipage}%
  \begin{minipage}[t]{0.5\textwidth}
    \centering
    \includegraphics[width=\textwidth]{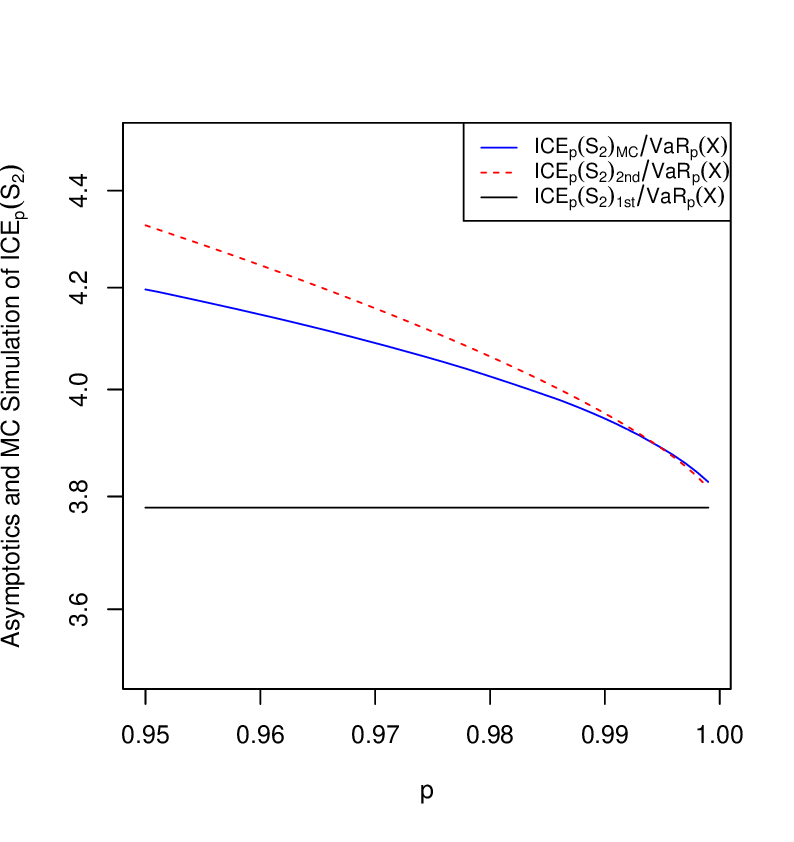}
  \end{minipage}
 \caption{Comparison of the simulation values (MC), first-order and second-order asymptotic values of  $\MES_p(S_2)/\VaR_p(X)$ and $\ICE_p(S_2)/\VaR_p(X)$. We use the Burr distribution with $\beta= - 0.5$ for the left panel and  $\beta= - 2$ for the right panel.}\label{f:5}
\end{figure}

\begin{figure}[thpb]
  \begin{minipage}[t]{0.5\textwidth}
    \centering
    \includegraphics[width=\textwidth]{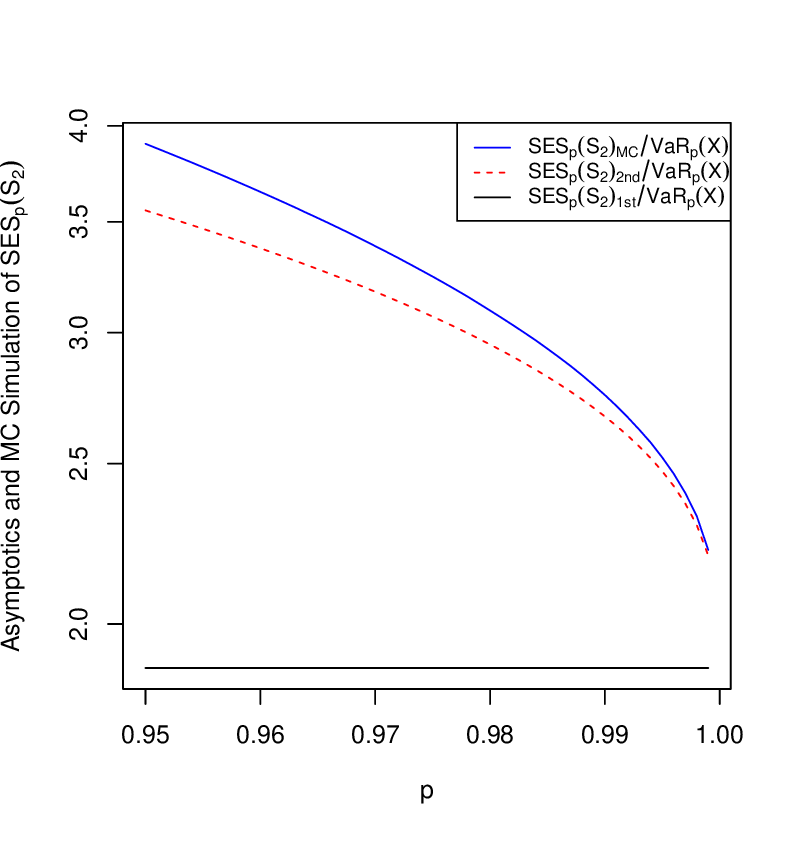}
  \end{minipage}%
  \begin{minipage}[t]{0.5\textwidth}
    \centering
    \includegraphics[width=\textwidth]{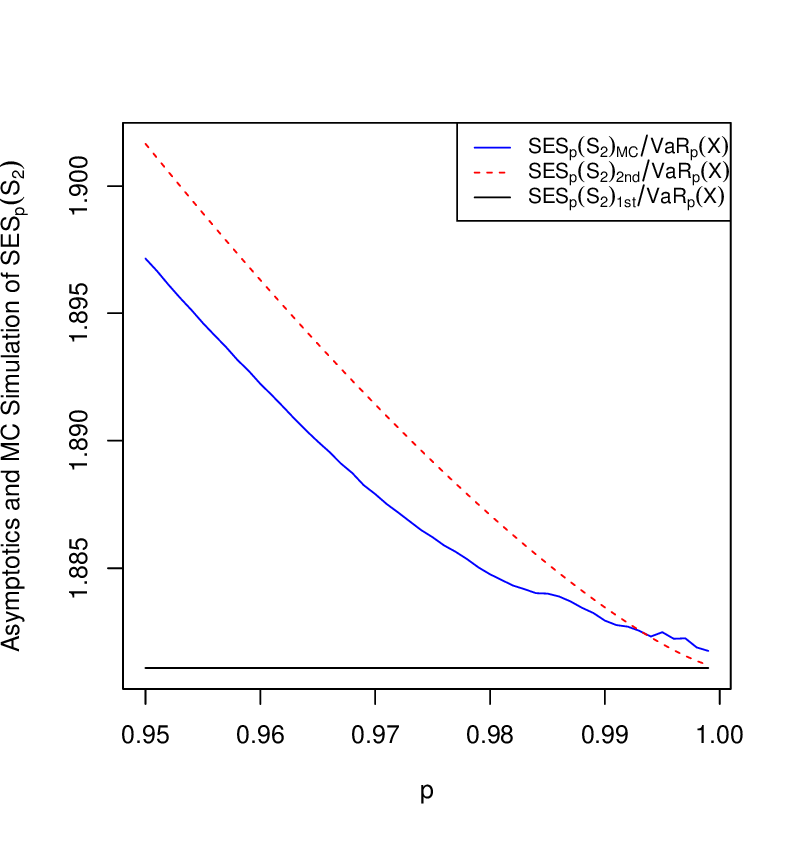}
  \end{minipage}
  \begin{minipage}[t]{0.5\textwidth}
    \centering

    \includegraphics[width=\textwidth]{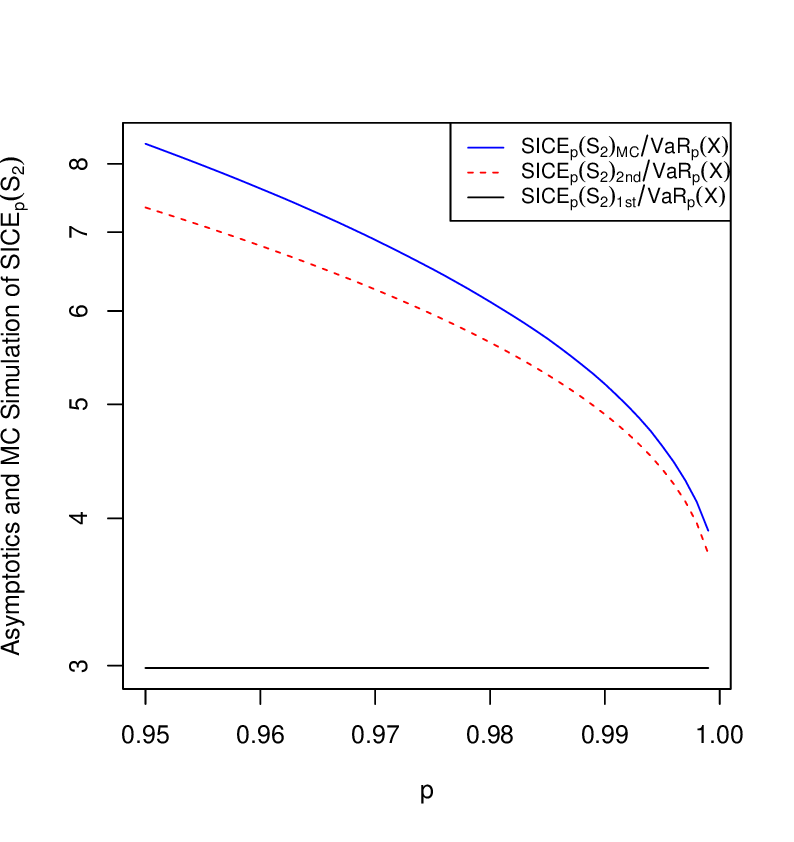}
  \end{minipage}%
  \begin{minipage}[t]{0.5\textwidth}
    \centering
    \includegraphics[width=\textwidth]{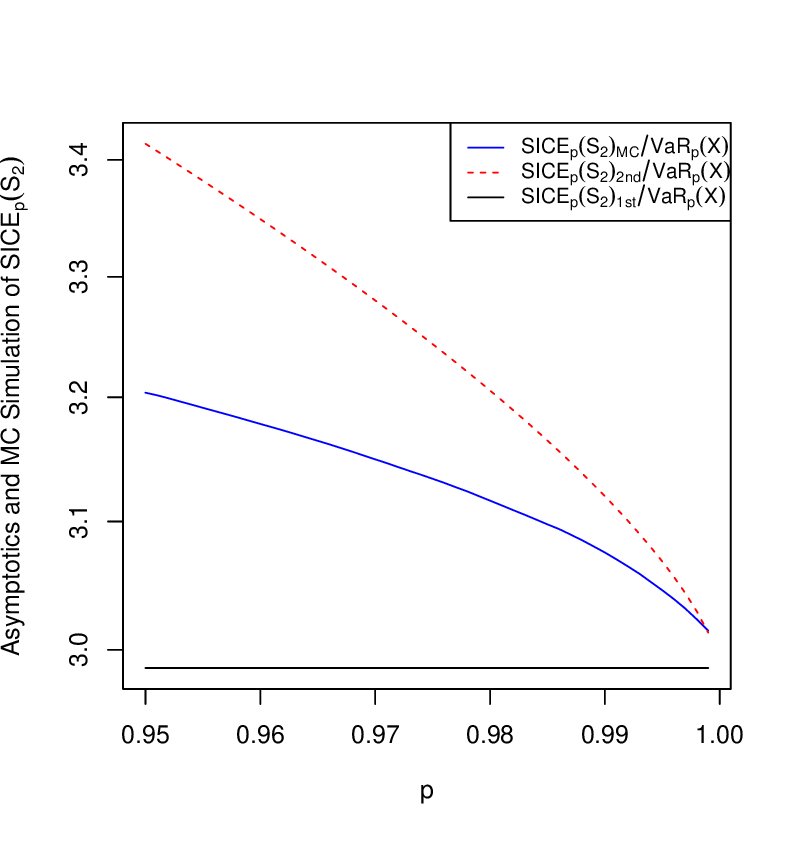}
  \end{minipage}
  \caption{Comparison of the simulation values (MC), first-order and second-order asymptotic values of  $\SES_p(S_2)/\VaR_p(X)$ and $\SICE_p(S_2)/\VaR_p(X)$. We use the Burr distribution with $\beta= - 0.5$ for the left panel and  $\beta= - 2$ for the right panel.}\label{f:6}
\end{figure}

\section{Application}\label{sec:app}

 {The concept of the \textit{diversification benefit} (see \cite{mcneil2015quantitative}) represents the retained capital gained by collectively managing all risks within a portfolio, in contrast to addressing each risk individually. Formally, for a fixed threshold of $0<p<1$, the diversification benefit is defined by}
\begin{align}\label{eq:D}
   D_p^{\rho}(S_n)=1-\frac{\rho_p(S_n)-\E[S_n]}{\sum\limits_{i=1}^n(\rho_p(X_i)-\E[X_i])},
\end{align}
where $\rho_p$ represents a risk measure at a specific level $p$ (e.g., $\VaR_p$, $\e_p$, $\CTE_p$ and $\CE_p$).   

In this framework, for a fixed risk measure $\rho$, $D_p^{\rho} > 0$ indicates that diversification is advantageous, potentially reducing an insurer's risk by market engagement. Conversely, $D_p^{\rho} \leq 0$ suggests that diversification is not advantageous for a single insurer. 
Further, this aspect of whether $D_p^{\rho} > 0$ is technically linked to the risk measure $\rho$'s subadditivity. It further relates to the so-called coherence; e.g., see \cite{artzner1999coherent}. 

The diversification benefit 
aids in portfolio selection. By maximizing diversification benefits, the investor can mitigate the risk and boost the performance of a portfolio. It is worthwhile to mention that the usage of $D_p^{\rho}(S_n)$ is not always applicable; its value depends on the number of risks involved and the specific risk measures employed. Recent findings from \cite{dacorogna2018validation} and \cite{chen2022ordering} highlighted that diversification benefits vary notably based on the type of dependence and the risk measures. 

Experts emphasize caution against careless diversification practices, especially when confronted with heavy-tailed risks. By adopting the above results of risk measures and deriving formulas for diversification benefits, we can evaluate the performance of a portfolio $S_n$ in contrast to individual risks operating independently. In the following, we first derive the second-order asymptotics of $D_p^{\rho}(S_n)$ with $\rho$ based on $\VaR$, $\e$, $\CTE$ and $\CE$. 
\begin{theorem}\label{The:D}
 Under the conditions of Proposition  \ref{pro:sum} with $\alpha>1$, we have, as $p\uparrow1$,
     \begin{align*}
     D_p^{\VaR}(S_n)&=1-n^{1/\alpha-1}\(1+\frac{n^{\beta/\alpha}-1}{\alpha\beta}A\(F^{\leftarrow}(p)\)\(1+o(1)\)\right.\\
     &\quad\left.+\frac{\mu_n^*\(F^{\leftarrow}(p)\)-(n-n^{1/\alpha})\mu}{n^{1/\alpha} F^{\leftarrow}(p)}\(1+o(1)\)\)+o\(\sum\limits_{i=1}^n|\phi_i\(F^{\leftarrow}(p)\)-d_i|\),
\end{align*}
\begin{align*}
 D_p^{\e}(S_n)&=1-n^{1/\alpha-1}\(1+\frac{\(n^{\beta/\alpha}-1\)(\alpha-1)^{1-\beta/\alpha}}{\alpha\beta(\alpha-\beta-1)}A\(F^{\leftarrow}(p)\)\(1+o(1)\)\right.\\
 &\quad\left. +\frac{(\alpha-1)^{1/\alpha+1}\(\mu_n^*\(F^{\leftarrow}(p)\)-(n-n^{1/\alpha})\mu\)}{\alpha n^{1/\alpha}  F^{\leftarrow}(p)}\(1+o(1)\)\)+o\(\sum\limits_{i=1}^n|\phi_i\(F^{\leftarrow}(p)\)-d_i|\),
\end{align*}
\begin{align*}
 D_p^{\CTE}(S_n)
 &=1-n^{1/\alpha-1}\(1+\frac{1}{\alpha\beta}\(\frac{n^{\beta/\alpha}(\alpha-1)-\beta}{\alpha-\beta-1}-1\)A\(F^{\leftarrow}(p)\)\(1+o(1)\)\right.\\
   &\quad\left.+\frac{(\alpha-1)\(\mu_n^*\(F^{\leftarrow}(p)\)-\(n-n^{1/\alpha}\)\mu\)}{\alpha n^{1/\alpha}  F^{\leftarrow}(p)}\(1+o(1)\)\)+o\(\sum\limits_{i=1}^n|\phi_i\(F^{\leftarrow}(p)\)-d_i|\),
\end{align*}
and
\begin{align*}
 D_p^{\CE}(S_n)
 &=1-n^{1/\alpha-1}\(1+\frac{\(n^{\beta/\alpha}-1\)(\alpha-1)^{-\beta/\alpha}\(\alpha+\beta-1\)}{\alpha\beta(\alpha-\beta-1)}A\(F^{\leftarrow}(p)\)\(1+o(1)\)\right.\\
 &\quad\left. +\frac{(\alpha-1)^{1/\alpha}(\alpha-2)\(\mu_n^*\(F^{\leftarrow}(p)\)-\(n-n^{1/\alpha}\)\mu\)}{\alpha n^{1/\alpha}  F^{\leftarrow}(p) }\(1+o(1)\)\)+o\(\sum\limits_{i=1}^n|\phi_i\(F^{\leftarrow}(p)\)-d_i|\).
\end{align*}    
Clearly, the first-order asymptotics of $D^{\VaR}_p(S_n)$, $D^{\e}_p(S_n)$, $D^{\CTE}_p(S_n)$ and $D^{\CE}_p(S_n)$ are $1-n^{1/\alpha-1}$.
 \end{theorem}
\begin{proof}
See the Appendix.
\end{proof}

For a numerical illustration, we give an example of the Weiss distribution.
\begin{example}
\label{ex:weiss}
(Weiss distribution) A Weiss distribution function $F$ satisfies that
$$F(x)=1-\frac{1}{2}x^{-\alpha}\(1+x^{\beta}\),~~~~ x\geq 1, $$
with parameters $\alpha>1, \beta<0$. It is easy to check that $\overline{F}(\cdot)\in 2\RV_{-\alpha,\beta}$  with an auxiliary function $A(t)=\beta t^{\beta}$. Let $X_1$ and $X_2$ have an identical Weiss distribution $F$. Suppose that the random vector $(X_1,X_2)$ follows a Sarmanov distribution in \eqref{eq:sarmanov} with $\phi_i(\cdot)=1-2F(\cdot)$. Clearly, $d_i=-1, ~i=1,2$. 

\end{example}
\begin{figure}[htbp]%
  \begin{minipage}[t]{0.5\textwidth}
    \centering
    \includegraphics[width=\textwidth]{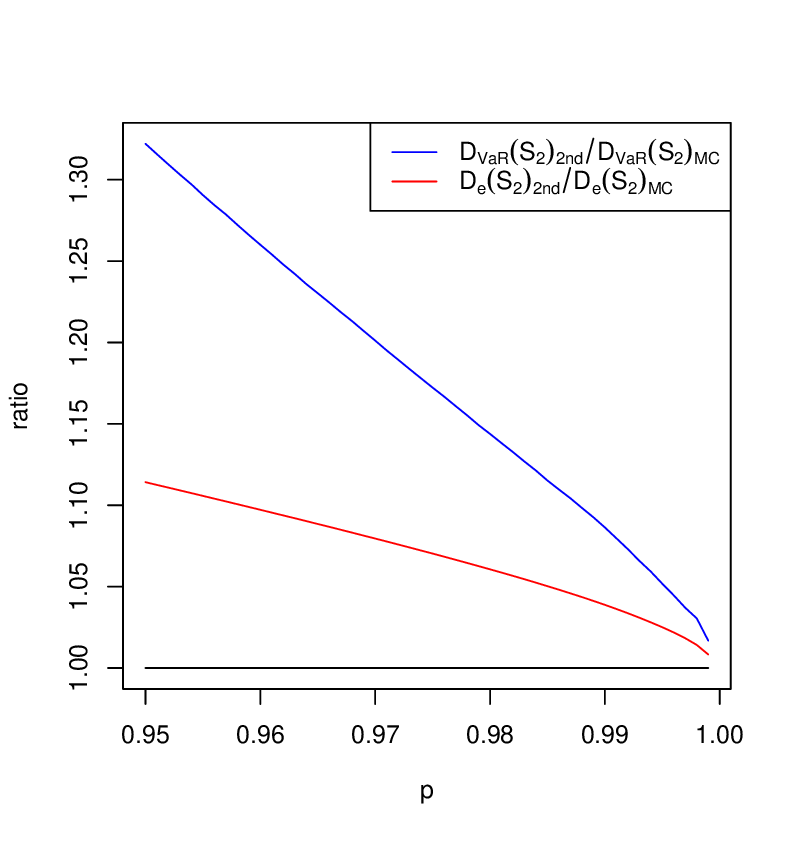}
  \end{minipage}%
  \begin{minipage}[t]{0.5\textwidth}
    \centering
    \includegraphics[width=\textwidth]{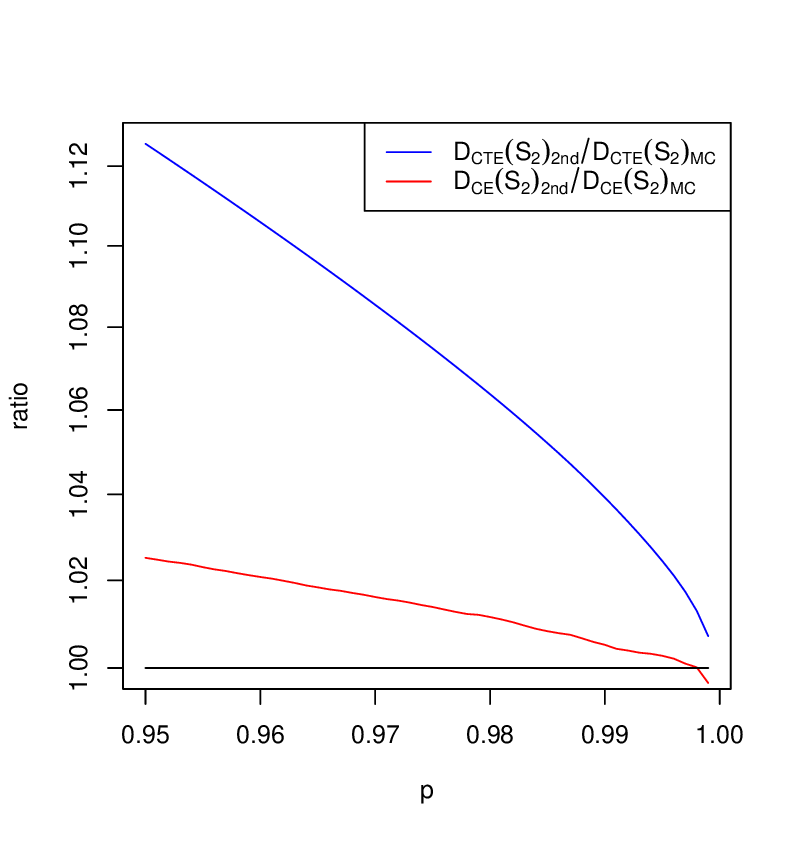}
  \end{minipage}%
 \caption{Comparison of the simulation values (MC), first-order and second-order asymptotic values of $D_{\VaR}(S_2)$, $D_{\e}(S_2)$, $D_{\CTE}(S_2)$ and $D_{\CE}(S_2)$. We use the Weiss distribution with $\alpha=2.5, ~\beta=-1$ and $a_{12}=0.5$. }\label{f:8}
\end{figure}
 {
\begin{table}[t]%
   \fontsize{8}{2}\centering
  \setlength{\tabcolsep}{0.8mm}{
  \begin{tabular}{cccc cccccc}%
  \hline\hline\noalign{\smallskip}
  $n$ & $D_{\rho}(S_n)_{1st}$& $D_{\VaR}(S_n)_{MC}$& $D_{\VaR}(S_n)_{2nd}$  & $D_{\e}(S_n)_{MC}$& $D_{\e}(S_n)_{2nd}$&$\frac{D_{\rho}(S_n)_{1st}}{D_{\VaR}(S_n)_{MC}}$ &$\frac{D_{\VaR}(S_n)_{2nd}}{D_{\VaR}(S_n)_{MC}}$ &$\frac{D_{\rho}(S_n)_{1st}}{D_{\e}(S_n)_{MC}}$ &$\frac{D_{\e}(S_n)_{2nd}}{D_{\e}(S_n)_{MC}}$ \\
\noalign{\smallskip}\hline\noalign{\smallskip}
 2& 0.3402& 0.2817 &0.3162& 0.3065 &0.3243& 1.2079 &1.1225 &1.1101& 1.0581\\
3& 0.4827& 0.4091& 0.4572 &0.4399& 0.4647& 1.1801& 1.1177& 1.0973& 1.0563\\
 4& 0.5647& 0.4877& 0.5409& 0.5200& 0.5479& 1.1580& 1.1092& 1.0860& 1.0537		 \\
 5 &0.6193& 0.5418& 0.5977 &0.5742& 0.6041& 1.1429& 1.1031 &1.0785& 1.0520\\
  \noalign{\smallskip}\hline
  \end{tabular}}
  \caption{Simulation values $D_{\VaR}(S_n)_{MC},D_{\e}(S_n)_{MC}$ versus the first-order asymptotic values $D_{\rho}(S_n)_{1st}$ (note that $D_{\rho}(S_n)_{1st}$ is the same for $\VaR$ and $\e$) and the second-order asymptotic values $D_{\VaR}(S_n)_{2nd},D_{\e}(S_n)_{2nd}$. We use the Weiss distribution with $\alpha=2.5,\beta=-1,a_{ij}=-0.1$ for all $1 \leq i< j\leq n$ and $p=0.99$.}
  \label{tab:6}
\end{table}
						
	\begin{table}[t]%
   \fontsize{8}{2}\centering
  \setlength{\tabcolsep}{0.8mm}{
  \begin{tabular}{cccc cccccc}%
  \hline\hline\noalign{\smallskip}
  $n$ & $D_{\rho}(S_n)_{1st}$& $D_{\CTE}(S_n)_{MC}$& $D_{\CTE}(S_n)_{2nd}$  & $D_{\CE}(S_n)_{MC}$& $D_{\CE}(S_n)_{2nd}$&$\frac{D_{\rho}(S_n)_{1st}}{D_{\CTE}(S_n)_{MC}}$ &$\frac{D_{\CTE}(S_n)_{2nd}}{D_{\CTE}(S_n)_{MC}}$ &$\frac{D_{\rho}(S_n)_{1st}}{D_{\CE}(S_n)_{MC}}$ &$\frac{D_{\CE}(S_n)_{2nd}}{D_{\CE}(S_n)_{MC}}$ \\
\noalign{\smallskip}\hline\noalign{\smallskip}
 2& 0.3402& 0.3085& 0.3258& 0.3307& 0.3346 &1.1029 &1.0561& 1.0289& 1.0117
\\
3& 0.4827& 0.4430& 0.4674& 0.4704& 0.4767 &1.0897& 1.0551 &1.0262& 1.0135
\\
 4 &0.5647& 0.5235& 0.5504& 0.5524& 0.5591& 1.0788& 1.0515& 1.0224& 1.0123
		 \\
5& 0.6193& 0.5779& 0.6063& 0.6067& 0.6140& 1.0715& 1.0491& 1.0206& 1.0120\\
  \noalign{\smallskip}\hline
  \end{tabular}}
  \caption{Simulation values $D_{\CTE}(S_n)_{MC},D_{\CE}(S_n)_{MC}$ versus the first-order asymptotic values $D_{\rho}(S_n)_{1st}$ (note that $D_{\rho}(S_n)_{1st}$ is the same for $\CTE$ and $\CE$) and the second-order asymptotic values $D_{\CTE}(S_n)_{2nd}$ and $D_{\CE}(S_n)_{2nd}$. We use the Weiss distribution with $\alpha=2.5,\beta=-1,a_{ij}=-0.1$ for all $1 \leq i< j\leq n$ and $p=0.99$.}
  \label{tab:7}
\end{table}	}

In the context of Example \ref{ex:weiss}, we aim to present the asymptotic performance for diversification benefits $D_p^{\rho}$ across four risk measures (i.e., VaR, expectile, CTE, and CE) based on the results obtained in Theorem \ref{The:D}. Here, we denote by $\Hat{D_p^{\rho}}$ the second-order asymptotic and employ the ratio $\Hat{D_p^{\rho}}/D_p^{\rho}$ to assess the asymptotic performance of diversification benefits. A value of $\Hat{D_p^{\rho}}/D_p^{\rho}$ closer to 1 indicates a more accurate asymptotic result, while deviations from 1 imply poorer outcomes. Besides, according to \cite{mcneil2015quantitative}, $\Hat{D_p^{\rho}}/D_p^{\rho} > 1$ signifies the overestimation of the diversification benefit.

The numerical experiment is shown in Figure \ref{f:8}. 
For comparison purposes, we use blue lines to represent the outcomes derived from VaR-based risk measures (i.e., VaR and CTE), while red lines represent the outcomes obtained from expectile-based risk measures (i.e., expectile and CE).  {Tables \ref{tab:6}-\ref{tab:7} show that as $n$ increases, all diversification benefits become larger and hence diversification is more advantageous as more risks are incorporated. 
Additionally, the numerical results show that the asymptotic performance of $\hat{D}^{\rho}_p / D^{\rho}_p$ improves as $n$ increases. 
Consistent with our earlier findings, expectile-based risk measures continue to provide more accurate asymptotic approximations than VaR-based ones, and second-order asymptotic approximations generally outperform first-order ones. 
}



 {
New methods to quantify diversification continue to emerge, such as the diversification quotient \cite{han2023diversification,han2022diversification}. Our asymptotic treatment provides a unified framework to investigate these new quotients, which will be studied in the future.
}

\section{Generalized-quantile-based systemic risk measures}\label{sec:7}

We discuss some related results in a more general framework of generalized quantiles in \cite{bellini2014generalized}. The general quantile is a risk measure defined as 
$$\psi_p(X) = \inf_{x\in\R}\Big\{
 p\E[\Phi_1((X-x)_+)]+(1-p)\E[\Phi_2((X-x)_-)]\Big\}, ~~ p \in (0, 1).
$$
Particularly, let $\Phi_1(t)=\Phi_2(t)=t^k$. If $k=1$, $\psi_p$ becomes the $\VaR_p$. If $k = 2$, $\psi_p$ becomes the expectile $\e_p$. We obtain the second-order asymptotic results of generalized quantiles.  

\begin{proposition}\label{pro:7}
    Assume that $X \sim F$ and $\overline{F}(\cdot)\in 2\RV_{-\alpha,\beta}$ with an auxiliary function $A(\cdot)$, $\alpha>1$ and $\beta\leq0$. Let $\Phi_1(t)=\Phi_2(t)=t^k$ for some $1< k<\alpha$ and $k\in \N$. Then, as $p\uparrow1$, we have
\begin{align*}
    \psi_p(X)=L_{\alpha,k}F^{\leftarrow}(p)\(1+\(\frac{\Delta_{\alpha,\beta,k}-1}{\alpha\beta}A(F^{\leftarrow}(p))+\frac{(k-1)\mu}{\alpha L_{\alpha,k}F^{\leftarrow}(p)}\)(1+o(1))\),
\end{align*}
where $$B_{\alpha,k}:=\int_0^1y^{\alpha-k}(1-y)^{k-2}\d y, ~~L_{\alpha,k}:=\((k-1)B_{\alpha,k}\)^{1/\alpha},~~\mbox{and}~~\Delta_{\alpha,\beta,k}:=L_{\alpha,k}^{\beta}B_{\alpha-\beta,k}B_{\alpha,k}^{-1}.$$  Clearly, the first-order asymptotic is given by $L_{\alpha,k}F^{\leftarrow}(p)$.
\end{proposition}

\begin{proof}
See the Appendix.
\end{proof}

In particular, for $k=2$, Proposition \ref{pro:7} reduces to the result in Proposition \ref{pro:e}:  
$$
\e_p(X)=\(\alpha-1\)^{-1/\alpha}F^{\leftarrow}(p)\(1+\(\frac{1}{\alpha\beta}\(\frac{(\alpha-1)^{1-\beta/\alpha}}{\alpha-\beta-1}-1\)A(F^{\leftarrow}(p))+\frac{(\alpha-1)^{1/\alpha}\mu}{\alpha F^{\leftarrow}(p)}\)\(1+o(1)\)\).
$$
Using Proposition \ref{pro:7}, we obtain the second-order asymptotics of generalized-quantile-based systemic risk measures: $\psi_p(S_n)$, $\E\[S_n|S_n>\psi_p(S_n)\],\E[X_m|S_n>\psi_p(S_n)]$ and $\E[(X_m-\psi_p(X_m))_+|S_n>\psi_p(S_n)]$.
\begin{theorem}\label{the:7.1}
Let $X_1,\ldots,X_n$ be nonnegative random variables with a common marginal distribution $F$ satisfying $\overline{F}(\cdot)\in 2\RV_{-\alpha,\beta}$ with $\alpha>1$, $\beta\leq0$ and an auxiliary function $A(\cdot)$. Assume that ($X_1,\ldots,X_n$) follows an $n$-dimensional Sarmanov distribution given by \eqref{eq:sarmanov} and $\lim\limits_{x_i\rightarrow\infty}\phi_i(x_i)=d_i\in \R, \phi_i(\cdot)-d_i\in \RV_{\rho_i}$ with $\rho_i\leq 0$  for each $i=1,\ldots,n$. Let  $\Phi_1(t)=\Phi_2(t)=t^k$ for some $1 < k<\alpha$ and $k\in \N$. Then as $p\uparrow 1$, we get that
\begin{align*}
    \frac{\psi_p(S_n)}{F^{\leftarrow}(p)}&=n^{1/\alpha}L_{\alpha,k}\(1+\frac{n^{\beta/\alpha}\Delta_{\alpha,\beta,k}-1}{\alpha\beta}A(F^{\leftarrow}(p))(1+o(1))\right.\\
    &\quad\left.+\frac{(\alpha-k+1)\mu_n^*\(F^{\leftarrow}(p)\)+(k-1)n\mu}{\alpha n^{1/\alpha} L_{\alpha,k}F^{\leftarrow}(p)}(1+o(1))\)+o\(\sum\limits_{i=1}^n|\phi_i\(F^{\leftarrow}(p)\)-d_i|\),
\end{align*}
\begin{align*}
    \frac{\E\[S_n|S_n>\psi_p(S_n)\]}{F^{\leftarrow}(p)}&=\frac{\alpha n^{1/\alpha}L_{\alpha,k}}{\alpha-1}\(1+\frac{1}{\alpha}\(\frac{n^{\beta/\alpha}L_{\alpha,k}^{\beta}}{\alpha-\beta-1}+\frac{n^{\beta/\alpha}\Delta_{\alpha,\beta,k}-1}{\beta}\)A(F^{\leftarrow}(p))(1+o(1))\right.\\
    &\quad \left.+\frac{(\alpha-k)\mu_n^*\(F^{\leftarrow}(p)\)+(k-1)n\mu}{\alpha n^{1/\alpha} L_{\alpha,k}F^{\leftarrow}(p)}(1+o(1))\)+o\(\sum\limits_{i=1}^n|\phi_i\(F^{\leftarrow}(p)\)-d_i|\),
\end{align*}
 \begin{align*}
      \frac{\E[X_m|S_n>\psi_p(S_n)]}{F^{\leftarrow}(p)}&=\frac{\alpha n^{1/\alpha} L_{\alpha,k} }{(\alpha-1)n}\(1+\frac{1}{\alpha}\(\frac{n^{\beta/\alpha}L_{\alpha,k}^{\beta}}{\alpha-\beta-1}+\frac{n^{\beta/\alpha}\Delta_{\alpha,\beta,k}-1}{\beta}\)A(F^{\leftarrow}(p))(1+o(1))\right.\\
         &\quad\left.+\frac{(k-1)\(n\mu-\mu_n^*\(F^{\leftarrow}(p)\)\)}{\alpha n^{1/\alpha} L_{\alpha,k}F^{\leftarrow}(p)}(1+o(1))\)+o\(\sum\limits_{i=1}^n|\phi_i\(F^{\leftarrow}(p)\)-d_i|\),
  \end{align*}
  and
   \begin{align*}
        \frac{\E[(X_m-\psi_p(X_m))_+|S_n>\psi_p(S_n)]}{F^{\leftarrow}(p)}&=\frac{\E[X_m|S_n>\psi_p(S_n)]}{F^{\leftarrow}(p)}-\frac{\psi_p(X)}{n F^{\leftarrow}(p)}.
    \end{align*}
Obviously, the first-order asymptotics of $\psi_p(S_n), \E\[S_n|S_n>\psi_p(S_n)\]$, $\E[X_m|S_n>\psi_p(S_n)]$ and $\E[(X_m-\psi(X_m))_+|S_n>\psi_p(S_n)]$ are $n^{1/\alpha}L_{\alpha,k}F^{\leftarrow}(p)$, $\frac{\alpha n^{1/\alpha}L_{\alpha,k}}{\alpha-1}F^{\leftarrow}(p)$, $\frac{\alpha n^{1/\alpha} L_{\alpha,k} }{(\alpha-1)n}F^{\leftarrow}(p)$ and $\(\frac{\alpha n^{1/\alpha}}{\alpha-1}-1\)\frac{L_{\alpha,k}}{n}F^{\leftarrow}(p)$. 
\end{theorem}
\begin{proof}
    According to Proposition \ref{pro:7} and Corollary \ref{cor:3.1}, as $p\uparrow1$, we have
    \begin{align*}
       &\quad \frac{\psi_p(S_n)}{F^{\leftarrow}(p)}\\
       &=L_{\alpha,k}\frac{G^{\leftarrow}(p)}{F^{\leftarrow}(p)}\(1+\(\frac{\Delta_{\alpha,\lambda,k}-1}{\alpha\lambda}A_G^{n}(G^{\leftarrow}(p))+\frac{(k-1)\E[S_n]}{\alpha L_{\alpha,k}G^{\leftarrow}(p)}\)(1+o(1))\)\\
          &=L_{\alpha,k}n^{1/\alpha}\(1+\(\frac{\mu_n^*\(F^{\leftarrow}(p)\)}{n^{1/\alpha}F^{\leftarrow}(p)}+\frac{n^{\beta/\alpha}-1}{\alpha\beta}A\(F^{\leftarrow}(p)\)\)\(1+o(1)\)+o\(\sum\limits_{i=1}^n|\phi_i\(F^{\leftarrow}(p)\)-d_i|\)\)\\
        &\quad\cdot \(1+\(\frac{(\Delta_{\alpha,\beta,k}-1)n^{\beta/\alpha}}{\alpha\beta}A(F^{\leftarrow}(p))+\frac{(\Delta_{\alpha,-1,k}-1)\mu_n^*\(F^{\leftarrow}(p)\)}{n^{1/\alpha}F^{\leftarrow}(p)}
        +\frac{(k-1)n\mu}{\alpha n^{1/\alpha} L_{\alpha,k}F^{\leftarrow}(p)}\)(1+o(1))\)\\
        &=n^{1/\alpha}L_{\alpha,k}\(1+\(\frac{n^{\beta/\alpha}\Delta_{\alpha,\beta,k}-1}{\alpha\beta}A(F^{\leftarrow}(p))+\frac{(\alpha-k+1)\mu_n^*\(F^{\leftarrow}(p)\) +(k-1)n\mu}{\alpha n^{1/\alpha} L_{\alpha,k}F^{\leftarrow}(p)}\)(1+o(1))\)\\
        &\quad +o\(\sum\limits_{i=1}^n|\phi_i\(F^{\leftarrow}(p)\)-d_i|\).
    \end{align*}
Next, by Lemma \ref{lem:cr} in the Appendix and Corollary \ref{cor:3.1}, as $p\uparrow1$, it follows that
\begin{align*}
    &\quad\frac{\E\[S_n|S_n>\psi_p(S_n)\]}{F^{\leftarrow}(p)}\\
    &=\frac{\alpha}{\alpha-1}\(1+\frac{1}{\alpha(\alpha-\lambda-1)}A_G^{n}(\psi_p(S_n))(1+o(1))\)\frac{\psi_p(S_n)}{F^{\leftarrow}(p)}\\
    &=\frac{\alpha}{\alpha-1}\(1+\(\frac{n^{\beta/\alpha}L_{\alpha,k}^{\beta}}{\alpha(\alpha-\beta-1)}A(F^{\leftarrow}(p))-\frac{\mu_n^*\(F^{\leftarrow}(p)\)}{\alpha n^{1/\alpha}L_{\alpha,k}F^{\leftarrow}(p)}\)(1+o(1))\)\\
    &\quad \cdot n^{1/\alpha}L_{\alpha,k}\(1+\(\frac{n^{\beta/\alpha}\Delta_{\alpha,\beta,k}-1}{\alpha\beta}A(F^{\leftarrow}(p))+\frac{(\alpha-k+1)\mu_n^*\(F^{\leftarrow}(p)\)+(k-1)n\mu}{\alpha n^{1/\alpha} L_{\alpha,k}F^{\leftarrow}(p)}\)(1+o(1))\)\\
    &\quad +o\(\sum\limits_{i=1}^n|\phi_i\(F^{\leftarrow}(p)\)-d_i|\)\\
    &=\frac{\alpha n^{1/\alpha}L_{\alpha,k}}{\alpha-1}\(1+\frac{1}{\alpha}\(\frac{n^{\beta/\alpha}L_{\alpha,k}^{\beta}}{\alpha-\beta-1}+\frac{n^{\beta/\alpha}\Delta_{\alpha,\beta,k}-1}{\beta}\)A(F^{\leftarrow}(p))(1+o(1))\right.\\
    &\quad \left.+\frac{(\alpha-k)\mu_n^*\(F^{\leftarrow}(p)\)+(k-1)n\mu}{\alpha n^{1/\alpha} L_{\alpha,k}F^{\leftarrow}(p)}(1+o(1))\)+o\(\sum\limits_{i=1}^n|\phi_i\(F^{\leftarrow}(p)\)-d_i|\).
\end{align*}
    Furthermore, by Lemma \ref{lem:csum} in the Appendix and Theorem \ref{the:7.1}, as $p\uparrow1$, we get that
    \begin{align*}
         &\quad \frac{\E[X_m|S_n>\psi_p(S_n)]}{F^{\leftarrow}(p)}\\
         &=\frac{\alpha }{(\alpha-1)n}\(1+\widetilde{B}(\psi_p(S_n))(1+o(1))\)\frac{\psi_p(S_n)}{F^{\leftarrow}(p)}+o\(\sum\limits_{i=1}^n|\phi_i(\psi_p(S_n))-d_i|\)\\
         &=\frac{\alpha }{(\alpha-1)n}\(1+\(\frac{n^{\beta/\alpha}L_{\alpha,k}^{\beta}}{\alpha(\alpha-\beta-1)}A(F^{\leftarrow}(p))-\frac{\mu_n^*\(F^{\leftarrow}(p)\)}{n^{1/\alpha} L_{\alpha,k}F^{\leftarrow}(p)}\)(1+o(1))\)\\
         &\quad \cdot n^{1/\alpha}L_{\alpha,k}\(1+\(\frac{n^{\beta/\alpha}\Delta_{\alpha,\beta,k}-1}{\alpha\beta}A(F^{\leftarrow}(p))+\frac{(\alpha-k+1)\mu_n^*\(F^{\leftarrow}(p)\)+(k-1)n\mu}{\alpha n^{1/\alpha} L_{\alpha,k}F^{\leftarrow}(p)}\)(1+o(1))\)\\
         &\quad +o\(\sum\limits_{i=1}^n|\phi_i\(F^{\leftarrow}(p)\)-d_i|\)\\
         &=\frac{\alpha n^{1/\alpha} L_{\alpha,k} }{(\alpha-1)n}\(1+\frac{1}{\alpha}\(\frac{n^{\beta/\alpha}L_{\alpha,k}^{\beta}}{\alpha-\beta-1}+\frac{n^{\beta/\alpha}\Delta_{\alpha,\beta,k}-1}{\beta}\)A(F^{\leftarrow}(p))(1+o(1))\right.\\
         &\quad\left.+\frac{(k-1)\(n\mu-\mu_n^*\(F^{\leftarrow}(p)\)\)}{\alpha n^{1/\alpha} L_{\alpha,k}F^{\leftarrow}(p)}(1+o(1))\)+o\(\sum\limits_{i=1}^n|\phi_i\(F^{\leftarrow}(p)\)-d_i|\).
    \end{align*}
  Lastly, according to Lemma \ref{lem:cr} and Proposition \ref{pro:7}, as $p\uparrow 1$, it follows that
    \begin{align*}
        \frac{\E[(X_m-\psi_p(X_m))_+|S_n>\psi_p(S_n)]}{F^{\leftarrow}(p)}&=\frac{\E[X_m|S_n>\psi_p(S_n)]}{F^{\leftarrow}(p)}-\frac{\psi_p(X)}{n F^{\leftarrow}(p)}.
    \end{align*}
 Thus, this completes the proof of the Theorem.   
\end{proof}

This theorem further generalizes the results of Theorems \ref{the:4.1} and \ref{the:4.2}. In particular, the asymptotics in this theorem are reduced to 
$\e_p(S_n)$, $\CE_p(S_n)$, $\ICE_{p,m}(S_n)$ and $\SICE_{p,m}(S_n)$ if $k=2$.

\section{Conclusion}\label{sec:conc}

In this paper, we study systemic risk measures with a multivariate Sarmanov distribution. We first classify them into two families of $\VaR$- and expectile-based systemic risk measures. We have the second-order asymptotics of $\VaR$, $\CTE$, $\MES$ and $\SES$ in the first family. Furthermore, we obtain the second-order asymptotics of expectile, $\CE$, $\ICE$, and $\SICE$ in the second family. In addition, we give concrete numerical and analytical examples to illustrate the main results. We emphasize that the second-order asymptotics can provide a much better approximation as $p\uparrow 1$ than the first-order asymptotics. Moreover, we provide a comprehensive comparison among $\VaR$- and expectile-based systemic risk measures. We find that expectile-based measures often yield higher asymptotic accuracy than $\VaR$-based measures. 
Finally, we apply the asymptotic treatment to financial diversification and extend second-order asymptotic results to generalized-quantile-based systemic risk measures. We believe that our results stimulate future research in risk management, and our findings have implications for financial practitioners and regulators striving to better understand systemic risks in complex financial systems.


\setcounter{equation}{0}
\section{Appendix}\label{Appendix}
\subsection{Proofs of Sections \ref{VaR}-\ref{expectile}}
 In order to prove the Proposition \ref{cor:3.1},  we need to show that $\overline{\widetilde{F}}(\cdot)\in 2\RV$ with
 $\widetilde{F}(t)=\int_0^t \(1-\frac{\phi(x)}{c}\)\d F(x)$ under some weak conditions, which is stated in the following lemma.

\begin{lemma}\label{lem:2rv}
    Assume that $\overline{F}(\cdot)\in 2\RV_{-\alpha,\beta}$ with $\alpha>0,\beta\leq 0$ and an auxiliary function $A(\cdot)$. Define $\widetilde{F}(t)=\int_0^t \(1-\frac{\phi(x)}{c}\)\d F(x)$, where $\phi: \R \to \R$ is a function with $\E\left[\phi(X)\right]=0$, $|\phi(x)|\leq c-1$ for all $x\geq 0$, $\lim\limits_{t\rightarrow\infty}\phi(t)=b$ and $\phi(\cdot)-b\in \RV_{\rho}$ with $b\in \R$ and $\rho\leq 0$. We have $\overline{\widetilde{F}}(\cdot)\in 2\RV_{-\alpha,\gamma}$ with $\gamma= \max\{\beta,\rho\}$ and $\widetilde{A}(\cdot)=A(\cdot)-\frac{\rho\alpha}{(c-b)(\alpha-\rho)}\(\phi(\cdot)-b\)$. In addition, as $t\rightarrow\infty$, we have
\begin{align}\label{eq:8.1}
     \frac{\overline{\widetilde{F}}(t)}{\overline{F}(t)}=\(1-\frac{b}{c}\)\(1-\frac{\alpha}{(c-b)(\alpha-\rho)}\(\phi(t)-b\)\(1+o(1)\)\).
\end{align}
\end{lemma}
\begin{proof}
Firstly, we need to prove that $\overline{\widetilde{F}}(\cdot)\in \RV_{-\alpha}$. Because $\lim\limits_{t\rightarrow\infty}\phi(t)=b$, for any $\epsilon > 0$, there exists $t_0> 0$ such that
$$
b - \epsilon \leq \phi(t) \leq b + \epsilon, \;\;\;\; \forall ~t > t_0.
$$
Fix any $x > 0$. 
We have
\begin{align*}
\lim_{t\rightarrow\infty}\frac{\overline{\widetilde{F}}(tx)}{\overline{\widetilde{F}}(t)}=\lim_{t\rightarrow\infty}\frac{\(1-\frac{b}{c}\)\overline{F}(tx)-\int_{tx}^\infty \(\frac{\phi(y)-b}{c}\)\d F(y)}{\(1-\frac{b}{c}\)\overline{F}(t)-\int_t^\infty \(\frac{\phi(y)-b}{c}\)\d F(y)}=x^{-\alpha}.
\end{align*}
Secondly, because $\overline{F}(\cdot) \in \RV_{-\alpha}$, $\phi(\cdot)-b\in \RV_{\rho}$ and Potter's inequality (Proposition B.1.9 (5) of \cite{de2006extreme}), for any $\delta>0$, there exists $t_1> t_0$ such that for $t,ty>t_1$, 
\begin{align*}
    \left|\frac{\overline{F}(ty)}{\overline{F}(t)}-y^{-\alpha}\right|\leq \epsilon\max\{y^{-\alpha+\delta},y^{-\alpha-\delta}\}, 
\end{align*}
and 
\begin{align*}
    \left|\frac{\phi(ty)-b}{\phi(t)-b}-x^\rho\right|\leq \epsilon\max\{y^{\rho+\delta},y^{\rho-\delta}\}.
\end{align*}
 For any $t>t_1$, by the dominated convergence theorem, we have
\begin{align*}
\frac{\overline{\widetilde{F}}(t)}{\overline{F}(t)}&=\(1-\frac{b}{c}\)-\int_t^\infty\frac{\phi(x)-b}{c\overline{F}(t)}\d F(x)\nonumber\\
    &=\(1-\frac{b}{c}\)+\frac{\phi(t)-b}{c\overline{F}(t)}\int_1^\infty\(\frac{\phi(tx)-b}{\phi(t)-b}-x^\rho+x^\rho\)\d\overline{F}(tx)\nonumber\\
     &=\(1-\frac{b}{c}\)-\(1+\int_1^\infty\(\frac{\overline{F}(tx)}{\overline{F}(t)}-x^{-\alpha}+x^{-\alpha}\)\d x^\rho\)\frac{\phi(t)-b}{c}\(1+o(1)\)\nonumber\\
    &=\(1-\frac{b}{c}\)\(1-\frac{\alpha}{(c-b)(\alpha-\rho)}\(\phi(t)-b\)\(1+o(1)\)\).
\end{align*}
By $\overline{F}(\cdot)\in 2\RV_{-\alpha,\beta}$ and Drees-type inequality (e.g., \cite{mao2013secondorder}), there exists $t_2>t_1$ such that for all $y > 0$ and $t > \max\{t_2, \frac{t_2}{y} \}$,
\begin{align*}
   \left|\frac{1}{A(t)}\(\frac{\overline{F}(ty)}{\overline{F}(t)}-y^{-\alpha}\) -H_{-\alpha,\beta}(y)\right|\leq \epsilon y^{-\alpha+\rho}\max\{y^{\delta},y^{-\delta}\}.
\end{align*}
For any $t > \max\{t_2, \frac{t_2}{x} \}$,
according to the dominated convergence theorem, for all $x>0$, it follows that
\begin{align*}
    \frac{\overline{\widetilde{F}}(tx)}{\overline{\widetilde{F}}(t)}-x^{-\alpha}&=\frac{\(1-\frac{b}{c}\)\overline{F}(tx)\(1-\frac{\alpha}{(c-b)(\alpha-\rho)}\(\phi(tx)-b\)\(1+o(1)\)\)}{\(1-\frac{b}{c}\)\overline{F}(t)\(1-\frac{\alpha}{(c-b)(\alpha-\rho)}\(\phi(t)-b\)\(1+o(1)\)\)}-x^{-\alpha}\nonumber\\
    &=\(x^{-\alpha}+x^{-\alpha}\frac{x^{\beta}-1}{\beta}A(t)\(1+o(1)\)\)\(1-\frac{\alpha(x^\rho-1)}{(c-b)(\alpha-\rho)}\(\phi(t)-b\)\(1+o(1)\)\)-x^{-\alpha}\\
     &=x^{-\alpha}\frac{x^{\beta}-1}{\beta}A(t)\(1+o(1)\)-x^{-\alpha}\frac{x^{\rho}-1}{\rho}\frac{\rho\alpha}{(c-b)(\alpha-\rho)}\(\phi(t)-b\)\(1+o(1)\).
\end{align*}
Thus, this completes the proof of Lemma \ref{lem:2rv}.   
\end{proof}

 The next lemma extends Lemma 2.4 of \cite{mao2013second}.
\begin{lemma}\label{lem:beta}
Let $F$ be the distribution function of a nonnegative random variable satisfying $\overline{F}(\cdot)\in \RV_{-\alpha}$ with $\alpha>1$. For any fixed $z\in (0,1)$ and $\beta>0$, define
\begin{align*}
V_{\beta}(zt)=\int_0^{zt}\(\(1-\frac{y}{t}\)^{-\beta}-1\)\d F(y), ~~~t>0.
\end{align*}
Then, as $t\rightarrow\infty$, we have
\begin{align*}
V_{\beta}(zt)\sim\beta t^{-1}\mu(t).
\end{align*}
\end{lemma}

\begin{proof}
Since $\alpha>1$ and $\mu(t)\rightarrow\mu$ as $t\rightarrow\infty$, we have that $\mu(\cdot)<\infty$ and $\frac{\mu(\cdot)}{\cdot}\in \RV_{-1}$. We have
\begin{align*}\label{eq:lem1}
\mu(t)=\int_0^tx\d F(x)=-\int_0^tx\d \overline{F}(x)=-t\overline{F}(t)+\int_0^t\overline{F}(x)\d x.
\end{align*}
According to Karamata's theorem, it is easy to check that
\begin{align}
\mu(t)\sim \int_0^t\overline{F}(x)\d x ~~\mbox{as}~~ t\rightarrow\infty, ~~\mbox{and}~~~ \lim_{t\rightarrow\infty}\frac{t\overline{F}(t)}{\int_0^t\overline{F}(y)\d y}=0.
\end{align}
By the integration by parts, it follows that
\begin{align*}
V_{\beta}(zt)&=-\int_0^{zt}\(\(1-\frac{y}{t}\)^{-\beta}-1\)\d \overline{F}(y)=-(1-z)^{-\beta}\overline{F}(z t)+\frac{\beta}{t}\int_0^{z t}\overline{F}(y)\(1-\frac{y}{t}\)^{-\beta-1}\d y.
\end{align*}
For any fixed $z\in (0,1)$, by \eqref{eq:lem1},
\begin{align*}
\lim_{t\rightarrow\infty}\frac{t\overline{F}(zt)}{\int_0^{zt}\overline{F}(y)\(1-\frac{y}{t}\)^{-\beta-1}\d y}\leq\lim_{t\rightarrow\infty}\frac{t\overline{F}(z t)}{\int_0^{zt}\overline{F}(y)\d y}=0.
\end{align*}
Since \eqref{eq:lem1} holds for all $\alpha>1$ and  $1+(\beta+1)x\leq(1-x)^{-\beta-1}\leq 1+(\beta+1)(1-z)^{-\beta-2}x$ for $x\in (0,z)$, we obtain
\begin{align*}
\lim_{t\rightarrow\infty}\frac{tV_{\beta}(zt)}{\mu(zt)}&=\beta \lim_{t\rightarrow\infty}\frac{\int_0^{zt}\overline{F}(y)\(1-\frac{y}{t}\)^{-\beta-1}\d y}{\int_0^{zt}\overline{F}(y)\d y}\nonumber\\
&\geq\beta \lim_{t\rightarrow\infty}\frac{\int_0^{z t}\overline{F}(y)\d y+(\beta+1)\int_0^{z t}\overline{F}(y)y/t\d y}{\int_0^{z t}\overline{F}(y)\d y}\nonumber\\
&=\beta \(1+(\beta+1)\lim_{t\rightarrow\infty}\frac{zt\overline{F}(zt)}{\int_0^{zt}\overline{F}(y)\d y+t\overline{F}(zt)}\)\nonumber\\
&=\beta ,
\end{align*}
and
\begin{align*}
\lim_{t\rightarrow\infty}\frac{tV_{\alpha}(zt)}{\mu(zt)}\leq\beta \(1+(\beta+1)(1-z)^{-\beta-2}\lim_{t\rightarrow\infty}\frac{\int_0^{zt}\overline{F}(y)y\d y}{t\int_0^{zt}\overline{F}(y)\d y}\)=\beta .
\end{align*}
By $\frac{\mu(\cdot)}{\cdot}\in\RV_{-1}$, it follows that
$$V_{\beta}(zt)\sim \beta t^{-1}\mu(zt)\sim\beta t^{-1}\mu(t),~~\mbox{as} ~~t\rightarrow\infty.$$
This completes the proof.
\end{proof}

\begin{lemma}\label{lem:csum}
   Under the conditions of Theorem \ref{the:3.2}, as $t\rightarrow\infty$, we have that
\begin{align*}
    \E\left[X_m\big|S_n>t\right]=\frac{\alpha t}{(\alpha-1)n}\(1+\widetilde{B}(t)(1+o(1))\)+o\(\sum\limits_{i=1}^n|\phi_i(t)-d_i|\),
\end{align*}
where
\begin{align}\label{eq:B}
\widetilde{B}(t)=\frac{1}{\alpha(\alpha-\beta-1)}A(t)-\frac{\mu_n^*(t)}{t}.
\end{align} 
\end{lemma}
\begin{proof}
It follows that
\begin{align*}
\E\left[X_m\big|S_n>t\right]&=\int_0^\infty \frac{\p\left(S_n>t,X_m>z\right)}{\p(S_n>t)}\d z\\
&=t\(1+\frac{\p(X>t)}{\p(S_n>t)}\(\int_1^\infty\frac{\p(X_m>zt)}{\p(X>t)}\d z-\int_0^1\frac{\p\left(S_n>t,X_m\leq zt\right)}{\p(X>t)}\d z\)\)\\
&:=t\(1+\frac{\p(X>t)}{\p(S_n>t)}\(Q_1(t)-Q_2(t)\)\).
\end{align*}
For $Q_1(t)$, using the fact that $\overline{F}(\cdot)\in 2\RV_{-\alpha,\beta} $ with $\alpha>1$, $\beta\leq0$ and an auxiliary function $A(\cdot)$, as $t\rightarrow \infty$, we have
\begin{align*}
Q_1(t)&=\int_1^\infty z^{-\alpha}\(1+\frac{z^{\beta}-1}{\beta}A(t)(1+o(1))\)\d z=\frac{1}{\alpha-1}\(1+\frac{1}{\alpha-\beta-1}A(t)(1+o(1))\).
\end{align*}
For $t>0$ and $z\in(0,1)$, write $\Omega_{t,z}:=\{(x_1,\ldots,x_n)\in \R_+^n: \sum_{i=1}^nx_i>t, x_m\leq zt\}$.
For $Q_2(t)$, the key idea is to connect $\p\left(S_n>t,X_m\leq zt\right)$.  Similar to the proof of Proposition \ref{pro:sum}, we have that
\begin{align}\label{eq:Q_2}
 &\quad \p\left(S_n>t,X_m\leq zt\right)\nn\\
&=\int_{\Omega_{t,z}}\prod_{k=1}^{n}\d F(x_k)+\sum_{1\leq i< j\leq n}a_{ij}c_ic_j\(\int_{\Omega_{t,z}}\prod_{k=1}^{n}\d F(x_k)-\int_{\Omega_{t,z}}\prod_{k=1,k\neq i}^{n}\d F(x_k)\d\widetilde{F_i}(x_i)\right.\nonumber\\
&\quad\left.-\int_{\Omega_{t,z}}\prod_{k=1,k\neq j}^{n}\d F(x_k)\d\widetilde{F_j}( x_j)+\int_{\Omega_{t,z}}\prod_{k=1,k\neq i,j}^{n}\d F(x_k)\d\widetilde{F_i}(x_i)\d\widetilde{F_j}(x_j)\)\nonumber\\
&:=J(t,z)+\sum_{1\leq i< j\leq n}a_{i j}c_i c_j\(J(t,z)-K_i(t,z)-K_j(t,z)+K_{i,j}(t,z)\).
\end{align}
For simplicity, denote $S_n^{(m)}=\sum\limits_{i = 1, i\neq m}^nX_i^*$  which has the distribution $G_m$.
  By Theorem 3.5 of \cite{mao2013second} with $\alpha>1$ and the induction assumption, as $t\rightarrow \infty$, it is easy to check that
  \begin{align*}
\frac{\overline{G_m}(t)}{\overline{F}(t)}=(n-1)\(1+(n-2)\alpha t^{-1}\mu(t)\(1+o(1)\)\)+o\(A(t)\).
\end{align*}
 Clearly, $\overline{G_m}(\cdot)\in 2\RV_{-\alpha,\kappa}$ with an auxiliary function $B(\cdot)$. 
  Then, $B(\cdot)$ is given by
 \begin{align*}
B(t)=A(t)-(n-2)\alpha t^{-1}\mu(t).
\end{align*}
For $J(t,z)$, by Corollary 3.3 of \cite{mao2013secondorder} and the dominated convergence theorem, as $t\rightarrow \infty$, it follows that
 \begin{align*}
J(t,z)&=\int_0^{zt}\overline{G_m}(t-y)\d F(y)\nonumber\\
 &=\overline{G_m}(t)\int_0^{zt} \(\(1-\frac{y}{t}\)^{-\alpha}+H_{-\alpha,\lambda}\(1-\frac{y}{t}\)B(t)\(1+o(1)\) \) \d F(y)\nonumber\\
 &:=\overline{G_m}(t)\(J_1(t,z)+J_2(t,z)\).
\end{align*}
For $J_1(t,z)$, by Lemma \ref{lem:beta} and the fact that $\overline{F}(t)=o\(\frac{\mu(t)}{t}\)$ as $t\rightarrow \infty$, we have that
\begin{align*}
J_1(t,z)&=\int_0^{zt}\(\(1-\frac{y}{t}\)^{-\alpha}-1\)\d F(y)+F(zt)=\alpha t^{-1}\mu(t)+1-\overline{F}(zt)=1+\alpha t^{-1}\mu(t)\(1+o(1)\).
\end{align*}
For $J_2(t,z)$, since $H_{\alpha,\lambda}\(1-\frac{y}{t}\)\leq\frac{(1-z)^{-\alpha}}{|\lambda|}\(\(1-\frac{y}{t}\)^{\lambda}-1\)$ for any $y\in (0,zt)$ and $z\in (0,1)$, by Lemma 5.6 in \cite{barbe2005asymptotic}, we have
\begin{align*}
\int_0^{z t}H_{\alpha,\lambda}\(1-\frac{y}{t}\)\d F(y)\leq \int_0^{z t}\frac{(1-z)^{-\alpha}}{|\lambda|}\(\(1-\frac{y}{t}\)^{\lambda}-1\)\d F(y)=0.
\end{align*}
Then,
\begin{align*}
J_2(t,z)=o\(B(t)\).
\end{align*}
Thus,
\begin{align*}
\frac{J(t,z)}{\overline{F}(t)}&=\frac{\overline{G_m}(t)}{\overline{F}(t)}\Big(1+\alpha t^{-1}\mu(t)\(1+o(1)\)+o\(B(t)\)\Big)\\
&=(n-1)\Big(1+(n-1)\alpha t^{-1}\mu(t)\(1+o(1)\)+o\(|A(t)|\)\Big).
\end{align*}
 Similarly, it is easy to see that
\begin{align*}
\frac{K_i(t,z)}{\overline{F}(t)}&=
\(n-1-\frac{d_i}{c_i}\)\Big(1+\alpha t^{-1}(n-1)\mu(t)\(1+o(1)\)\Big)+o\(|A(t)|\)\\
&\quad-\(\frac{\alpha(n-1)\mu_i(t)}{c_it}+\frac{\alpha(\phi_i(t)-d_i)}{c_i(\alpha-\rho_i)}\)(1+o(1)),
\end{align*}
and
\begin{align*}
\frac{K_{i,m}(t,z)}{\overline{F}(t)}&=
\(n-1-\frac{d_i}{c_i}-\frac{d_m}{c_m}\)\Big(1+\alpha t^{-1}(n-1)\mu(t)\(1+o(1)\)\Big)+o\(|A(t)|\)\\
&\quad -\(\frac{\alpha(n-1)\mu_i(t)}{c_it}+\frac{\alpha(\phi_i(t)-d_i)}{c_i(\alpha-\rho_i)}+\frac{\alpha(n-1)\mu_m(t)}{c_mt}+\frac{\alpha(\phi_m(t)-d_m)}{c_m(\alpha-\rho_m)}\)(1+o(1)).
\end{align*}
Plugging $J(t,z)$, $K_i(t,z),K_m(t,z)$ and $K_{i,m}(t,z)$ into \eqref{eq:Q_2} yields that
\begin{align*}
 \frac{\p\left(S_n>t,X_m\leq zt\right)}{\overline{F}(t)}=\(n-1\)\Big(1+\alpha t^{-1}\mu_n^*(t)\(1+o(1)\)\Big)+o\(|A(t)|+\sum_{i=1}^n|\phi_i(t)-d_i|\).
\end{align*}
Thus, we have that
\begin{align*}
Q_2(t)=\int_0^1\frac{\p\left(S_n>t,X_m\leq zt\right)}{\overline{F}(t)}\d z=\(n-1\)\Big(1+\alpha t^{-1}\mu_n^*(t)\(1+o(1)\)\Big)+o\(|A(t)|+\sum_{i=1}^n|\phi_i(t)-d_i|\).
\end{align*}
According to the dominated convergence theorem, it follows that as $t\rightarrow\infty$,
\begin{align*}
\E\left[X_m\big|S_n>t\right]&=t\(1+\(n\(1+\widetilde{A}_n(t)(1+o(1))\)\)^{-1}\(Q_1(t)-Q_2(t)\)\)\nonumber\\
&=t\(1+\frac{1}{n}\(1-\alpha t^{-1}\mu_n^*(t)(1+o(1))+o\(A(t)+\sum_{i=1}^n(\phi_i(t)-d_i)\)\)\right.\\
&\quad\left.\cdot\(\frac{1}{\alpha-1}\(1+\frac{1}{\alpha-\beta-1}A(t)(1+o(1))\)-(n-1)\(1+\alpha t^{-1}\mu_n^*(t)+o(|A(t)|)\)\)\)\\
&=\frac{\alpha t}{(\alpha-1)n}\(1+\(\frac{1}{\alpha(\alpha-\beta-1)}A(t)-\frac{\mu_n^*(t)}{t}\)(1+o(1))\)+o\(\sum_{i=1}^n|\phi_i(t)-d_i|\).
\end{align*}
This completes the proof of the lemma.
\end{proof}

\begin{lemma}\label{lem:5.2}
Under the conditions of Theorem \ref{the:3.2}, there exists some sufficiently large $t_1$ such that for any $t\geq t_1$, 
\begin{align*}
    \E\left[\(X_m-t_1\)_+\Big|S_n>t\right]=\frac{\alpha t}{(\alpha-1)n}\(1+\widetilde{B}(t)(1+o(1))\)-\frac{t_1}{n}+o\(\sum\limits_{i=1}^n|\phi_i(t)-d_i|\),
\end{align*}
where  $\widetilde{B}(t)$ is defined in \eqref{eq:B}.
\end{lemma}
\begin{proof}
This proof proceeds similarly as that of Proposition \ref{pro:sum}. Applying the integration by parts, we conclude that
\begin{align*}
\E\left[\(X_m-t_1\)_+\Big|S_n>t\right]
&=t\int_{\frac{t_1}{t}}^\infty\frac{\p\(X_m>zt,S>t\)}{\p\(S>t\)}\d z\\
&=t\(1-\frac{t_1}{t}+\frac{\p(X>t)}{\p(S_n>t)}\(\int_1^\infty\frac{\p(X_m>zt)}{\p(X>t)}\d z-\int_{\frac{t_1}{t}}^1\frac{\p\left(S_n>t,X_m\leq z t\right)}{\p(X>t)}\d z\)\).
\end{align*}
Following a similar analysis of $Q_2(t)$ of Proposition \ref{pro:sum}, we have that
\begin{align*}
\int_{\frac{t_1}{t}}^1\frac{\p\left(S_n>t,X_m\leq zt\right)}{\p(X>t)}\d z=(1-t_1/t)\(n-1\)\(1+\alpha t^{-1}\mu_n^*(t)\)+o\(A(t)+\sum_{i=1}^n(\phi_i(t)-d_i)\).
\end{align*}
Thus, by  {Taylor's expansion}, as $t\rightarrow \infty$, we have
\begin{align*}
 &\quad\E\left[\(X_m-t_1\)_+\Big|S_n>t\right]\nn\\
&= t\(1-\frac{t_1}{t}+\(n\(1+\widetilde{A}_n(t)(1+o(1))\)\)^{-1}\(\frac{1}{\alpha-1}\(1+\frac{1}{\alpha-\beta-1}A(t)\)\right.\right.\\
&\quad\left.\left.-(1-t_1/t)\(n-1\)\(1+\alpha t^{-1}\mu_n^*(t)+o\(A(t)+\sum_{i=1}^n(\phi_i(t)-d_i)\)\)\)\)\\
&=\frac{\alpha t}{(\alpha-1)n}\(1+\(\frac{1}{\alpha(\alpha-\beta-1)}A(t)-\frac{\mu_n^*(t)}{t}\)(1+o(1))\)-\frac{t_1}{n}+o\(\sum_{i=1}^n(\phi_i(t)-d_i)\).
\end{align*}
We complete the proof of this lemma.
\end{proof}

\begin{lemma}\label{lem:cr}
Let $Y$ be the nonnegative random variable with a distribution $H$ satisfying that $\overline{H}(\cdot)\in 2\RV_{-\alpha, \rho}$
 with $\alpha>1$, $\rho\leq 0$ and an auxiliary function $A_H(\cdot)$. As $t\rightarrow\infty$, we have
 \begin{align*}
     \E\left[Y\big|Y>t\right]=\frac{\alpha t}{\alpha-1}\(1+\frac{1}{\alpha(\alpha-\rho-1)}A_H(t)(1+o(1))\).
 \end{align*}
\end{lemma}
\begin{proof}
    By the dominated convergence theorem ensured by Theorem 2.3.9 of \cite{de2006extreme}, it follows that, as $t\rightarrow \infty$,
\begin{align*}
\E\left[Y\big|Y>t\right]&=\int_0^\infty\frac{\p\left(Y>z,Y>t\right)}{\p(Y>t)}\d z\\
&=t\(1+\int_1^\infty z^{-\alpha}\(1+\frac{z^\rho-1}{\rho}A_H(t)(1+o(1))\)\d z\)\\
&=\frac{\alpha t}{\alpha-1}\(1+\frac{1}{\alpha(\alpha-\rho-1)}A_H(t)(1+o(1))\).
\end{align*}
Thus, we prove this lemma. 
\end{proof}

\noindent\textbf{Proof of Proposition \ref{pro:e}.}
Because of Equation \eqref{eq:1.1}, for large enough $p\uparrow1$ satisfying $\e_p(X)>0$, we have
\begin{align*}
1-\frac{\E(X)}{\e_p(X)}=\frac{2p-1}{1-p}\E\Bigg[\left(\frac{X}{\e_p(X)}-1\right)\id_{\{X/\e_p(X)\geq 1\}}\Bigg].
\end{align*}
 Applying the integration by parts, we have
\begin{align*}
\E\Bigg[\left(\frac{X}{\e_p(X)}-1\right)\id_{\{X/\e_p(X)\geq 1\}}\Bigg]
&=\overline{F}\(\e_p(X)\)\(\int_1^\infty x^{-\alpha}\d x+\int_1^\infty\(\frac{\overline{F}\(x\e_p(X)\)}{\overline{F}\(\e_p(X)\)}-x^{-\alpha}\)\d x\)\\
&=\overline{F}\(\e_p(X)\)\(\frac{1}{\alpha-1}+\int_1^\infty H_{-\alpha,\beta}(x)A(\e_p(X))\(1+o(1)\)\d x\)\\
&=\frac{\overline{F}\(\e_p(X)\)}{\alpha-1}\(1+\frac{1}{\alpha-\beta-1}A(\e_p(X))\(1+o(1)\)\),
\end{align*}
where the third step is due to the dominated convergence theorem ensured by Theorem 2.3.9 of \cite{de2006extreme}. In particular, \cite{bellini2014generalized} showed that
$$\e_p(X)\sim (\alpha-1)^{-1/\alpha}F^{\leftarrow}(p),~~~~~ p\uparrow1.$$
Since $\e_p(X)\rightarrow\infty$, $1-p\downarrow 0$ and $A(\e_p(X))\downarrow 0$ as $p\uparrow 1$,  by the first-order Taylor's expansion, we have that
\begin{align*}
\frac{1-p}{\overline{F}\(\e_p(X)\)}&=\frac{1}{\alpha-1}(1-2(1-p))\(1-\frac{\mu}{\e_p(X)}\)^{-1}\(1+\frac{1}{\alpha-\beta-1}A(\e_p(X))\(1+o(1)\)\)\\
&=\frac{1}{\alpha-1}(1-2(1-p))\(1+\frac{\mu}{\e_p(X)}\(1+o(1)\)\)\(1+\frac{1}{\alpha-\beta-1}A(\e_p(X))\(1+o(1)\)\)\\
&=\frac{1}{\alpha-1}\(1+\frac{(\alpha-1)^{-\beta/\alpha}}{\alpha-\beta-1}A(F^{\leftarrow}(p))\(1+o(1)\)+\frac{(\alpha-1)^{1/\alpha}\mu}{F^{\leftarrow}(p)}\(1+o(1)\)\),
\end{align*}
where in the third step we use $1-p\sim\overline{F}(F^{\leftarrow}(p))$ as $p\uparrow 1$. {Notably, in the second last step, we use $\lim\limits_{p\uparrow1}\overline{F}(F^{\leftarrow}(p))F^{\leftarrow}(p)=0$, and thus $\overline{F}(F^{\leftarrow}(p))=o(1/F^{\leftarrow}(p))$.} 
In addition, due to the fact that $U_F(\cdot)\in \RV_{1/\alpha,\beta/\alpha}$ with an auxiliary function $\alpha^{-2}A\circ U_F(\cdot)$, it follows that
\begin{align*}
\frac{\e_p(X)}{F^{\leftarrow}(p)}
&=\(\frac{1-p}{\overline{F}\(\e_p(X)\)}\)^{1/\alpha}\(1+\frac{\(\frac{1-p}{\overline{F}\(\e_p(X)\)}\)^{\beta/\alpha}-1}{\beta/\alpha}\alpha^{-2}A\circ U_F\(1/(1-p)\)\)\\
&=\(\alpha-1\)^{-1/\alpha}\(1+\frac{(\alpha-1)^{-\beta/\alpha}}{\alpha(\alpha-\beta-1)}A(F^{\leftarrow}(p))\(1+o(1)\)+\frac{(\alpha-1)^{1/\alpha}\mu}{\alpha F^{\leftarrow}(p)}\(1+o(1)\)\)\\
&\quad\cdot\(1+\frac{\(\alpha-1\)^{-\beta/\alpha}-1}{\alpha\beta}A(F^{\leftarrow}(p))\(1+o(1)\)\)\\
&=\(\alpha-1\)^{-1/\alpha}\(1+\(\frac{1}{\alpha\beta}\(\frac{(\alpha-1)^{1-\beta/\alpha}}{\alpha-\beta-1}-1\)A(F^{\leftarrow}(p))+\frac{(\alpha-1)^{1/\alpha}\mu}{\alpha F^{\leftarrow}(p)}\)\(1+o(1)\)\).
\end{align*}
Thus, we obtain the desired results.

\subsection{Proofs of Sections \ref{sec:app}-\ref{sec:7}}
\noindent\textbf{Proof of Theorem \ref{The:D}.}
By Theorem \ref{the:3.1}, \eqref{eq:D} and  {Taylor's expansion}, as $p\uparrow1$,  based on $\VaR$, we obtain
     \begin{align*}
     D_p^{\VaR}(S_n)
     &=1-\frac{\VaR_p(S_n)-n\mu}{nF^{\leftarrow}(p)}\(1-\frac{\mu}{F^{\leftarrow}(p)}\)^{-1}\\
     &=1-\frac{\VaR_p(S_n)-n\mu}{n F^{\leftarrow}(p)}\(1+\frac{\mu}{F^{\leftarrow}(p)}(1+o(1))\)\\
     &=1-n^{1/\alpha-1}\(1+\frac{n^{\beta/\alpha}-1}{\alpha\beta}A\(F^{\leftarrow}(p)\)\(1+o(1)\)\right.\\
     &\quad\left.+\frac{\mu_n^*\(F^{\leftarrow}(p)\)-(n-n^{1/\alpha})\mu}{n^{1/\alpha} F^{\leftarrow}(p)}\(1+o(1)\)\)+o\(\sum\limits_{i=1}^n(\phi_i\(F^{\leftarrow}(p)\)-d_i)\).
\end{align*}
According to  Proposition \ref{pro:e} and Theorem \ref{the:4.1}, as $p\uparrow 1$, we have
\begin{align*}
 D_p^{\e}(S_n)
&=1-\frac{\e_p(S_n)/n-\mu}{\e_p(X)-\mu}\\
 &=1-n^{1/\alpha-1}\(1+\frac{\(n^{\beta/\alpha}-1\)(\alpha-1)^{1-\beta/\alpha}}{\alpha\beta(\alpha-\beta-1)}A\(F^{\leftarrow}(p)\)\(1+o(1)\)\right.\\
 &\quad \left.+\frac{(\alpha-1)^{1/\alpha+1}\(\mu_n^*\(F^{\leftarrow}(p)\)-(n-n^{1/\alpha})\mu\)}{\alpha n^{1/\alpha} F^{\leftarrow}(p)}\(1+o(1)\)\)+o\(\sum\limits_{i=1}^n(\phi_i\(F^{\leftarrow}(p)\)-d_i)\).
\end{align*}   
Based on $\CTE$ and Lemma \ref{lem:cr}, as $p\uparrow1$, we have
 \begin{align*}
 D_p^{\CTE}(S_n)
 &=1-\frac{\CTE_p(S_n)-n\mu}{\frac{n\alpha F^{\leftarrow}(p) }{\alpha-1}\(1+\frac{1}{\alpha(\alpha-\beta-1)}A\(F^{\leftarrow}(p)\)-\frac{(\alpha-1)\mu}{\alpha F^{\leftarrow}(p)}\)}\\
   &=1-n^{1/\alpha-1}\(1+\(\zeta_{\alpha,\beta}^n-\frac{1}{\alpha(\alpha-\beta-1)}\)A\(F^{\leftarrow}(p)\)\(1+o(1)\)\)+o\(\sum\limits_{i=1}^n(\phi_i\(F^{\leftarrow}(p)\)-d_i)\)\\
   &\quad-\(\frac{(\alpha-1)\mu_n^*\(F^{\leftarrow}(p)-n\mu\)}{n \alpha F^{\leftarrow}(p)}+\frac{(\alpha-1) n^{1/\alpha-1} \mu}{\alpha F^{\leftarrow}(p)}\)\(1+o(1)\)\\
   &=1-n^{1/\alpha-1}\(1+\frac{1}{\alpha\beta}\(\frac{n^{\beta/\alpha}(\alpha-1)-\beta}{\alpha-\beta-1}-1\)A\(F^{\leftarrow}(p)\)\(1+o(1)\)\right.\\
   &\quad\left.+\frac{(\alpha-1)\(\mu_n^*\(F^{\leftarrow}(p)\)-\(n-n^{1/\alpha}\)\mu\)}{\alpha n^{1/\alpha} F^{\leftarrow}(p)}\(1+o(1)\)\)+o\(\sum\limits_{i=1}^n(\phi_i\(F^{\leftarrow}(p)\)-d_i)\).
\end{align*}
In addition, applying Lemma \ref{lem:cr}, it follows that 
\begin{align*}
 D_p^{\CE}(S_n)
 &=1-\frac{\CE_p(S_n)-n\mu}{\frac{n\alpha \e_p(X) }{\alpha-1}\(1+\frac{1}{\alpha(\alpha-\beta-1)}A\(\e_p(X)\)\)-n\mu}\\
  &=1-\frac{\CE_p(S_n)/n-\mu}{\frac{\alpha \e_p(X) }{\alpha-1}\(1+\frac{(\alpha-1)^{-\beta/\alpha}}{\alpha(\alpha-\beta-1)}A\(F^{\leftarrow}(p)\)(1+o(1))\)-\mu}\\
 &=1-n^{1/\alpha-1}\(1+\frac{\(n^{\beta/\alpha}-1\)(\alpha-1)^{-\beta/\alpha}\(\alpha+\beta-1\)}{\alpha\beta(\alpha-\beta-1)}A\(F^{\leftarrow}(p)\)\(1+o(1)\)\right.\\
 &\quad \left.+\frac{(\alpha-1)^{1/\alpha}(\alpha-2)\(\mu_n^*\(F^{\leftarrow}(p)\)-\(n-n^{1/\alpha}\)\mu\)}{\alpha n^{1/\alpha} F^{\leftarrow}(p) }\(1+o(1)\)\)+o\(\sum\limits_{i=1}^n(\phi_i\(F^{\leftarrow}(p)\)-d_i)\).
\end{align*}
Thus, this completes the proof of Theorem \ref{The:D}.

\noindent\textbf{Proof of Proposition \ref{pro:7}.}
According to Corollary 3 of \cite{bellini2014generalized}, we have
\begin{align*}
    p\E\[((X-\psi_p(X))_+)^{k-1}\]=(1-p)\E\[((X-\psi_p(X))_-)^{k-1}\].
\end{align*}
If $k-1$ is odd (i.e., $k=2,4,6,\dots$), using $S^{k-1} = (S_+)^{k-1} - (S_-)^{k-1}$, we have
\begin{align}\label{eq:1}
\E\[(\psi_p(X)-X)^{k-1}\]=\frac{2p-1}{1-p}\E\[((X-\psi_p(X))_+)^{k-1}\].
\end{align} 
If $k-1$ is even (i.e., $k=3,5,7,\dots$), using $S^{k-1} = (S_+)^{k-1} + (S_-)^{k-1}$, we have
\begin{align*}
    \E\[(\psi_p(X)-X)^{k-1}\]=\frac{1}{1-p}\E\[((X-\psi_p(X))_+)^{k-1}\].
\end{align*}
As $p\uparrow 1$ and $\alpha>1$, we have $1-p\sim \overline{F}(F^{\leftarrow}(p))=o\(1/F^{\leftarrow}(p)\)$ and $2p-1=1-2(1-p)=1-o\(1/F^{\leftarrow}(p)\)$, which implies $\frac{2p-1}{1-p}=\frac{1}{1-p}\Big(1-o\(1/F^{\leftarrow}(p)\)\Big).$ 
Hence, the above cases share asymptotically equivalent coefficients, and it suffices to prove our results under the case \eqref{eq:1} where $k-1$ is odd. First, we need to prove the first-order asymptotic. For the left-hand side of \eqref{eq:1}, applying the integration by parts, as $p\uparrow 1$, we have
 \begin{align*}
    \frac{\E\[((X-\psi_p(X))_+)^{k-1}\]}{\overline{F}(\psi_p(X))}&=\frac{1}{\overline{F}(\psi_p(X))}\int_{\psi_p(X)}^\infty(y-\psi_p(X))^{k-1}\d F(y)\nonumber\\
     &=(k-1)\psi_p(X)^{k-1}\int_1^\infty\overline{F}(\psi_p(X)y)(y-1)^{k-2}\d y\nonumber\\
     &\sim (k-1)\psi_p(X)^{k-1}\int_1^\infty y^{-\alpha}(y-1)^{k-2}\d y\nonumber\\
     &=(k-1)\psi_p(X)^{k-1}\int_0^1y^{\alpha-k}(1-y)^{k-2}\d y. 
 \end{align*}   
For the left-hand side of \eqref{eq:1}, as $p\uparrow 1$, we have 
\begin{align*}
    \E\[(\psi_p(X) -X)^{k-1}\]&= \psi_p(X) ^{k-1}\E\[\(1-\frac{X}{\psi_p(X) }\)^{k-1}\].
\end{align*}
Thus,
\begin{align*}
    \frac{1-p}{\overline{F}(\psi_p(X) )}
\sim \frac{(2p-1)(k-1)B_{\alpha,k}}{\E\[\(1-\frac{X}{\psi_p(X) }\)^{k-1}\]}.
\end{align*}
Clearly, as $p\uparrow 1$, $\psi_p(X)\rightarrow \infty$ and  $1-p\sim\overline{F}\(F^{\leftarrow}(p)\)$, we have
\begin{align}\label{eq:2}
   \psi_p(X)\sim L_{\alpha,k}F^{\leftarrow}(p).
\end{align}
Furthermore, we prove the second-order asymptotic. As $p\uparrow 1$, we have
\begin{align*}
&\quad\E\[(X-\psi_p(X) )_+^{k-1}\]\nn\\
&= (k-1)\psi_p(X)^{k-1}\overline{F}(\psi_p(X) )\int_1^\infty\frac{\overline{F}(\psi_p(X) y)}{\overline{F}(\psi_p(X) )}(y-1)^{k-2}\d y\nonumber\\
  &=(k-1)\psi_p(X) ^{k-1}\overline{F}(\psi_p(X) )\int_1^\infty y^{-\alpha}\(1+\frac{y^\beta-1}{\beta}A(\psi_p(X) )(1+o(1))\)(y-1)^{k-2}\d y\nonumber\\
  &=(k-1)\psi_p(X) ^{k-1}\overline{F}(\psi_p(X) )B_{\alpha,k}\(1+\frac{B_{\alpha-\beta,k}-B_{\alpha,k}}{\beta B_{\alpha,k}}A(\psi_p(X) )(1+o(1))\),
\end{align*}
where the second step is due to the dominated convergence theorem. For the left-hand side of \eqref{eq:1}, as $p\uparrow 1$, we have $1-p\sim\overline{F}\(F^{\leftarrow}(p)\)(1+o(1)) $ and
\begin{align}\label{eq:9.9}
    \E\[(\psi_p(X) -X)^{k-1}\]&= \psi_p(X) ^{k-1}\E\[\(1-\frac{X}{\psi_p(X) }\)^{k-1}\](1+o(1))\nonumber\\
    &= \psi_p(X) ^{k-1}\(1-(k-1)\frac{\mu}{\psi_p(X) }(1+o(1))\).
\end{align}
By Taylor's expansion and \eqref{eq:2}, it follows that
\begin{align*}
   &\quad \frac{1-p}{\overline{F}(\psi_p(X) )}\\
   &=(k-1)B_{\alpha,k}\Big(1-2(1-p)\Big)\(1+\frac{B_{\alpha-\beta,k}-B_{\alpha,k}}{\beta B_{\alpha,k}}A(\psi_p(X) )(1+o(1))\)\(1+(k-1)\frac{\mu}{\psi_p(X) }\)\nonumber\\
   &=(k-1)B_{\alpha,k}\(1+\(\frac{B_{\alpha-\beta,k}-B_{\alpha,k}}{\beta B_{\alpha,k}}L_{\alpha,k}^{\beta}A(F^{\leftarrow}(p))+\frac{(k-1)\mu}{L_{\alpha,k}F^{\leftarrow}(p)}\)(1+o(1))\).
\end{align*}
 In addition, due to the fact that $U_F(\cdot)\in2\RV_{1/\alpha,\beta/\alpha}$ with an auxiliary function $\alpha^{-2}A\circ U_F(\cdot)$, it follows that
\begin{align*}
    \frac{\psi_p(X) }{F^{\leftarrow}(p)}&=\(\frac{1-p}{\overline{F}(\psi_p(X) )}\)^{1/\alpha}\(1+\frac{\(\frac{1-p}{\overline{F}(\psi_p(X) )}\)^{\beta/\alpha}-1}{\alpha\beta}A\circ U_F(1/(1-p))\)\nonumber\\
    &=L_{\alpha,k}\(1+\(\frac{B_{\alpha-\beta,k}-B_{\alpha,k}}{\alpha\beta B_{\alpha,k}}L_{\alpha,k}^{\beta}A(F^{\leftarrow}(p))+\frac{(k-1)\mu}{\alpha L_{\alpha,k}F^{\leftarrow}(p)}\)(1+o(1))\)\nonumber\\
    &\quad\cdot\(1+\frac{L_{\alpha,k}^{\beta}-1}{\alpha\beta}A(F^{\leftarrow}(p))(1+o(1))\)\nonumber\\
    &=L_{\alpha,k}\(1+\(\frac{L_{\alpha,k}^{\beta}B_{\alpha-\beta,k}B_{\alpha,k}^{-1}-1}{\alpha\beta}A(F^{\leftarrow}(p))+\frac{(k-1)\mu}{\alpha L_{\alpha,k}F^{\leftarrow}(p)}\)(1+o(1))\).
\end{align*}
Thus, as $p\uparrow 1$, we have
\begin{align*}
   \psi_p(X) =L_{\alpha,k}F^{\leftarrow}(p)\(1+\(\frac{\Delta_{\alpha,\beta,k}-1}{\alpha\beta}A(F^{\leftarrow}(p))+\frac{(k-1)\mu}{\alpha L_{\alpha,k}F^{\leftarrow}(p)}\)(1+o(1))\).
\end{align*}
This completes the proof.

\subsubsection*{Acknowledgement} 
The authors would like to thank the editors and two anonymous referees for their insightful suggestions and valuable comments, which helped us greatly improve the paper. The authors are grateful to Ziyue Shi, Shijie Wang, Zhichen Wang, Fan Yang, Yang Yang and members of the research group on financial mathematics and risk management at The Chinese University of Hong Kong, Shenzhen for their valuable comments and useful conversations.
Y. Liu acknowledges financial support from the National Natural Science Foundation of China (Grant No. 12401624), Guangdong Science and Technology Program (Grant No. 2024QN11X076), Shenzhen Science and Technology Program (Grant No. RCBS20231211090814028, JCYJ20250604141203005, 2025TC0010) as well as The Chinese University of Hong Kong (Shenzhen) University Development Fund (Grant No. UDF01003336) and is partly supported by the Guangdong Provincial Key Laboratory of Mathematical Foundations for Artificial Intelligence (Grant No. 2023B1212010001). 
B. Geng acknowledges financial support from the research startup fund (Grant No. S020318033/015) at Anhui University and the Provincial Natural Science Research Project of Anhui Colleges (2024AH050037). Y. Zhao is supported by the Hickman Scholarship from the Society of Actuaries.



\small
\bibliographystyle{apalike}
\bibliography{reference}

\end{CJK}
\end{document}